\documentclass[a4paper,11pt]{article}
\usepackage{jheppub}

\usepackage{amsthm}
\newtheorem{theorem}{Theorem}
\newtheorem{lemma}[theorem]{Lemma}
\newtheorem{corollary}[theorem]{Corollary}
\theoremstyle{definition}
\newtheorem{definition}{Definition}[section]
\newtheorem{prop}[theorem]{Proposition}

\usepackage{enumerate}
\usepackage[utf8]{inputenc}
\usepackage{amssymb,amsmath,amsbsy,mathtools}
\usepackage{braket}
\usepackage{graphicx}
\usepackage{xcolor}
\usepackage{hyperref}
\usepackage{natbib}
\usepackage{bm}
\usepackage{soul}

\newcommand{\mA}{\mathcal{A}}
\newcommand{\mB}{\mathcal{B}}

\newcommand{\mF}{\mathcal{F}}
\newcommand{\mH}{\mathcal{H}}

\newcommand{\mK}{\mathcal{K}}
\newcommand{\mL}{\mathcal{L}}
\newcommand{\mM}{\mathcal{M}}
\newcommand{\mN}{\mathcal{N}}

\newcommand{\mR}{\mathcal{R}}
\newcommand{\mS}{\mathcal{S}}
\newcommand{\mT}{\mathcal{T}}

\newcommand{\mbC}{\mathbb{C}}

\newcommand{\mbI}{\mathbb{I}}

\newcommand{\mbR}{\mathbb{R}}

\newcommand{\mbZ}{\mathbb{Z}}

\newcommand{\lb}{\left(}
\newcommand{\rb}{\right)}
\newcommand{\nn}{\nonumber}

\newcommand{\p}{\partial}
\newcommand{\M}{\mathcal{M}}

\newcommand{\sgn}{\text{sgn}}

\begin{document}

\makeatletter
\gdef\@fpheader{}
\makeatother

\title{\Large{Modular Intersections, Time Interval Algebras and Emergent AdS$_2$}}
\author[a,1]{Nima Lashkari,}
\author[a,2]{Kwing Lam Leung,}
\author[a,b,3]{Mudassir Moosa,}
\author[a,c,4]{Shoy Ouseph}
\affiliation[a]{Department of Physics and Astronomy, Purdue University, West Lafayette, IN 47907, USA}
\affiliation[b]{Department of Physics and Center for Theory of Quantum Matter, University of Colorado Boulder, Boulder, CO 80309, USA}
\affiliation[c]{New York University Abu Dhabi, P.O. Box 129188, Abu Dhabi, United Arab Emirates}

\emailAdd{$^1$nima@purdue.edu}
\emailAdd{$^2$leung60@purdue.edu}
\emailAdd{$^4$shoy.o@nyu.edu}
\emailAdd{$^3$mudassir.moosa@colorado.edu}

\abstract{We compute the modular flow and conjugation of time interval algebras of conformal Generalized Free Fields (GFF) in $0+1$-dimensions in the vacuum. For non-integer scaling dimensions, for general time intervals, the modular conjugation and the modular flow of operators outside the algebra are non-geometric. This is because they involve a Generalized Hilbert Transform (GHT) that treats positive and negative frequency modes differently. However, the modular conjugation and flows viewed in the dual bulk AdS$_2$ are local, because the GHT geometrizes as the local antipodal symmetry transformation that pushes operators behind the Poincar\'e horizon. These algebras of conformal GFF satisfy a {\it Twisted Modular Inclusion} and a {\it Twisted Modular Intersection} property. We prove the converse statement that the existence of a (twisted) modular inclusion/intersection in any quantum system implies a representation of the (universal cover of) conformal group $PSL(2,\mathbb{R})$, respectively. We discuss the implications of our result for the emergence of AdS$_2$ geometries in large $N$ theories without a large gap. Our result applied to higher dimensional eternal AdS black holes explains the emergence of two copies of $PSL(2,\mathbb{R})$ on future and past Killing horizons.}

\maketitle


\section{Introduction}

In holography, in the limit of small gravitational constant ($G_N$), the vacuum and thermofield double states of a large class of quantum systems with $N^2\sim 1/G_N$ local degrees of freedom are universally dual to quantum fields living on Anti-de Sitter space (AdS) and eternal black holes, respectively \cite{maldacena2003eternal}. 
In the strict limit $N=\infty$, on the boundary, we have Generalized Free Fields (GFF), and in the bulk, we have massive free fields living on the classical geometry with no back-reaction. The choice of the spectral densities of the GFF commonly discussed in holography corresponds to the large $N$ limit of gauge theories with large t'Hooft coupling. An important question is: What are the most general spectral densities consistent with the emergence of a dual bulk? Since such geometries do not necessarily originate from a strongly coupled gauge theory with large $N$ and a large gap in the spectrum of conformal primaries, it is natural to interpret them as {\it Stringy spacetimes}. Recent developments in operator algebraic approaches to quantum gravity, and Quantum Field Theory (QFT) have provided new tools to address this question.

A key recent development was the realization that in a large $N$ theory, in the strict $N\to \infty$ limit, the set of single-trace operators can form nontrivial nets of von Neumann algebras associated with finite and half-infinite time bands \cite{leutheusser2023emergent,leutheusser2023causal,witten2022gravity}. In particular, focusing on the spherically symmetric mode, or on a $0+1$-dimensional boundary theory, for specific choices of spectral density, there are type III$_1$ von Neumann algebras associated with time intervals.\footnote{In \cite{Furuya:2023fei}, it was shown that if the spectral density is Lebesgue-measurable the resulting von Neumann algebra is type III$_1$.} When it exists, the algebra of a half-infinite time-band from some time $t$ until eternity is an example of a {\it future subalgebra}. For any dynamical system, if the set of all operators from some time until eternity forms a proper subalgebra of observables we call that subalgebra a future subalgebra with respect to that time evolution.\footnote{In this work, we are primarily concerned with modular time evolution. At the risk of confusing the reader, in the remainder of this work, we will refer to modular future algebras, simply as future algebras.} The connections between the spectral density of GFF, future algebras, and the emergence of the horizon and its entropy have been further explored in  \cite{leutheusser2022subalgebra,chandrasekaran2023large,ouseph2024local,gesteau2023large,gesteau2024emergent,gesteau2024explicit}.
 
It was pointed out in \cite{ouseph2024local,gesteau2024emergent,Gesteau:2024rpt} that the ergodic properties of the modular flow of future algebras of GFF are tied to the spectral density. Modular flows that are chaotic enough (quantum Kolmogorov systems) result in the emergence of a {\it Stringy horizon} and generate an algebra of null translation and dilatation on it. So far, progress on deducing the existence of horizon from the universal ergodic properties of the future algebras, for the most part, have relied on a fundamental theorem in operator algebraic ergodic theory called {\it Half-Sided Modular Inclusion} (HSMI) theorem \cite{borchers2000revolutionizing,araki2005extension}; see Appendix \ref{App:Hsmi} for the statement of the theorem.\footnote{An intimately tied theorem is the so-called Half-Sided Translation theorem (HST); see \cite{borchers2000revolutionizing,ouseph2024local} and Appendix \ref{App:Hsmi}.} This theorem reconstructs the Lie algebra of null translations and null scaling on the future/past Killing horizons from the modular operators of the algebra and its subalgebra.\footnote{To simplify notation, we denote the future (past) HSMI with HSMI$+$ (HSMI$-$).}

In this work, we contribute to the connection between boundary time interval algebras and the emergence of geometry both at the technical and conceptual levels. At the technical level, we prove the following two sets of new results:
\begin{enumerate}
    \item {\bf Modular data of time interval algebras of GFF:} In conformal $0+1$-dimensional GFF i.e., $\rho(\omega)=\omega^{2\Delta-1}$, we explicitly compute the modular data (modular flow and modular conjugation) of time interval algebras in Section \ref{sec:modular-flow}. We find that for non-integer $\Delta$, the modular conjugation involves a nonlocal boundary transform we call the Generalized Hilbert Transform (GHT). The modular flow of operators is local until their bulk dual crosses the horizon, at which point, the flow becomes non-local from the boundary point of view (see Lemma \ref{lemma-interval-map}).\footnote{An interesting observation is that, the nonlocality emerges only when the dual bulk operator crosses the horizon. We will see in the example of future algebras that the modular flow of operators that always remain behind the horizon is also local.} However, viewed in terms of the bulk algebras, there are new ``local operators" and ``local algebras" behind the horizon, and the modular flow remains local for all bulk operators (see Lemma \ref{lemma-geometric-bulk-action}). 

    \item {\bf From algebras to conformal group:} We identify two algebraic requirements in terms of modular future algebras and time interval algebras that are each sufficient to imply the emergence of a representation of the universal cover the $0+1$-dimensional conformal group $PSL(2,\mathbb{R})$ by proving a {\it twisted inclusion theorem} (see Theorem \ref{twistModularInclusion}) and a {\it twisted modular intersection theorem} (see Theorem \ref{twistedModInt}).
\end{enumerate}
At the conceptual level, we make progress in the following two directions: 
\begin{enumerate}
    \item {\bf Bulk as a local representation of the boundary algebras:} In holography, the bulk dual of every boundary time interval algebra $\mA_{I}$ corresponds to the causal development of the time interval algebra $I$; see Figure \ref{fig:intro} (a). In a duality, there is only one algebra $\mA_I$ and one Hilbert space. The bulk or boundary fields correspond to different fields (frames) we use to generate the same algebra. The advantage of the language of operator algebras is that it is frame-independent. However, as we establish in this work, the notion of locality and Haag's duality are frame-dependent. On the boundary, we have a {\it net} of time interval von Neumann algebras. However, the bulk theory has a larger net because it has ``local" subalgebras corresponding to the regions that extend behind the horizon; see Figure \ref{fig:intro} (a). They do not correspond to any time interval algebras on the boundary, allowing for the possibility of geometrization of non-geometric boundary inclusions.

    \begin{figure}[t]
   \centering
    \includegraphics[width=\textwidth]{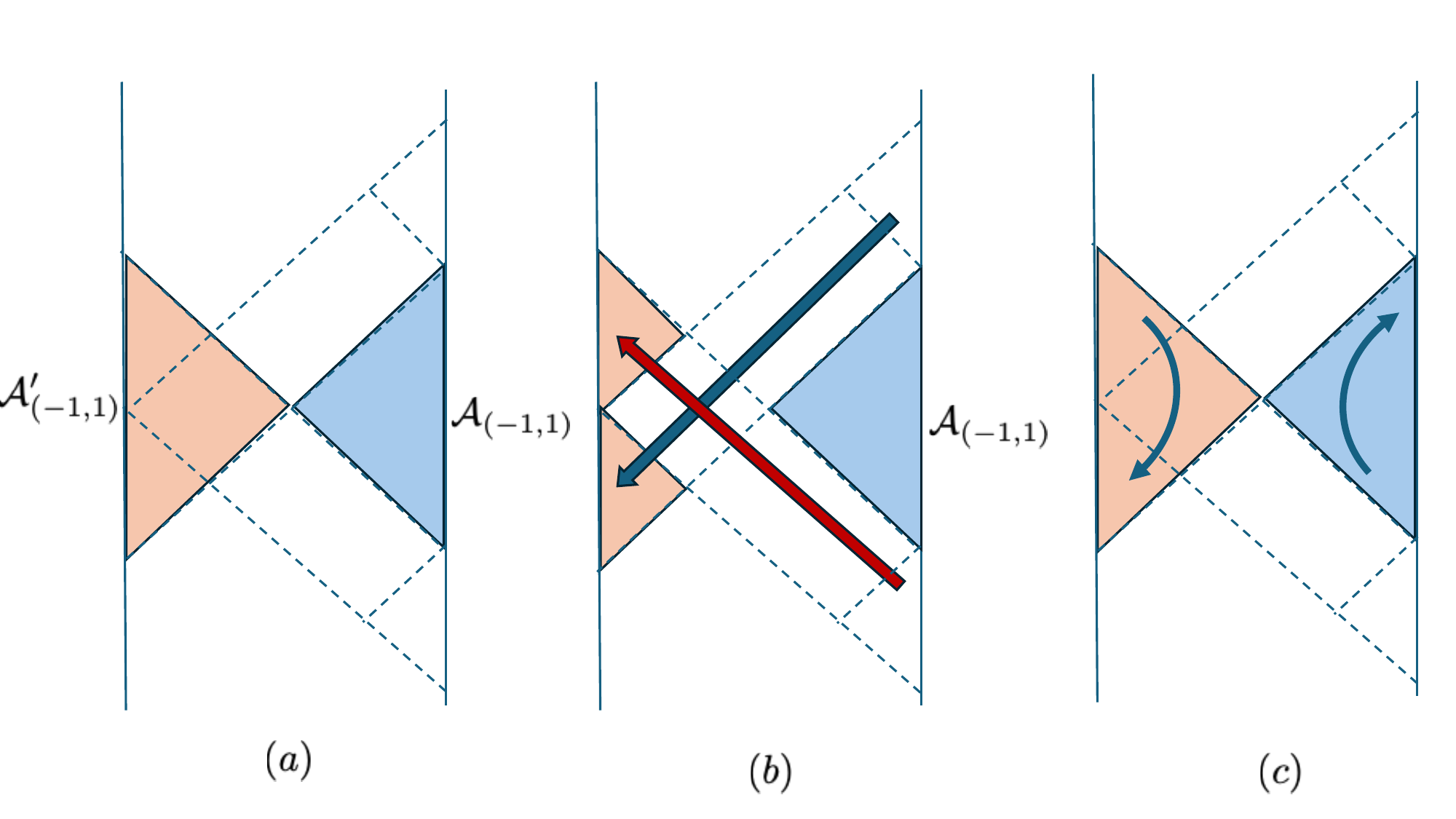}
    \caption{\small Bulk dual of time interval algebras in AdS$_2$. (a) The algebra of massive QFT in the blue region of AdS$_2$ is dual to the boundary GFF algebra of the time interval $(-1,1)$. However, its commutant $(\mA_{(-1,1)})'$, a local bulk algebra, does not correspond to any time interval algebra on the boundary. (b) The algebra $(\mA_{(-1,1)})'$ can be written in terms of the boundary local algebras using the antipodal transform (the red arrow) and its inverse (the blue arrow). The boundary dual of the antipodal transform is the Generalized Hilbert transform (GHT). GHT is a nonlocal transform on the boundary because it treats the negative and positive frequency modes differently. (c) Modular flow of time interval algebras pushes operators behind the horizon. This corresponds to a nonlocal transform on the boundary.}
    \label{fig:intro}
\end{figure}
    
    There are ``local" bulk transformations that are non-local on the boundary. An important example of such transformation in this work is the antipodal transformation in global AdS$_2$ $(\tau,\rho)\to (\tau+\pi,-\rho)$ which is local in the bulk, but is dual to the Generalized Hilbert Transform (GHT) (see Lemma \ref{lamma-anitpodal}); see Figure \ref{fig:intro} (b). From the boundary viewpoint, as the operator crosses the horizon, the transformation develops nonlocal features, however, viewed in the bulk, the spacetime continues behind the horizon and the modular flow remains local for all modular time; see Figure \ref{fig:intro} (c). 
    
    The algebras $\mA_I$ viewed from the boundary GFF fields is neither M\"obius covariant nor local, and violates Haag's duality, whereas, the same algebra viewed from the bulk massive free fields perspective is both M\"obius covariant and local, and satisfies Haag's duality.\footnote{See Section \ref{sec:gff_conformal} for the definition of a local net that satisfies Haag's duality.} We turn this into the following slogan: {\it The bulk emerges from insisting on a local representation of time interval algebras that satisfy Haag's duality.}

\item {\bf AdS$_2$ from Twisted Inclusions and Twisted Modular Intersections:}
In \cite{ouseph2024local}, we showed that for any system with modular future and past algebras, there exists a {\it local Poincar\' e algebra} in the vicinity of the emergent stringy bifurcate Killing horizon.
That is to say if $\mA^+\subset \mathcal{R}$ is HSMI$+$ and $\mA^-\subset\mathcal{R}$ is HSMI$-$ 
then the positive operators 
\begin{eqnarray}
    \pm G_\pm(s)=\mathcal{K}_\mR-\mathcal{K}_{\mA^\pm(s)}\,
\end{eqnarray}
with $\mK_\mA$ the modular Hamiltonian of $\mA$, generate the Lie algebra
\begin{eqnarray}\label{curvaturenearKilling}
&&[G_\pm(s),\mathcal{K}_\mR]=\pm i G_\pm(s)\nn\\
  &&  [G_+(-s),G_-(s)]=e^{-2\pi s}[G_+(0),G_-(0)]\ .
\end{eqnarray}
The large $s$ limit corresponds to zooming near an emergent bifurcate Killing horizon. A scaling limit of the right-hand-side of (\ref{curvaturenearKilling}) will capture the near-horizon curvature \cite{ouseph2024local,de2020holographic,faulkner2019modular}.
In the vacuum AdS$_2$ or any other geometry that is locally AdS$_2$ such as the near horizon geometry of extremal black holes, we have a universal emergent $\mathfrak{psl}(2,\mathbb{R})$ algebra \cite{lin2019symmetries}:
\begin{eqnarray}
    &&[G_\pm,\mK_\mR]=\pm i G_\pm(s)\nn\\
    &&[G_+,G_-]=-2i\mK_\mR\ .
\end{eqnarray}
One of the main results of this paper is that in Theorems \ref{twistedModInt} and \ref{twistModularInclusion}, we identify algebraic requirements in terms of modular future algebras and time interval algebras that are sufficient to imply the emergence of the $\mathfrak{psl}(2,\mathbb{R})$ Lie algebra. Then, insisting on a local representation of this symmetry group implies the emergence of AdS$_2$. In other words, AdS$_2$ emerges from insisting on a local description of the algebra that the modular Hamiltonians of time intervals satisfy.\footnote{For other attempts at reconstruction of bulk geometries using modular flows or in the context of GFF see \cite{faulkner2017bulk,faulkner2019modular,Nebabu:2023iox}.} Starting from the abstract algebraic requirements in Theorems \ref{twistedModInt} and \ref{twistModularInclusion} we explicitly construct a local net of algebras corresponding to QFT living in this emergent AdS$_2$.

\end{enumerate}

The outline of the paper is as follows: First, in Section \ref{sec:gff}, we review GFF and their time interval von Neumann algebras. This Section is included to set the notation, and basic definitions and keep the presentation self-contained.\footnote{The readers familiar with the quantization of GFF can skip Section \ref{sec:gff}.} Then, in Section \ref{sec:gff_conformal} we specialize to conformal GFF in $0+1$-dimensions, and establish that the commutant of time interval algebras for non-integer $\Delta$ is non-local because it involves the GHT transform in (\ref{GHT}) that treats the positive and negative frequencies differently. In Section \ref{sec:modular-flow}, we explicitly derive the modular flow and conjugation of the future algebra and time interval algebras in conformal GFF. We discuss the bulk dual AdS$_2$ in Section \ref{sec:bulk-dual}, we show that the bulk dual of the GHT is the local antipodal transform that sends operators behind the Poincar\'e horizon. We establish that the boundary modular flows viewed in terms of the bulk fields become local. In Section \ref{sec:Twisted}, we identify the {\it twisted inclusion} property and the {\it twisted modular intersection} property, each as a sufficient condition for the emergence of a representation of the universal cover of $PSL(2,\mathbb{R})$, which henceforth we will denote by $\widetilde{PSL}(2,\mbR)$. 
Starting from abstract algebras that satisfy the conditions of twisted modular intersection and/or twister inclusion theorems we explicitly construct a {\it local} and $\widetilde{PSL}(2,\mbR)$ covariant net of algebras in an emergent global AdS$_{2}$ spacetime.
Finally, in Section \ref{Sec:Stringy spacetime} we discuss how our results fit without the ergodic hierarchy of von Neumann algebras, and the emergence of AdS$_2$. We end with a brief discussion of generalization to higher dimensions and black hole spacetimes.

\section{$0+1$-dimensional GFF}\label{sec:gff}

A $0+1$-dimensional theory of Generalized Free Fields (GFF) is a collection of harmonic oscillators of frequencies $\omega$ such that
\begin{eqnarray}\label{CCR}
&&[a_\omega,a_{\omega'}^\dagger]=\delta(\omega-\omega')\ .
\end{eqnarray} 
In general, we consider a continuum of frequencies with density $\rho(\omega)$ for every positive frequency $\omega$. Equivalently, we can  describe GFF in terms of a collective field $\varphi_\rho$:
\begin{eqnarray}
        &&\varphi_\rho(t)=\int_0^\infty d\omega\, \sqrt{\rho(\omega)}(e^{-i\omega t}a_\omega+e^{i\omega t}a_\omega^\dagger)\ .
\end{eqnarray}
Defining $a_{-\omega}=a_\omega^\dagger$ and $ \varphi_\rho(\omega)=\sqrt{2\pi \rho(|\omega|)} a_{\omega}$ we write 
\begin{eqnarray}
        \varphi_\rho(t)&&=  \frac{1}{\sqrt{2\pi}}\int_{-\infty}^\infty d\omega\, e^{-i\omega t}\varphi_\rho(\omega)\ .
       \end{eqnarray}
The commutator of the collective fields is
\begin{eqnarray}\label{Gcrho}
    &&G_{c,\rho}(t,t')=[\varphi_\rho(t),\varphi_\rho(t')]=\int_0^\infty d\omega \rho(\omega)\lb e^{-i\omega(t-t')}-e^{i\omega(t-t')}\rb\nn\\
    &&G_{c,\rho}(\omega)=\sqrt{2\pi}\rho(|\omega|)\text{sgn}(\omega)\ .
\end{eqnarray}
Note that in our notation, $\rho(\omega)$ is defined only for positive $\omega$, whereas $G_c(-\omega)=-G_c(\omega)$ runs over all frequencies, and $G_{c,\rho}(t)$ is purely imaginary.\footnote{Our choice for the inverse Fourier transform is $ f(t) = \frac{1}{\sqrt{2\pi}} \int_{-\infty}^\infty d\omega\, f(\omega) e^{-i\omega t}$.}

A single harmonic oscillator corresponds to the density with positive frequency $m>0$:
\begin{eqnarray}
&&\rho_m(\omega)= \delta(\omega^2-m^2)\Theta(\omega)=\frac{1}{2m}\delta(\omega-m)\nn\\
&&G_c^{(m)}(t)=\frac{1}{2m}\lb e^{-imt}-e^{imt}\rb\ .
\end{eqnarray}
The generalization to higher dimensions is straightforward, as follows \cite{greenberg1961generalized,jost1967general,dutsch2003generalized}: 
A theory of GFF in $d$-dimensional spacetime is defined in terms of the collective field \footnote{We are using the metric signature $(-1,+1,\cdots, +1)$ and that is why in our inverse Fourier transform we have $e^{-i\omega t}$.}
\begin{eqnarray}\label{GFFd+1}
\varphi_\rho(t,x)&&=\int_{V_+}d^{d}k \sqrt{\rho(k^2)} e^{i k_\mu x^\mu}a_k   
\end{eqnarray}
where $V_+$ is the open forward lightcone in momentum space. A single scalar mode of mass $m$ corresponds to $\rho_m(k^2)=\delta(-k_\mu k^\mu+m^2)\Theta(k^0)$.\footnote{Note that a $0+1$-dimensional theory of GFF can come from the zero-mode of a higher dimensional theory in $\mathbb{R}^{d-1,1}$, or $l=0$ mode on $S^{d-1}\times \mathbb{R}$.}

\subsection{Weyl C$^*$-algebras: choice of commutator}\label{sec-gff-comm-choice}

To avoid the singular $t\to t'$ behavior of the GFF commutator in (\ref{Gcrho}) we smooth out the collective field with real functions of time 
\begin{eqnarray}
f\to \varphi(f)=\int_{-\infty}^\infty dt \: f(t) \varphi_\rho(t)= \sqrt{2\pi}\int_0^\infty d\omega \sqrt{\rho(\omega)} (f(-\omega)a_\omega+f(\omega) a_\omega^\dagger)\ . \label{eq:field}
\end{eqnarray}
The commutator of the smeared fields are
\begin{eqnarray}\label{commutatorfg}
[\varphi(f),\varphi(h)]&&= 2\pi\int_0^\infty d\omega \rho(\omega)(f(-\omega)h(\omega)-f(\omega)h(-\omega))\nn\\
&&=2\pi\int_0^\infty d\omega\, [f(\omega),h(\omega)]_\rho
\end{eqnarray}
where
\begin{eqnarray}\label{Poissonbraket}
&&    [f(\omega),h(\omega)]_\rho = \rho(\omega)
   F(\omega)^\dagger \begin{pmatrix}
       1 &0\\0& -1
   \end{pmatrix}H(\omega)\ ,\qquad F(\omega)=\begin{pmatrix}
f_+(\omega)\\
f_-(-\omega)
\end{pmatrix},
\end{eqnarray}
and $f_\pm(\omega)=f(\omega)\Theta(\pm\omega)$ are the positive and negative frequencies parts of the function, respectively.\footnote{Note that for real functions $f(t)$ the positive and negative frequency parts are complex conjugates of each other $f(-\omega)=f(\omega)^*$. In real-time, we have 
\begin{eqnarray}
f_\pm(t)=\mathcal{F}^{-1}(f(\omega)\Theta(\pm \omega))=\pm \frac{i}{2\pi}\int_{-\infty}^\infty dt'\:\frac{f(t')}{(t'-t)\pm i\epsilon}\ .
\end{eqnarray} The anti-symmetric bilinear $[\cdot, \cdot]_\rho$ in (\ref{Poissonbraket}) is a classical Poisson bracket on the space of functions $f$. We chose this notation to avoid confusion between the standard notation $\{\cdot,\cdot\}$ for the Poisson bracket and anti-commutator.} At this point, we make two important comments:
\begin{enumerate}
\item {\bf One-particle Hilbert space:} We have achieved our goal of regulating the singularity in $t\to t'$ of (\ref{Gcrho}) only if we choose functions $f$ and $g$ such that the integral in (\ref{commutatorfg}) is finite. To ensure this, we choose $f$ to be square-integrable with respect to the measure $d\mu=d\omega\rho(|\omega|)$. Mathematically, we write $f\in L^2(d\mu,\mathbb{R}_t)$. For example, in the special case of a GFF theory with $\rho(\omega)=\omega^{2\Delta-1}$ with $\Delta>1/2$, we can ensure this by choosing functions that fall off faster than $|\omega|^{-\Delta}$ at large frequencies $|\omega|$. There can be functions in the $L^2$ space that commute with all other functions in this space. In this case, we quotient the $L^2$ space by these functions to obtain a Hilbert space that we refer to as the {\it one-particle} Hilbert space.

\item {\bf Complex-time domain:} It is convenient to define projections to the positive frequency $P_+=\Theta(\omega)$ and negative frequencies $P_-=\Theta(-\omega)$ and the functions $f_\rho(\omega)=\sqrt{\rho(|\omega|)}f(\omega)$ so that 
\begin{eqnarray}
[f(\omega),h(\omega)]_\rho= F_\rho(\omega)^
\dagger
\begin{pmatrix}
    1 & 0\\
    0 & -1
\end{pmatrix}
H_\rho(\omega),\qquad F_{\rho}(\omega)=\begin{pmatrix}
(P_+f_\rho)(\omega)\\
(P_- f_\rho)(-\omega)
\end{pmatrix}
\end{eqnarray}
because $(P_+f_\rho)^*(\omega)=(P_-f_\rho)(-\omega)$. In the time domain, $P_+(f_\rho)$ ($P_-(f_\rho)$) is analytic in the lower (upper) complex $t$-plane.    
\end{enumerate}
A Weyl operator can be defined through its series expansion
\begin{eqnarray}\label{eq-weyl-def}
    W(f)  = \sum_{n=0}^\infty \frac{i^n \lb \varphi(f)\rb^n}{n!}= e^{i\varphi(f)} 
\end{eqnarray}
or the spectral decomposition of $\varphi(f)$. Since the commutator of GFF is central (proportional to the identity operator), using the Baker-Campbell-Hausdorff (BCH) expansion we find \footnote{When $[X,[X,Y]] = 0 = [Y,[X,Y]]$, the BCH expansion gives $e^X e^Y = e^{X+Y+\frac{1}{2}[X,Y]}$.}
\begin{eqnarray}\label{weyl-algebra}
    W(f)W(h)  &=&   e^{-\frac{1}{2}[\varphi(f), \varphi(h) ] } W(f+h)  \ .
\end{eqnarray}
With the multiplication rule above the Weyl operators form a {\it Weyl group}. Taking the closure of the formal sum of Weyl operators in operator norm topology gives a {\it Weyl C$^*$-algebra} which we denote by $\mB$.\footnote{In our notation, we denote C$^*$-algebras by $\mathcal{B}$ and von Neumann algebras by $\mathcal{A}$.} Weyl operators $W(f)$, $W(h)$ commute if and only if $[\varphi(f), \varphi(h) ] = in2\pi, n\in\mathbb{Z}$. However, a Weyl operator $W(f)$ belongs to the center if and only if for all $W(h)$ the smeared fields $\varphi(f)$ and $\varphi(h)$ commute. Since we ensured that none of our $\varphi(f)$ commutes with all $\varphi(h)$, the resulting Weyl algebra is a factor i.e., it has a trivial center.

\subsection{Quasi-free states}

In physics, we would like a Hilbert space representation of the GFF Weyl C$^*$-algebra. To do this, we need to choose a state. As before, we build intuition using the example of a single harmonic oscillator. In a harmonic oscillator, the Hamiltonian is proportional to the number operator, therefore it is natural to consider the following two choices: pure vacuum state $\ket{0}\in \mH$ and the thermal state $e^{-\beta H}/Z$. The vacuum is defined to be the unique state killed by the annihilation operator $a\ket{0}=0$, whereas the canonically purified thermal state i.e., the thermofield double $\ket{TFD}\in \mH_L\otimes \mH_R$, is defined by the detailed balance condition
\footnote{Detailed balance is the assumption that in equilibrium $p_n P_{n\to n+1}= p_{n+1} P_{n+1\to n}$ where $p_i$ are the equilibrium distribution of microstates $i$ and $P_{i\to j}$ is the probability of transition from microstate $i$ to microstate $j$. In equilibrium $p_n=e^{-\beta \omega (n+1/2)}/Z$, the annihilation operator ``$a$" represents the transition $n+1\to n$, and the creation operator ``$a^\dagger$" the reverse transition $n\to n+1$.}
\begin{eqnarray}\label{detailedbalance}
&&e^{\beta \omega/4} a^{(L)}\ket{TFD}=e^{-\beta \omega/4}(a^{(R)})^\dagger\ket{TFD}\\
&&\ket{TFD}=\frac{1}{\sqrt{Z}}\sum_n e^{-\beta \omega (n+1/2)/2}\ket{n,n}\ .
\end{eqnarray}
Vacuum is the pure state we find in the zero temperature limit $\lim_{\beta\to \infty}\ket{TFD}=\ket{0,0}$. 
The thermal two-point function as a function of time is 
\begin{eqnarray}
G(t;\beta)=\frac{\cos(\omega(t+i\beta/2))}{2\omega\sinh(\beta \omega/2)}\ .
\end{eqnarray}
Splitting the two-point function into its symmetric and anti-symmetric part, for positive $\omega$ we define \footnote{Note that since the Fourier transform of $G(-t)$ is $G(-\omega)$. The symmetric/anti-symmetric parts of the correlator in the time and frequency domain are Fourier transforms of each other.}
\begin{eqnarray}
&&G_{sym}(\omega)= \frac{1}{2}(G(\omega)+G(-\omega)) = \frac{1}{2}(b_+^2(\omega)+b_-^2(\omega))G_c(\omega) \nonumber\\
&&G_{anti}(\omega)= \frac{1}{2}(G(\omega)-G(-\omega)) = \frac{1}{2}G_c(\omega)\nonumber\\
&&G_c(\omega) = \sqrt{2\pi}\rho(\omega)\nn\\
&&b_\pm(\omega) =\lb \frac{e^{\pm\beta \omega/2}}{e^{\beta \omega/2}-e^{-\beta \omega/2}}\rb^{1/2}=\lb Z e^{\pm \beta\omega/2}\rb^{1/2}\ .
\end{eqnarray}
Since the commutator is twice the anti-symmetric part of the two-point function, it follows from (\ref{Gcrho}) that the thermal state can be alternatively defined using the {\it KMS relation}
\begin{eqnarray}\label{KMS}
G(\omega;\beta)=\frac{1}{1-e^{-\beta \omega}}G_c(\omega)= b_+^2(\omega)G_c(\omega)\ .
\end{eqnarray}
We call the thermal state and the vacuum of a harmonic oscillator {\it quasi-free states}. The intuition behind the terminology is as follows: any density matrix is a thermal state with respect to its modular Hamiltonian, and we can always define the raising and lowering operators on the Hilbert space. The quasi-free states are special in that their modular Hamiltonians are proportional to the number operator $N=a^\dagger a$.\footnote{In a general state  $\ket{\rho^{1/2}}=\sum_n \sqrt{p_n}\ket{n,n}$ the dependence of $p_n$ on $n$ need not be of the exponential form $p_n \sim e^{-\alpha n}$.}
In a quasi-free state, we have 
\begin{eqnarray}\label{stateWeylop}
\braket{ W(f)}&&=e^{-\frac{1}{2}\braket{f,f}_\rho}\nn\\
\braket{f,h}_\rho &&\equiv\braket{\varphi(f)\varphi(h)}=  \sqrt{2\pi}\int_{-\infty}^\infty d\omega \: G(\omega;\beta)f(-\omega) h(\omega)\nonumber\\
&& =  \sqrt{2\pi} \int_0^\infty d\omega \: G_c(\omega) \lb b_+^2 f(-\omega) h(\omega) + b_-^2 f(\omega) h(-\omega)\rb \ .
\end{eqnarray}
Note that for real scalar fields with translation invariance, the two-point function satisfies $G(t)^*=G(-t)$ which implies $G(\omega)$ is real ensuring $\braket{f,f}_\rho\geq 0$ for real functions $f(t)$. More generally, we can separate the real and imaginary parts 
\begin{eqnarray}
  &&  \Re\braket{f,h}_\rho=\int_{-\infty}^\infty d\omega\, G(\omega;\beta)\Re(f(-\omega)h(\omega))=\int_{-\infty}^\infty d\omega\, G_{sym}(\omega;\beta)f(-\omega)h(\omega)\nn\\
&&    \Im\braket{f,h}_\rho=\int_{-\infty}^\infty d\omega\, G(\omega;\beta)\Im(f(-\omega)h(\omega))=\int_{-\infty}^\infty d\omega\, G_{asym}(\omega;\beta)f(-\omega)h(\omega)\ .
\end{eqnarray}
In other words, we have 
\begin{eqnarray}
    &&\Re\braket{f,h}_\rho=\frac{1}{2}\braket{\{\varphi(f),\varphi(h)\}}\equiv G_{sym}(f,h)\nn\\
    &&\Im\braket{f,h}_\rho =\frac{1}{2} \braket{[\varphi(f),\varphi(h)]}\equiv G_{anti}(f,h)\ .
\end{eqnarray}
The real part of $\braket{f,h}_\rho$ is a symmetric bilinear on the real vector space of real functions and its imaginary part is an anti-symmetric bilinear on this real vector space. In the notation of (\ref{Poissonbraket}) we can write
\begin{eqnarray}\label{innerproduct}
&&\braket{f,h}_\rho \equiv 2\pi \int_0^\infty d\omega \braket{f(\omega),h(\omega)}_\rho \nn \\
&&\braket{f(\omega),h(\omega)}_\rho =  F_\rho^\dagger(\omega)\begin{pmatrix}
    b_+^2(\omega) & 0\\
    0 & b_-^2(\omega)
\end{pmatrix}H_\rho(\omega)\ .
\end{eqnarray}
According to the Stone-von Neumann theorem, all representations of the canonical commutation relations are unitarily equivalent. In particular, the linear Bogoliubov transformations below correspond to the unitary transformation that maps the vacuum Hilbert space $\ket{0,0}$ and the $\ket{TFD}$ Hilbert space:
\begin{eqnarray}\label{linearbog}
&&c^{(R)}=b_+a^{(R)}-b_- (a^{(L)})^\dagger\nn\\
&&c^{(L)}=-b_- (a^{(R)})^\dagger- b_+ a^{(L)}
\end{eqnarray}
resulting in the defining equations
\begin{eqnarray}\label{newannihilation}
c^{(R)}\ket{TFD}=c^{(L)}\ket{TFD}=0\ .
\end{eqnarray}
These equations follow from (\ref{detailedbalance}).

In a theory of GFF, we have a collection of harmonic oscillators. A general quasi-free state $\ket{QF}$  corresponds to a state in which each oscillator $\omega$ is in a thermal state of a different temperature $\beta(\omega)$. Defining $\gamma(\omega)\equiv e^{-\beta(\omega)\omega}$ we have
\begin{eqnarray}
&&a^{(L)}\ket{QF}=\gamma(\omega)^{1/2}(a^{(R)})^\dagger\ket{QF}\nn\\
&&\ket{QF}=\int_0^\infty \frac{d\omega}{\sqrt{Z(\omega)}}\sum_n \gamma(\omega)^{(n+1/2)/2}\ket{n(\omega),n(\omega)}\nn\\
&&G_{QF}(\omega)=b_+^2(\omega)G_c(\omega)\nn\\
&&b_\pm(\omega)=\lb \frac{\gamma^{\mp 1/2}(\omega)}{\gamma^{- 1/2}(\omega)-\gamma^{+ 1/2}(\omega)} \rb^{1/2} \ .
\end{eqnarray}
The von Neumann algebra is the double commutant of the Weyl algebra generated by $W(f)$ in a quasi-free state: $\mA=\mB''$.
A GFF theory with a finite number of oscillators $\rho(\omega)$ has support on a finite point number of points, and by Stone-von Neumann theorem, the Hilbert spaces corresponding to all quasi-free states are unitarily equivalent. We only have one von Neumann algebra which is the type I algebra $B(\mH)$. The situation is more interesting when the spectral density $\rho(\omega)$ is supported on an infinite number of points. Then, depending on the choice of state, we might obtain inequivalent von Neumann algebras. The vacuum state can be defined as the vector killed by all $a(f)=\int_{-\infty}^\infty dt\, a(t) f(t)$. More generally, a quasi-free state is defined by the expectation value of Weyl operators in (\ref{stateWeylop}). In other words, the choice of a quasi-free state and its corresponding Hilbert space is fixed by a choice of spectral density $\rho(\omega)$ and functions $\gamma(\omega)$ (or equivalently $b_\pm(\omega)$).

\subsection{Time interval von Neumann algebras}\label{time-intervalvNalgebras}

\begin{definition}[Time interval algebras]
    Let $I$ be a simply connected open interval on the real line $I\subset \mbR$. We define the Weyl algebra $\mB_I$ on the time interval as the unital polynomial algebra generated by the Weyl operators $W(f)$ for functions $f$ inside the one-particle Hilbert space with support only on the interval $I$ 
    \begin{eqnarray}
        \mB_I = \{W(f)|\text{supp}(f) \subseteq I \subset\mbR , f\in L^2(d\mu,\mathbb{R}_t)\} \ .
    \end{eqnarray}
Given a quasi-free state (as we saw in the previous section) we obtain a Hilbert space representation, and a von Neumann algebra $\mA_I=(\mB_I)''$ associated with $I$.
\end{definition}
See \cite{haag2012local,longoNotes,hollands2018entanglement} for more detail. The choice of spectral density $\rho(\omega)$ defines the one-particle Hilbert space, and a quasi-free state defines a Hilbert space representation, and hence the von Neumann algebras. We have a map $I\to \mA_I$ from time intervals to von Neumann algebras. It is clear that in any theory of GFF the map respects the inclusion relations 
\begin{eqnarray}
    I_1\subset I_2:\qquad \mA_{I_1}\subseteq\mA_{I_2}\ .
\end{eqnarray}
We are interested in two particular types of von Neumann subalgebras:
\begin{enumerate}
    \item {\bf Future/Past von Neumann algebras:} The von Neumann algebra $\mA_{\mbR^+}$ ($\mA_{\mbR^-})$ on the half-infinite positive (negative) real line $\mbR^+$ ($\mbR^-$), when it is a proper subalgebra of $\mA_\mathbb{R}$, respectively.
    \item {\bf Time interval von Neumann algebras:} The Weyl algebra $\mA_I$ corresponding to a finite time interval $I=(t_1,t_2)\subset \mbR$ when it is a proper subalgebra of $\mA_{\mathbb{R}}$.
\end{enumerate}
The time interval Weyl algebras $\mA_I$ are well-defined von Neumann algebras for any choice of spectral density $\rho(\omega)$, including the case of a single harmonic oscillator $\rho(\omega)=\frac{1}{2m}\delta(\omega-m)$ for some $m>0$. However, in that case, the inclusions are trivial because $\mA_I=\mA_{\mathbb{R}}$ for all intervals $I\subset \mathbb{R}$. We do not have time intervals or future algebras in this case.

In systems with well-posed equations of motion, such as a single harmonic oscillator or a quantum field of definite mass and spin, the knowledge of the system for an arbitrarily small time interval is enough to determine the entire evolution of the system. This principle is sometimes referred to as the \textit{time-slice axiom}. Therefore, in a system with the time-slice axiom, we expect $\mA_I=\mA_{\mathbb{R}}$ for any $I$. There are neither future nor time interval von Neumann algebras.

In GFF, with an arbitrary spectral density, the time-slice axiom fails, and we can have non-trivial inclusions of time interval algebras. In the absence of the time-slice axiom, the question arises for how long should we observe the system so that we can predict the entire evolution of the system forever. In terms of time interval algebras, this is the question what is the smallest time $T_{dep}$ such that $\mA_{(-T_{dep},T_{dep})}=\mA_\mathbb{R}$. 
Following \cite{Gesteau:2024rpt}, we call this time $T_{dep}$ the {\it causal depth parameter}.

In systems with a finite causal depth parameter, the map $I\to \mA_I$ is not injective. A system has finite causal depth $T_{dep}$ if $\mA_{(-t_0,t_0)}=\mA_{(-\infty,\infty)}$ for all $t_0>T_{dep}$. It is clear that if there exists operators in $\mA_\mathbb{R}$ that are not in $\mA_I$, then $\mA_I$ is a proper subalgebra of $\mA_\mathbb{R}$. However, importantly the converse is not true. In other words, if the relative commutant of $\mA_I\subset \mA_{\mathbb{R}}$ is non-trivial, it is a proper inclusion, however, there are proper inclusions with trivial relative commutant (singular inclusions) \cite{longo1984solution}.\footnote{The {\it relative commutant} of an inclusion of von Neumann algebras $\mA_1\subset \mA_2$ is the algebra of operators in $\mB$ that are not in $\mA$ i.e., $\mA_2 \cap \mA_1'$.} 
Singular inclusions can appear in interacting quantum field theories but not for GFF (see \cite{Gesteau:2024rpt} for the review of a proof).\footnote{A physical example of a singular inclusion is the inclusion of the wedge in an interacting QFT in its null-shifted wedge. In a genuinely interacting QFT in higher than $1+1$-dimensions we cannot localize any operators on a null plane \cite{bousso2015entropy}. As we will see in Section \ref{sec:bulk-dual}, the inclusion of $\mA_{\mathbb{R}_+}\in \mA_{\mathbb{R}}$ in GFF is dual to the algebra of a null interval times sphere in the dual geometry. However, since the bulk fields dual to GFF are free fields, there are operators localized on null interval (in higher dimensions null interval times spheres), hence, the inclusion is not singular.}

In \cite{Gesteau:2024rpt} the authors explored the connections between the causal depth parameter and the spectral density of GFF in in a KMS state. Here, we highlight the following three results from \cite{Gesteau:2024rpt}:
\begin{enumerate}
\item A system has causal depth parameter $T_{dep}$ if the spectral density measure $d\omega \rho(|\omega|)$ is exponential type $T_{dep}$. 
\begin{definition}
In harmonic analysis, a measure $\mu(\omega)$ is exponential type $T$ if $T$ is the infimum for which the exponentials $e^{i\omega t}$ with $t\in [-T,T]$ form a complete basis for $L^2(\mu)$ \cite{poltoratski2013problem}.
\end{definition}

\item For a system with a spectral density that is compactly supported in frequency space we say $T_{dep}=0$. 

\item We can have almost periodic two-point functions and still have $T_{dep}=\infty$. An explicit example is the case of ``primon gas" or ``Riemannium" where the spectral density is discrete as in (\ref{diracdelta}) with $\omega_n$ at the logarithm of prime numbers. 

\end{enumerate}

Following \cite{Gesteau:2024rpt} we say there is {\it stringy connectivity} if $T_{dep}=\infty$. The example of primon gas establishes that we can have infinite exponential type, stringy connectivity but no information loss.\footnote{For a recent discussion on conformal primon gas related to dynamical gravity, see \cite{Hartnoll:2025hly}} To obtain a chaotic GFF theory, it is natural to choose the spectral density to be infinite exponential type so that we have stringy connectivity. We postpone the discussion of ergodic classes of spectral density $\rho(\omega)$ for which the future and time interval algebras are proper subalgebras of $\mA_{\mathbb{R}}$ to Section \ref{Sec:Stringy spacetime}.

\section{Conformal GFF}\label{sec:gff_conformal}

We start this section with a discussion of higher dimensional GFF before specializing to $0+1$-dimensions.
We saw in Section \ref{sec:gff} that the $d$-dimensional GFF collective fields are defined using (\ref{GFFd+1}) with spectral density $\rho(k^2)$. The one-particle Hilbert space is $L^2(\rho(|k^2|)d^{d}k,\mathbb{R}^{d-1,1})$ which is equipped with a unitary positive-energy representation of the Poincar\'e group + scaling with the generators \cite{dutsch2003generalized}
\begin{eqnarray}ƒ
&&P_\mu=\int_{V_+} d^{d}k \: k_\mu a^\dagger(k) a(k)\nn\\
&&M_{\mu\nu}=i\int_{V_+}d^{d}k\: a^\dagger(k)\lb k_\nu \frac{\p}{\p k^\mu}-k_\mu \frac{\p}{\p k^\nu}\rb a(k)\nn\\
&&D=\frac{i}{2}\int_{V_+}d^{d}k\:a^\dagger(k) \lb (k\cdot \p_k)+(\p_k\cdot k)\rb a(k)
\end{eqnarray}
where 
\begin{eqnarray}
&&e^{i\lambda D}\varphi_\rho(x) e^{-i \lambda D}=\varphi_{\rho_\lambda}(\lambda x)\nn\\
&&\rho_\lambda(k^2)=\lambda^{d}\rho(\lambda^2 k^2)\ .
\end{eqnarray}
If the spectral density has homogeneous weight $k^{\Delta-d/2}$, under scaling, the collective GFF field transforms locally
\begin{eqnarray}
e^{i\lambda D}\varphi_\rho(x) e^{-i \lambda D}=\lambda^\Delta \varphi_{\rho}(\lambda x)\ .
\end{eqnarray}
Therefore, we call a GFF theory with spectral density $\rho(k^2)=k^{2\Delta-d}$ a {\it conformal GFF} in $d$-spacetime dimensions.

The vacuum two-point function of conformal GFF transforms covariantly under the Poincar\'e + scaling transformations. In other words, these transformations act as point transformations on $\mathbb{R}^{d-1,1}$ and the generators are independent of $\Delta$. In addition to the Poincar\'e + scaling, the conformal group also includes special conformal transformations. Unlike Poincar\'e + scaling generators, special conformal transformations do not preserve Minkowski space $\mathbb{R}^{d-1,1}$ because finite transformations can send points to infinity \cite{mack1969finite}.\footnote{In mathematical terms, the conformal group does not admit a transitive global action on Minkowski space. It acts quasi-globally \cite{brunetti1993modular}.} Enlarging the spacetime manifold by including the points at infinity (its conformal compactification) does not help because it results in closed time-like curves. One has to formulate conformal field theory in the Lorentzian cylinder $\mathbb{R}_T\times S^{d-1}$ \cite{luescher1975global}.\footnote{In mathematical terms, the conformal group acts transitively on the orientable universal cover of Minkowski space \cite{brunetti1993modular}. In $d>2$ the orientable universal cover of Minkowski space is infinite sheeted and diffeomorphic to $\mathbb{R}_t\times S^{d-1}$. In $d=1$ it is $\mathbb{R}$, and in $d=2$ it is $\mathbb{R}^{1,1}$. However, as argued by \cite{brunetti1993modular} the physically relevant manifold to consider is $\mathbb{R}_t\times S^{d-1}$ for all $d$.} The universal cover of the conformal group $\widetilde{SO}(d,2)$ acts on the Lorentzian cylinder. 

In conformal GFF, the generator of special conformal transformations depends on $\Delta$ :
\begin{eqnarray}
K^{(\Delta)}=\int_{V_+}a^\dagger(k)\lb \frac{\p}{\p k^\alpha}k_\mu \frac{\p}{\p k_\alpha}-(k\cdot \p_k)\frac{\p}{\p k^\mu}-\frac{\p}{\p k^\mu}(\p_k\cdot k)+(\frac{d}{2}-\Delta)^2 \frac{k^\mu}{k^2}\rb a(k)\ .
\end{eqnarray}
They do not act as point transformations, except in the special case $\Delta=d/2$. We will see this explicitly below in the example of $0+1$-dimensional GFF.

\subsection{Representation of $PSL(2,\mbR)$ and its universal cover $\widetilde{PSL}(2,\mbR)$}

\label{sec-Mobius}

In $0+1$-dimensions, conformal GFF has spectral density $\rho(\omega)=\omega^{2\Delta-1}$. The analog of Poincar\'e + scaling is the group $t\to a t+b$ generated $H$ (generator of translations) and $D$ (generator of scaling). In addition, the conformal group includes special conformal transformations $t\to t/(ct+1)$ and the full $0+1$-dimensional restricted conformal group is $SO^+(1,2)\simeq PSL(2,\mathbb{R})$:
\begin{eqnarray}\label{psl2r_action}
    t\to \frac{at+b}{ct+d}\ , \qquad ad-bc=1
\end{eqnarray}
for real $a,b,c,d$ quotiented out by $(a,b,c,d)\sim (-a,-b,-c,-d)$. In terminology motivated by the dual geometry, we refer to the $t$-coordinates as the ``Poincar\'e coordinates" (see Section \ref{sec:bulk-dual} for more detail). To obtain the full conformal group, we only need to add to the $t\to a t+b$ group the inversion $t\to -1/t$.\footnote{A special conformal transformation is an inversion followed by translation followed by another inversion.} Inversion sends $t=0$ to infinity. Adding a point at $t=\infty$ results in a circle, and taking the universal cover of this circle gives a line covered by ``global coordinates" $\tau\in (-\infty,\infty)$. The 
coordinate transformation $\tau=2\tan^{-1}(t)$ maps the Poincar\'e time $t\in (-\infty,\infty)$ to the range $\tau\in (-\pi,\pi)$ in global time. The universal cover $\widetilde{PSL}(2,\mathbb{R})$ acts on global time. The translation $\tau\to \tau+2\pi$ is in the center of $\widetilde{PSL}(2,\mathbb{R})$ sends one Poincar\'e patch to the next. An element of $\widetilde{PSL}(2,\mathbb{R})$ can be labeled by a $PSL(2,\mbR)$ element and an integer that keeps track of shifts of $\tau$ by $2\pi$.

Another useful choice of coordinates is $z=\frac{i-t}{i+t}$ which manifests that adding a point at infinity to $\mathbb{R}$ results in a circle $|z|=1$. This change of variable identifies the $PSL(2,\mathbb{R})$ conformal group with $PSU(1,1)$:
\begin{eqnarray}
    z\to \frac{\alpha z+\beta}{\beta^* z+\alpha^*}
\end{eqnarray}
where $|\alpha|^2-|\beta|^2=1$. In this frame, the inversion map is simply rotation by $\pi$: $z\to e^{i\pi} z$. This transformation squares to identity in $PSU(1,1)$, however, in its universal cover it becomes a rotation by $2\pi$ that takes us to the next patch (see Appendix \ref{app:psl2r} for a review of $PSL(2,\mathbb{R})$ and its universal cover $\widetilde{PSL}(2,\mbR)$). 

To construct a representation of $PSL(2,\mathbb{R})$ we consider functions of time $f(t)$ and translations $T_{b}$ and scaling $D_s$ transformations:
\begin{eqnarray}\label{eq-rep-dilatation}
     T_{b} :  f(t) &\to& f(t-b) \nn \\
    D_{s} :   f(t) &\to& e^{s(\Delta-1)}  f\left(e^{-s} t\right)
\end{eqnarray}
Note that since $e^{-i\omega t}$ provide a basis for functions on the real line $\mathbb{R}_t$, positive frequency functions $f_+(t)$ can be analytically continued to the lower half $t$-plane and negative frequency functions $f_-(t)$ can be analytically continued to the upper half $t$-plane. Since $PSL(2,\mathbb{R})$ is the isometry group of the upper half-plane, the conformal group sends positive frequency functions to positive frequency functions and negative frequency functions to negative frequency functions. This immediately implies that any transformation of the form
\begin{eqnarray}\label{diagonaltransf}
    f_\pm(t)\to c_\pm f_\pm(t)
\end{eqnarray}
with constants $c_\pm$ commutes with all conformal transformations.\footnote{Note that such transformations are, in general, not local in time (not point transformations in $\mathbb{R}_t$)  because they act differently on the positive and negative frequency modes.}

The inversion map $t\to -1/t$ acts as 
\begin{eqnarray}
&& f(t)\to f'(t)\nn\\
&&(P_{\pm}f)'(t)= P_{\pm}\left[|t|^{2(\Delta-1)} e^{\mp 2\pi i \Delta \Theta(t)}  f (-1/t)\right] \, . \label{eq-inv-map}
\end{eqnarray}
Squaring this transformation we find
\begin{eqnarray}\label{tau2}
f_\pm(t)\to e^{\pm 2\pi i\Delta}f_\pm(t)
\end{eqnarray}
which is of the form (\ref{diagonaltransf}).
This makes it explicit that for integer $\Delta=n\in\mathbb{Z}_+$ rotation by $2\pi$ is identity, $f(t)$ form a representation of $PSL(2,\mathbb{R})$, and conformal transformations act locally \cite{Yngvason:1994nk}:
\begin{eqnarray}\label{eq-psl2r-rep-loc}
    f(t)\to   (a-ct)^{2(n-1)} f\lb \frac{d\,t-b}{a-ct}\rb\ .
\end{eqnarray}

However, for $\Delta\notin\mathbb{Z}_+$ the transformation in (\ref{tau2}) is a non-trivial transformation commuting with all $PSL(2,\mathbb{R})$ generators. We will see below that, for non-integer $\Delta$, we obtain a representation of the universal cover $\widetilde{PSL}(2,\mbR)$ : 
\begin{eqnarray}\label{universal_psl2r_action}
   (P_\pm f)(t)\to (P_\pm f')(t)= P_\pm\lb |a-ct|^{2(\Delta-1)} e^{\mp 2\pi i\Delta(\sgn(c) \Theta(ct-a) + n)}f\lb \frac{d\,t-b}{a-ct}\rb\rb
\end{eqnarray}
where $n$ is an integer. $\widetilde{PSL}(2,\mbR)$ can be obtained by unwrapping the rotation part of $PSL(2,\mathbb{R})$ in the KAN decomposition, which is unwrapping $S^1$ to its universal cover $\mathbb{R}$. We can label the elements in $\widetilde{PSL}(2,\mbR)$ with $(a,b,c,d)$ and $n$. $(a,b,c,d)$ denotes the element in the base as we label them in $PSL(2,\mathbb{R})$. $n$ denotes the $n$-th copy in the sheets. Rotation by $2\pi$ becomes moving to the next copy, instead of returning to the identity.\footnote{For rational $\Delta$, this is also a representation of a $k$-sheeted covering as $e^{\mp 2\pi i\Delta k} = 1$ for some integer $k$.} We assume $n=0$ (elements in the base) if we do not specify the value of $n$. See also Appendix \ref{app:psl2r}.

\begin{lemma}\label{lemma:psl2r_rep}
    The transformations in (\ref{universal_psl2r_action}) form a representation of the universal covering group $\widetilde{PSL}(2,\mathbb{R})$. When $\Delta$ is a positive integer, the transformation becomes (\ref{eq-psl2r-rep-loc}) forming a representation of the $PSL(2,\mathbb{R})$ group. 
\end{lemma}
\begin{proof}
    See Appendix \ref{app:proof_lemma_psl2r_rep} for a proof.
\end{proof}
The special case $a=1=d$ and $b=0$ corresponds to the special conformal transformation $t\to t/(1+ct)$. The expression above makes it manifest that special conformal transformations also treat the positive and negative frequencies differently, and are hence non-local. In Appendix \ref{proof:cor-local-net-integer-delta}, we will elaborate on this point further by providing an explicit example.
 
The vacuum of conformal GFF is defined by the two-point function 
\begin{eqnarray}\label{conformal0p1}
&&\braket{\varphi(f)\varphi(h)}=\iint_{-\infty}^\infty dtdt'f(t)h(t') G(t-t')\nn\\
&&G(t-t')=\Gamma(2\Delta) \frac{e^{-i\pi\Delta  \sgn(t-t')}}{|t-t'|^{2\Delta}} \nn\\
&&G(\omega)=\sqrt{2\pi}\omega^{2\Delta-1}\Theta(\omega),\nn\\
&&\quad \rho(\omega) = |\omega|^{2\Delta-1}\text{sgn}(\omega),
\end{eqnarray}
and is invariant under the action of $\widetilde{PSL}(2,\mbR)$ 
as established by the following lemma:
\begin{lemma}\label{lemma-psl-unitary}
    Consider $0+1$-dimensional conformal GFF with the correlator in (\ref{conformal0p1}) and the action of $\widetilde{PSL}(2,\mbR)$
    on functions $f$ in (\ref{universal_psl2r_action}). Then,
    \begin{enumerate}
    \item The two-point function is invariant:
    \begin{eqnarray}
    \braket{\varphi(f)\varphi(h)}=\braket{\varphi(f')\varphi(h')}\ .
    \end{eqnarray}
    \item All correlation functions of the Weyl operators are also invariant:
    \begin{align}
        \braket{W(f_{1})\cdots  W(f_{n})} = \braket{W(f'_{1}) \cdots W(f'_{n})}\ . 
    \end{align}
    \end{enumerate}
    Therefore, we have a state-preserving unitary representation of $PSL(2,\mathbb{R})$ when $\Delta$ is an integer or that of $\widetilde{PSL}(2,\mathbb{R})$ when $\Delta$ is not an integer. We denote the unitary by $U_g$. 
\end{lemma}
\begin{proof}
See Appendix \ref{app:proof_lemma_psl2r_unitary} for a proof.
\end{proof}
The proof of the above lemma manifests that for the vacuum correlators to remain conformally invariant in the case of non-integer $\Delta$, the special conformal transformations must treat the positive and negative frequencies differently; hence they are nonlocal.

\subsection{Generalized Hilbert Transformation (GHT)}

The expression in (\ref{tau2}) motivates defining the {\it Generalized Hilbert Transform} (GHT):
\begin{eqnarray}\label{GHT}
    H_\Delta:f_\pm(\omega)\to e^{\mp i\pi \Delta}f_\pm(\omega)
\end{eqnarray}
or equivalently, in the notation of \eqref{Poissonbraket}
\begin{eqnarray}\label{GHT2}
H_\Delta:F(\omega)\to \begin{pmatrix} e^{- i\pi\Delta} & 0\\
0 & e^{ i\pi\Delta}\end{pmatrix} F(\omega)
\end{eqnarray}
which also commutes with all $PSL(2,\mathbb{R})$ transformations. The GHT transform squared 
$H_\Delta^2$ is non-trivial in the universal cover. In global coordinates, $H_\Delta^2$ is the translation $\tau\to \tau+2\pi$. It follows from (\ref{innerproduct}) that the GHT also preserves the two-point function because the matrix in (\ref{GHT}) commutes with the diagonal matrix in (\ref{innerproduct}):
\begin{eqnarray}
   \left[ \begin{pmatrix} e^{-i\pi \Delta} & 0\\
0 & e^{i\pi \Delta}\end{pmatrix} , \begin{pmatrix} b_+^2 & 0\\
0 & b_-^2\end{pmatrix}\right]=0\ . 
\end{eqnarray}
We will see in Lemma \ref{lamma-anitpodal} in Section \ref{sec:bulk-dual} that the bulk geometry dual to conformal GFF geometrizes the GHT as a local bulk transformation that pushes operators behind the Poincar\'e horizon (antipodal transform). In the remainder of this work, we denote the unitary transformation corresponding to GHT $H_\Delta$ with $\mT_\Delta$ with the index $\Delta$ to emphasize that it does not correspond to point transformations on $\mathbb{R}_t$ i.e., it is a non-local transform:\footnote{Note that since GHT is a symmetry, we need to represent it unitarily.}
\begin{eqnarray}\label{GHT3}
    \mT_\Delta^\dagger\varphi(f)\mT_\Delta=\varphi(H_\Delta^{-1} (f))\ .
\end{eqnarray}

\begin{corollary}
    The generalized Hilbert transform (GHT) (\ref{GHT}) leaves both the two-point functions in (\ref{conformal0p1}) and all the correlators of Weyl operators invariant. 
\end{corollary}

\subsection{Time interval algebras for integer $\Delta$}

We start this section with a review of the definition of M\"obius covariant, local, net of von Neumann algebras \cite{longo}. 
\begin{definition}\label{def-net-Longo}
A {\it net} of von Neumann algebras is a map from every open interval $I \subset \mathbb{R}$ to a von Neumann algebra $\mA_{I}$ that respects inclusions (isotony):  $\mA_{I_{1}} \subset \mA_{I_{2}}$ if $I_{1} \subset I_{2}$. 
\end{definition}

\begin{definition}[Covariance]
A net of von Neumann algebras in $\mathbb{R}$ is said to be {\it M\"obius covariant} if there exists a strongly continuous positive-energy unitary representation of $PSL(2,\mathbb{R})$ such that
    \begin{align}
        U^{\dagger}_{g} \mA_{I} U_{g} \, = \, \mA_{gI} \, ,
    \end{align}
and there exists a unique $U_{g}$ invariant vacuum state which is cyclic for $\vee_{I} \mA_{I}\, $.
\end{definition}

\begin{definition}(Locality)
A M\"obius covariant net of von Neumann algebras is called {\it local} if $\mA_{I_{1}}$ and $\mA_{I_{2}}$ commute when $I_{1}$ and $I_{2}$ are disjoint intervals
\begin{align}
    I_{1} \cap I_{2} = \emptyset \implies  \mA_{I_{1}} \subseteq (\mA_{I_{2}})' \ .
\end{align}
\end{definition}

\begin{definition}[Haag's duality]
   A local M\"obius covariant net of von Neumann algebras satisfies Haag's duality when each local algebra satisfies Haag's duality
    \begin{eqnarray}\label{Haagdual}
        (\mA_I)'  =  \mA_{I^{c}}
    \end{eqnarray}
     where $I^{c}$ is the complement of the closure $\bar{I}$.
\end{definition}

It is natural to ask for what spectral functions $\rho(\omega)$ the time interval algebra of $0+1$-dimensional conformal GFF is a local M\"obius covariant net of von Neumann algebras.
The vacuum of GFF is cyclic and separating for time interval algebras.\footnote{See Lemma 21 of \cite{Furuya:2023fei} for a proof.}
Intuitively, we can think of $0+1$-dimensional GFF as a QFT with degrees of freedom organized in the time direction, instead of space.\footnote{There are other QFT with algebras in time. In any chiral CFT in $1+1$ dimensions, we can associate algebras to time intervals. In fact, a free chiral massless scalar in AdS$_2$ is equivalent to the conformal GFF on the boundary of AdS$_2$ with $\Delta=1$.} The definitions above imply the following for the time interval algebras of GFF:
\begin{enumerate}
    \item Time interval algebras of GFF form a net of von Neumann algebras in time if and only if the map $I\to \mA_I$ is injective.\footnote{Clearly, this fails to be the case for any theory with finite causal depth parameter $T_{dep}$.}

    \item Conformal GFF for integer $\Delta$ is a local M\"obius covariant net.

    \item Conformal GFF for non-integer $\Delta$ is neither M\"obius covariant nor local, and violates Haag's duality.
\end{enumerate}
The first statement follows trivially from our discussion of the previous section. 

\begin{theorem}[Integer $\Delta$] \label{thm-loc-mobcov-net-integer-delta}
    When $\Delta$ is a positive integer, the map 
    \begin{align}
        I \to \mA_{I} = \mA_{I}'' 
    \end{align}
   is a local M\"obius covariant net of von Neumann algebras and the net satisfies Haag's duality.
\end{theorem}
\begin{proof}
    See Appendix \ref{proof:thm-loc-mobcov-net-integer-delta} for a proof.
\end{proof}

\subsection{Time interval algebras for non-integer $\Delta$}

When $\Delta$ is not an integer, $G_c(\omega)$ is not an entire function and has a branch point at $\omega = 0$. We will see that has important implications for commutant algebra; in particular, it results in a failure of Haag's duality. For non-integer $\Delta$, in Theorem \ref{thm-comm-half-line-general} and Theorem \ref{thm-comm-I-general}, we explicitly work out the commutant algebra of the future algebra and time interval algebras and find that they do not match the algebra of the complementary regions. Then, in the next section, we will see that, in both cases, the modular flow of the operators inside the algebra remains local, however, we find that modular conjugation involves non-local GHTs $\mT_\Delta$ in (\ref{GHT}). The modular flow of operators in the commutant algebra is, in general, non-local. In the special case where the interval is a half-line, i.e. $(t,\infty)$ or $(-\infty,t)$ for some $t$, the modular flow remains local even for operators in the commutant algebra. This is because the origin of nonlocality can be traced back to the dual operators crossing the Poincar\'e horizon. In the case of half-line, the modular flown operators in the commutant are behind the horizon and never cross it. 
\begin{prop} \label{lemma-support-1}
    Suppose $f(t) \in L^2(d\mu,\mbR_t) $. Denote by $f(t)\ast h(t)$ the convolution of the functions $f(t)$ and $h(t)$. 
    \begin{enumerate}
        \item The support of $\Theta(t) \ast f(t)$ is in a half-line $(a,\infty)$ if and only if the support of $f(t)$ is in  $(a,\infty)$. 
        \item The support of $\Theta(-t) \ast f(t)$ is in a half-line $(-\infty,a)$ if and only if the support of $f(t)$ is in  $(-\infty,a)$. 
        \item If $f(t)$ has support in the interval $(a,b)$, then $\Theta(t) \ast f(t)$ has support in $(a,\infty)$ whereas $\Theta(-t) \ast f(t)$ has support in $(-\infty,b) \, $.
    \end{enumerate}
\end{prop}
\begin{proof}
    See Appendix \ref{proof:lemma-support-1} for a proof. 
\end{proof}

\begin{theorem}  \label{thm-comm-half-line-general}
    The commutant of $\mA_{\mathbb{R}^{+}}$ is the generalized Hilbert transform (GHT) of the algebra of the complementary region
    \begin{align}
        \left(\mA_{\mathbb{R}^{+}}\right)' =  \mT_\Delta^\dagger \mA_{\mathbb{R}^{-}} \mT_\Delta
    \end{align}
    and similarly the commutant of $\mA_{\mathbb{R}^{-}}$ is
    \begin{align}
        \left(\mA_{\mathbb{R}^{-}}\right)' = \mT_\Delta \mA_{\mathbb{R}^{+}}  \mT_\Delta^\dagger 
    \end{align}
    where $\mT_\Delta$ is the unitary representing the GHT in \eqref{GHT3}. As a result, $\mA_{\mathbb{R}^+}$ and $\mA_{\mbR^-}$ are von Neumann algebras.
\end{theorem}
\begin{proof}
Consider the case for the commutant of $\mA_{\mathbb{R}^{+}}$. We will prove the equality by showing inclusions in both directions $\mT_\Delta^\dagger  \mA_{\mathbb{R}^{-}} \mT_\Delta  \subseteq (\mA_{\mathbb{R}^{+}})'$ and $(\mA_{\mathbb{R}^{+}})' \subset\mT_\Delta^\dagger  \mA_{\mathbb{R}^{-}} \mT_\Delta $. To show the first inclusion, consider the inverse Fourier transform 
\begin{eqnarray}\label{piplusminus}
    \theta_\pm(t) = \mF^{-1}\lb |\omega|^{2(\Delta-1)} e^{\mp i \pi \Delta \sgn(\omega) }\rb = \frac{2\Gamma(2\Delta-1)}{\sqrt{2\pi}} \frac{\Theta(\pm t)}{|t|^{2\Delta-1}} \ .
\end{eqnarray}
Let $W(f) \in \mA_{\mathbb{R}^{+}}$ and let $W(h) \in \mT_\Delta^\dagger \mA_{\mathbb{R}^{-}} \mT_\Delta $. Then the commutator \eqref{commutatorfg} for the smeared fields vanishes
\begin{eqnarray}
    [\varphi(f),\varphi(h)] &=&  \int_{-\infty}^\infty d\omega\, \omega  |\omega|^{2(\Delta-1)} e^{i\pi\Delta\sgn(\omega)}h(\omega) f(-\omega) \nn \\
    &=&  \int_{-\infty}^\infty d\omega\, \omega  \lb \theta_-(\omega) h(\omega) \rb f(-\omega) \nn \\
    &=& \frac{i}{\sqrt{2\pi}}\int_{-\infty}^\infty dt\, \p_t(\theta_- \ast h)(t)\, f(t) = 0
\end{eqnarray}
since $\theta_- \ast h$ is supported only in $\mbR^-$. Hence we have the inclusion $\mT_\Delta^\dagger  \mA_{\mathbb{R}^{-}}  \mT_\Delta \subseteq (\mA_{\mathbb{R}^{+}})'$. To show the reverse inclusion assume that $W(f) \in \mA_{\mathbb{R}^{+}}$ and let $W(h) \in (\mA_{\mathbb{R}^{+}})'$. Then
\begin{eqnarray}
    0 &=& \lbrack \varphi(f), \varphi(h)\rbrack = \int_{-\infty}^\infty d\omega\, \omega  |\omega|^{2(\Delta-1)} h(\omega) f(-\omega) \nn \\
    &=&  \int_{-\infty}^\infty d\omega\, \omega   \lb |\omega|^{2(\Delta-1)} e^{i\pi\Delta\sgn(\omega)} \rb \lb e^{-i\pi\Delta\sgn(\omega)}h(\omega)\rb f(-\omega) \nn \\ 
    &=& -\frac{i}{\sqrt{2\pi}} \int_{-\infty}^\infty dt\, (\theta_- \ast H_\Delta (h))(t)\, \p_t f(t)\ .
\end{eqnarray}
Since $f(t)$ has support in $\mathbb{R}^{+}$, the commutator vanishes for all $f(t)$ only when $H_\Delta (h)$ has support in $\mathbb{R}^{-}$. Hence, using (\ref{GHT3}) we find $ W\big(H_{\Delta}(h)\big)  =  \mT_\Delta W(h)  \mT_\Delta^{\dagger} \in \mA_{\mbR^-}$. Thus, we have shown the reverse inclusion $(\mA_{\mathbb{R}^{+}})' \subset\mT_\Delta^{\dagger} \mA_{\mathbb{R}^{-}} \mT_\Delta$.
The proof for the commutant of $\mA_{\mbR^-}$ follows identically using $\theta_+(t)$. We can see that $\mA_{\mbR^+}$ and $\mA_{\mbR^-}$ are von Neumann algebras since the commutants $(\mA_{\mathbb{R}^{+}})'$ and $ (\mA_{\mathbb{R}^{-}})'$ are von Neumann algebras and $\mT_\Delta$ is a unitary
\begin{eqnarray}\label{eq-vonNeumann-alg-Rp-general}
    \mA_{\mathbb{R}^{+}} &=& \mT_\Delta^{\dagger} (\mA_{\mathbb{R}^{-}})' \mT_\Delta \nn \\
     \mA_{\mathbb{R}^{-}} &=&  \mT_\Delta(\mA_{\mathbb{R}^{+}})' \mT_\Delta^{\dagger} \ .
\end{eqnarray}
\end{proof}
The commutant of the time intervals algebras is derived in the following theorem:
\begin{theorem} \label{thm-comm-I-general}
Consider a bounded interval $(a,b)$. 
The commutant algebra of $\mA_{(a,b)}$ is
\begin{eqnarray}\label{eq-res-comm-I-general-Delta}
    (\mA_{(a,b)})'  = \lb \mT_\Delta^\dagger \mA_{(-\infty,a)} \mT_\Delta \rb \vee \lb \mT_\Delta \mA_{(b,\infty)} \mT_\Delta^\dagger \rb \ .
\end{eqnarray}
As a result, $\mA_{(a,b)}$ is a von Neumann algebra.
\end{theorem}
\begin{proof}
    See Appendix \ref{proof:thm-comm-I-general} for a proof.
\end{proof}

\begin{corollary}\label{cor:local-net-integer-delta}
    The map $I \to \mA_{I}$ defines a net of von Neumann algebras. However, for non-integer $\Delta$ the net is neither local nor M\"obius covariant, and Haag's duality is violated.
\end{corollary}
\begin{proof}
    See Appendix \ref{proof:cor-local-net-integer-delta} for a proof.
\end{proof}

In Figure \ref{fig:thm4and6} and Figure \ref{fig:commutant-time-interval} of Section \ref{sec:bulk-dual}, we prove a geometric interpretation in the dual bulk geometry for the theorem and corollary above.

\section{Modular flow and conjugation in conformal GFF} \label{sec:modular-flow}

In this section, we compute the modular flow of time interval algebras of GFF. A result from the theory of von Neumann algebras that plays an important role in our derivation of modular flow is
\begin{lemma}\label{stratila}
Modular flow of a von Neumann algebra $\mA$ represented on the vacuum state $\ket{\Omega}$ is the unique state-preserving automorphism $\sigma_s$ of $\mA$
that satisfies the KMS condition:
\begin{eqnarray}
    \forall a_1,a_2\in \mA: \braket{a_1\Omega|\sigma_{s+i}(a_2)\Omega}=\braket{\sigma_s(a_2)\Omega|a_1\Omega}\ , \qquad s\in \mbR\ .
\end{eqnarray}
\end{lemma}
\begin{proof}
See \cite{stratila1981modular} for a proof.
\end{proof}
In our notation, the modular flow is represented on the Hilbert space using the adjoint unitary action 
\begin{eqnarray}
\forall a\in \mA:\qquad \sigma_s(a)=\Delta^{-is}a\Delta^{is}\in\mA\ .
\end{eqnarray}

\subsection{Future/Past algebras}

Let us first consider the algebra of half-line $\mA_{\mathbb{R}^{+}}$. In this case, we will show that the modular flow is dilatation. More precisely, we will show that the action of the modular flow $\Delta_{\mathbb{R}^{+}}^{is}$ on $W(f) \in \mM_{\mathbb{R}^{+}}$ is given by 
\begin{eqnarray}
    \Delta_{\mathbb{R}^{+}}^{-is}  W(f) \Delta_{\mathbb{R}^{+}}^{is}  =  W\left(D_{2\pi s}[f]\right) \ , \qquad s \in \mathbb{R} \label{eq-modflow-dil}
\end{eqnarray}
where $D_{s}$ is the dilatation defined in \eqref{eq-rep-dilatation} and 
\begin{eqnarray}
    D_{2\pi s}[f](t)  = e^{2\pi(\Delta-1)s}  f(e^{-2\pi s}t)  \ .
\end{eqnarray}

\begin{theorem}[Modular flow]\label{ModFlowThmInsideAlg}
    Consider the Weyl algebra of the half-line $\mA_{\mbR^+}$ in the vacuum state $\Omega$ for arbitrary $\Delta$. The modular flow $\Delta^{is}_{\mbR^+}$ is dilatation $D_{2\pi s}$ and the algebra $\mA_{\mbR^+}$ is a modular future algebra.\footnote{See appendix \ref{App:Hsmi} for a definition of modular future algebras.}
\end{theorem}
\begin{proof}
For any Weyl operator $W(f)\in\mA_{\mbR^+}$, $f(t)$ has support in $\mathbb{R}^{+}$ and $D_{2\pi s}[f](t)$ also has support in $\mathbb{R}^{+}$.  
Consider the correlation function
\begin{eqnarray}
    F(s)  =  \braket{W(f)  \Delta_{\mathbb{R}^{+}}^{-is} W(g) \Delta_{\mathbb{R}^{+}}^{is} } =  \braket{W(f)  W\left(D_{2\pi s}[g]\right) } \label{eq-kms-Fs}
\end{eqnarray}
for $s \in \mathbb{R}  $. The KMS condition states that $F(s)$ can be analytically continued to a strip $0< \text{Im}(s)< 1$ and is continuous in the closure with the boundary value
\begin{eqnarray}
    F(s+i)  =  \braket{\Delta_{\mathbb{R}^{+}}^{-is} W(g) \Delta_{\mathbb{R}^{+}}^{is}  W(f)}  =  \braket{W\left(D_{2\pi s}[g]\right)  W(f) }
\end{eqnarray}
for $s\in \mathbb{R}  $. In a theory of GFF, all the correlation functions are fixed by the two-point functions of the GFF. Therefore we consider the two-point function of the smeared fields
\begin{eqnarray}\label{eq-kms-int}
    \braket{\varphi(f) \varphi(D_{2\pi s}[g])} &=& \braket{\varphi(D_{-\pi s}[f]) \varphi(D_{\pi s}[g])} \nn \\
    &=& 2\pi  \int_0^\infty d\omega\, \omega^{2\Delta-1} f(-e^{-\pi s} \omega) g(e^{\pi s}\omega)
\end{eqnarray}
where in the first line we have used that we have a unitary representation of dilatation from Lemma \ref{lemma-psl-unitary} and in the second line we have used that $D_{2\pi s}[f](\omega) = e^{2\pi s \Delta}f(e^{2\pi s }\omega)$.

Since $f(t)$ and $g(t)$ are supported in $\mathbb{R}^{+}$, their Fourier transform $f(\omega)$ and $g(\omega)$ are analytic in the upper half-plane for complex $\omega$. This means we can analytically continue to complex $s$  for $0< \Im(s) < 1$ such that  $\Im (e^{\pi s}\omega) > 0$ and $\Im (-e^{-\pi s}\omega) > 0$ for positive $\omega>0$. 
Thus, \eqref{eq-kms-int} can be analytically continued in the complex plane $0 \le \text{Im}(s) \le 1$. 
Thus taking $s \to s+i$, we get
\begin{eqnarray}
    \braket{\varphi(f) \varphi(D_{2\pi (s+i)}[g])} = &=& 2\pi  \int_0^\infty d\omega\, \omega^{2\Delta-1} f(e^{-\pi s} \omega) g(-e^{\pi s}\omega) \nn \\
    &=&  \braket{\varphi(D_{2\pi s}[g])\varphi(f) } 
\end{eqnarray}
Thus, the KMS condition is satisfied for the smeared fields and the Weyl algebra \eqref{weyl-algebra} ensures that the KMS condition is satisfied for the Weyl operators. Then, it follows from Lemma \ref{stratila} that dilatation is the modular flow of the algebra $\mA_{\mbR^+}$.

To see that we have a modular future algebra i.e., HSMI$+$, consider a proper subalgebra $\mA_{(a,\infty)}\subset \mA_{\mbR^+}$ of functions supported on the interval $(a,\infty)$ for some $a >0$. Under the action of the modular flow 
\begin{eqnarray}
    \forall s>0:\qquad \Delta^{-is}_{\mbR^+} \mA_{(a,\infty)} \Delta^{is}_{\mbR^+} = \mA_{(e^{2\pi s}a,0)} \subset \mA_{(a,\infty)}\ .
\end{eqnarray}
Hence we have a modular future algebra $\mA_{\mbR^+}$.

\end{proof}

\begin{prop} \label{prop-unitary-related-algebra}
If two algebras $\mA_1$ and $\mA_2$ are related by a state-preserving unitary $U$
\begin{eqnarray*}
    U \ket{\Omega} = \ket{\Omega} \ , \qquad  
    \mA_{1} = U^{\dagger} \mA_{2} U \ ,
\end{eqnarray*}
then their modular operators and modular conjugations are related by
\begin{eqnarray}
    \Delta_{\mA_1} = U^\dagger \Delta_{\mA_2} U \ , \qquad J_{\mA_1} = U^\dagger J_{\mA_2} U \ .
\end{eqnarray}
\end{prop}
\begin{proof}
    See Appendix \ref{proof:prop-unitary-related-algebra} for a proof.
\end{proof}

\begin{corollary}[Modular flow of past algebra]
The modular flow of the past algebra $\mA_{\mathbb{R}^-}$ is dilatation and it satisfies
\begin{align}
    \Delta_{\mathbb{R^-}}^{-is} W(f) \Delta_{\mathbb{R^-}}^{is} = W(D_{-2\pi s}[f]) \, . 
\end{align}
\begin{proof}
    From \eqref{eq-vonNeumann-alg-Rp-general} and using the above Proposition \ref{prop-unitary-related-algebra}, we have 
    \begin{align}
        \Delta_{\mathbb{R^-}}^{is} = \mT_{\Delta} \Delta_{(\mathbb{R^+})'}^{is} \mT_{\Delta}^{\dagger}  
        = \mT_{\Delta} \Delta_{\mathbb{R^+}}^{-is} \mT_{\Delta}^{\dagger} 
        = \mT_{\Delta} D_{-2\pi s} \mT_{\Delta}^{\dagger} 
        = D_{-2\pi s} 
    \end{align}
    where the last equality follows because GHT commutes with dilatation.
\end{proof}
\end{corollary}

\subsubsection*{Modular conjugation}
Next, we discuss the Tomita operator and modular conjugation for $\mA_{\mathbb{R}^\pm}$.  The Tomita operator $S_{\mathbb{R}^\pm}$ is an anti-unitary operator defined as
\begin{eqnarray}
    S_{\mathbb{R^\pm}} W(f)\ket{\Omega}  &=&  W(f)^{\dagger} \ket{\Omega}  =  W(-f) \ket{\Omega}  ,
\end{eqnarray}
where $f(t)$ is a real function with support in $\mathbb{R^\pm}$ and we have used 
the definition of the Weyl operator in \eqref{eq-weyl-def}. Thus the Tomita operator maps 
\begin{eqnarray}
    S_{\mathbb{R^\pm}}: f(t) \to -f(t)\ , \qquad \text{supp}(f) \subseteq \mathbb{R^\pm}\ .
\end{eqnarray}

The modular conjugation operator $J_{\mathbb{R^\pm}}$ of the algebra $\mA_{\mathbb{R^\pm}}$ is an anti-unitary operator and is defined in terms of the polar decomposition of the Tomita operator $S_{\mathbb{R^\pm}}$:
\begin{eqnarray}
S_{\mathbb{R^\pm}} = J_{\mathbb{R^\pm}}\Delta_{\mathbb{R^\pm}}^{1/2} =\Delta_{\mathbb{R^\pm}}^{-1/2} J_{\mathbb{R^\pm}} \, .
\end{eqnarray}
In the following, we show that the action of $J_{\mathbb{R^\pm}}$ is the anti-unitary reflection map followed by the GHT.
\begin{theorem}[Modular conjugation] \label{thm-J_half-line}
    The modular conjugation operators of $\mA_{\mathbb{R}^+}$ and $\mA_{\mathbb{R}^-}$ are given by
    \begin{align}\label{J_half-line}
        J_{\mathbb{R^+}} =& - \mT_{\Delta}^\dagger \mR \nn\\
        J_{\mathbb{R^-}} =& - \mT_{\Delta} \mR \ ,
    \end{align}
    respectively, where $\mR$ is a anti-unitary reflection map that implements $t \to -t$. 
\end{theorem}
\begin{proof}
Let us first consider $\mA_{\mathbb{R^+}}$. 
Let us denote the action of $J_{\mathbb{R^+}}$ by
\begin{align}
    J_{\mathbb{R^+}} W(f) J_{\mathbb{R^+}} = W(f_J) \ .
\end{align}
To prove the theorem, we will first show that for $W(f) \in \mA_{\mathbb{R^+}}$,
\begin{align}
    J_{\mathbb{R^+}} W(f) \ket{\Omega} = W(f_J) \ket{\Omega}
\end{align}
where 
\begin{align}
    f_{J}(\omega) = - e^{i\pi\Delta \text{sgn}(\omega)} f(-\omega) \, .
\end{align}
This will show that $J_{\mathbb{R^+}}$ acts as $- \mT_{\Delta}^\dagger \mR$ on a dense set of states. Then we will show that $J_{\mathbb{R^+}} = - \mT_{\Delta}^\dagger \mR$ is a state-preserving anti-unitary. 

Since $J = \Delta_{\mathbb{R^+}}^{1/2} S_{\mathbb{R^+}}$, 
for any two Weyl operators $W(f), W(g)\in\mA_{\mbR^+}$, we have
\begin{align}
    \braket{W(f) J_{\mathbb{R^+}} W(g)} \, = \, 
    \braket{W(f) \Delta_{\mathbb{R^+}}^{1/2} S_{\mathbb{R^+}} W(g)} \, = \, \braket{W(f) \Delta_{\mathbb{R^+}}^{1/2} W(-g)} \, .
\end{align}
Since the correlation functions of the Weyl operators are determined by the correlation functions of the smeared fields, we instead consider 
\begin{eqnarray}
   \braket{\varphi(f) J_{\mathbb{R^+}}  \varphi(g)} = \braket{\varphi(f)\Delta^{1/2}_{\mbR^+}\varphi(-g)}\ ,
\end{eqnarray}
for any $f(t)$, $g(t)$ with support in $\mathbb{R}^+$. Then due to the KMS condition, the above correlation function can be determined by evaluating the modular flow correlation function in \eqref{eq-kms-int} at $s=i/2$. This yields
\begin{eqnarray}
    \braket{\varphi(f) J_{\mathbb{R^+}}  \varphi(g)} = - 2\pi \int_0^\infty d\omega \, \omega^{2\Delta-1}  f(-\omega) \lb e^{i\pi \Delta}  g(- \omega)\rb  = 2\pi \int_0^\infty d\omega \, \omega^{2\Delta-1}  f(-\omega) g_{J}(\omega)\nn 
\end{eqnarray}
or equivalently,
\begin{eqnarray}
           \braket{\varphi(f) J_{\mathbb{R^+}}  \varphi(g)} = - 2\pi \int_0^{\infty} d\omega \, \omega^{2\Delta-1} \lb e^{-i\pi \Delta}  f(\omega)  \rb g(\omega) = - 2\pi \int_0^{\infty} d\omega \, \omega^{2\Delta-1} f_{J}(-\omega) g(\omega) \nn.
\end{eqnarray}
Therefore, we have shown that $J_{\mathbb{R^+}}$ acts as $- \mT_{\Delta}^\dagger \mR$ on a dense set of states.

Next, we need to show that $- \mT_{\Delta}^\dagger \mR$ is a state-preserving anti-unitary. Since we have already established that GHT is a state-preserving unitary, we only need to show that reflection $\mR$ is a state-preserving anti-unitary. We denote the action of $\mR$ by
\begin{eqnarray}
    \mR W(f) \mR  =  W(f_{R})\ , \qquad f_{R}(t) = f(-t) \ .
\end{eqnarray}
Now note that for any two real-valued function $f(t)$ and $g(t)$ (not necessarily with support in $\mathbb{R^+}$), we have
\begin{eqnarray}
    \braket{\mR \varphi(f) \mR , \mR \varphi(g) \mR } &=& \braket{ \varphi(f_R)  , \varphi(g_R)  } \nn \\ 
    &=& \int_{-\infty}^{\infty} dt_1 dt_2\, G(t_1 - t_2) f_R(t_1) g_R(t_2) \nn \\ 
    &=& \int_{-\infty}^{\infty} dt_1 dt_2\, G(t_2 - t_1) f(t_1) g(t_2) \nn \\ 
    &=& \int_{-\infty}^{\infty} dt_1 dt_2\, G^{*}(t_1 - t_2) f(t_1) g(t_2) \nn \\ 
    &=& \overline{\braket{\varphi(f) ,\varphi(g)} } 
\end{eqnarray}
where we have used $G(-t)=G^{*}(t)$ from \eqref{conformal0p1}. This shows that $\mR$ is a state-preserving anti-unitary operator. This finishes our proof for $J_{\mathbb{R^+}}$.

Now we consider $J_{\mathbb{R^-}}$. From \eqref{eq-vonNeumann-alg-Rp-general} and Proposition \ref{prop-unitary-related-algebra}, we get
\begin{eqnarray}
    J_{\mathbb{R}^{-}} = \mT_{\Delta} J_{(\mathbb{R}^{+})'} \mT_{\Delta}^{\dagger} = \mT_{\Delta} J_{\mathbb{R}^{+}} \mT_{\Delta}^{\dagger} = -\mR \mT_{\Delta}^\dagger \ .
\end{eqnarray}
As we show in Appendix \ref{sec-hilbert}, 
\begin{eqnarray}
    \mR \mT_\Delta \mR =  \mT_\Delta^\dagger \ .\label{eq-comm-T-R}
\end{eqnarray}
Combining the above two results proves the theorem.
\end{proof}
This finishes our discussion of modular flow and conjugation of future and past algebras. Next, we focus on the modular flow of algebras of finite intervals.

\subsection{Time interval algebras}

We are now ready to compute the modular flow of any time interval algebra $\mathcal{A}_{I}$. Since we have already computed the modular flow of $\mA_{\mathbb{R}^+}$, due to Proposition \ref{prop-unitary-related-algebra}, all we need is a state-preserving unitary $U_I$ that maps $\mA_{I}$ to $\mA_{\mathbb{R}^+}$.

We can construct the action of this unitary $U_I$ using the transformations in \eqref{universal_psl2r_action}. As a first guess, we consider the $PSL(2,\mathbb{R})$ transformation that maps $I=(p,q)$ to $\mathbb{R}^+$: 
\begin{eqnarray}\label{conformalUT}
    u:t\to \frac{t-p}{q-t}\ ,\qquad u= \frac{1}{\sqrt{q-p}}
    \begin{pmatrix}
         1 & -p \\
         -1 & q
    \end{pmatrix}
    \in PSL(2,\mathbb{R})\ .
\end{eqnarray}
We remind the reader that, as we established in Corollary \ref{cor:local-net-integer-delta}, for non-integer $\Delta$, the GFF is not M\"obius covariant. Therefore we are not guaranteed that the unitary $U_I$ that implements the conformal transformation above also transforms time interval algebras as $U_I^\dagger\mA_{(p,q)}U_I=\mA_{\mathbb{R}^+}$. However, as we establish in the lemma below, the conformal transformation in (\ref{conformalUT}) is special in that 
it acts locally on $\mA_I$ and the covariant action $U_I^\dagger\mA_{(p,q)}U_I=\mA_{\mathbb{R}^+}$, indeed holds:
\begin{lemma} \label{lemma-interval-map}
   The conformal transformation corresponding to $u$ defined in \eqref{conformalUT} has action $U_I^\dagger W(f)U_I = W(f')$, where
   \begin{eqnarray}\label{eq-U-act}
       f'(t) = \sum_\pm \left[\left|\frac{t + 1}{\sqrt{q-p}} \right|^{2\Delta-2} e^{\pm i 2\pi\Delta\Theta(-t-1)} f\left(\frac{qt+p}{t+1}\right)\right]_\pm\ .
   \end{eqnarray}
   The sum is over positive and negative frequency components. It acts locally on $\mA_I$ and maps $\mA_I$ to $\mA_{\mathbb{R}^+}$
   \begin{eqnarray}\label{interval_to_line}
    U_I^\dagger\mA_{(p,q)}U_I=\mA_{\mathbb{R}^+} \ .
   \end{eqnarray}
\end{lemma}
\begin{proof}
    It is straightforward to check that
    \begin{eqnarray}
        \frac{qt+p}{t+1} \in I = (p,q) \iff t \in\mathbb{R}^+ \implies t + 1 >0\ ,
    \end{eqnarray}
    since
    \begin{eqnarray}
        \frac{qt+p}{t+1} = q -\frac{q-p}{t+1} = p + \frac{(q-p)t}{t+1}\ .
    \end{eqnarray}
    The transformation acts on the positive and negative modes in the same way because $t+1>0$. Hence, it acts locally on $\mA_I$ and maps it to $\mA_{\mathbb{R}^+}$.
\end{proof}
As a matter of fact, $U_I$ still acts non-locally on a general operator, as $t + 1$ can be negative. As a result, we find the action of the modular flow of time interval algebra:
\begin{theorem}[Modular flow for time interval] \label{thm-interval-mod}
    For an algebra $\mA_I$ of the interval $I = (p,q)$, the modular flow  is
    \begin{eqnarray}
        &&\Delta_I^{-is}\varphi(f)\Delta_I^{is} = \varphi(f')\ , \nn\\
        &&f_\pm'(t) = \left[\left(\tfrac{(q-t)e^{\pi s}+(t-p)e^{-\pi s}}{q-p} \mp\mathrm{sgn}(s)i\epsilon\right)^{2\Delta-2} f\left(\tfrac{(q-t)p\,e^{\pi s} + (t-p)q\,e^{-\pi s}}{(q-t)e^{\pi s} + (t-p)e^{-\pi s}}\right)\right]_\pm\ .
    \end{eqnarray}
\end{theorem}
\begin{proof}
    See Appendix \ref{sec-proof-interval-mod}.
\end{proof}
\begin{figure}[t]
    \centering
    \includegraphics[width=0.85\linewidth]{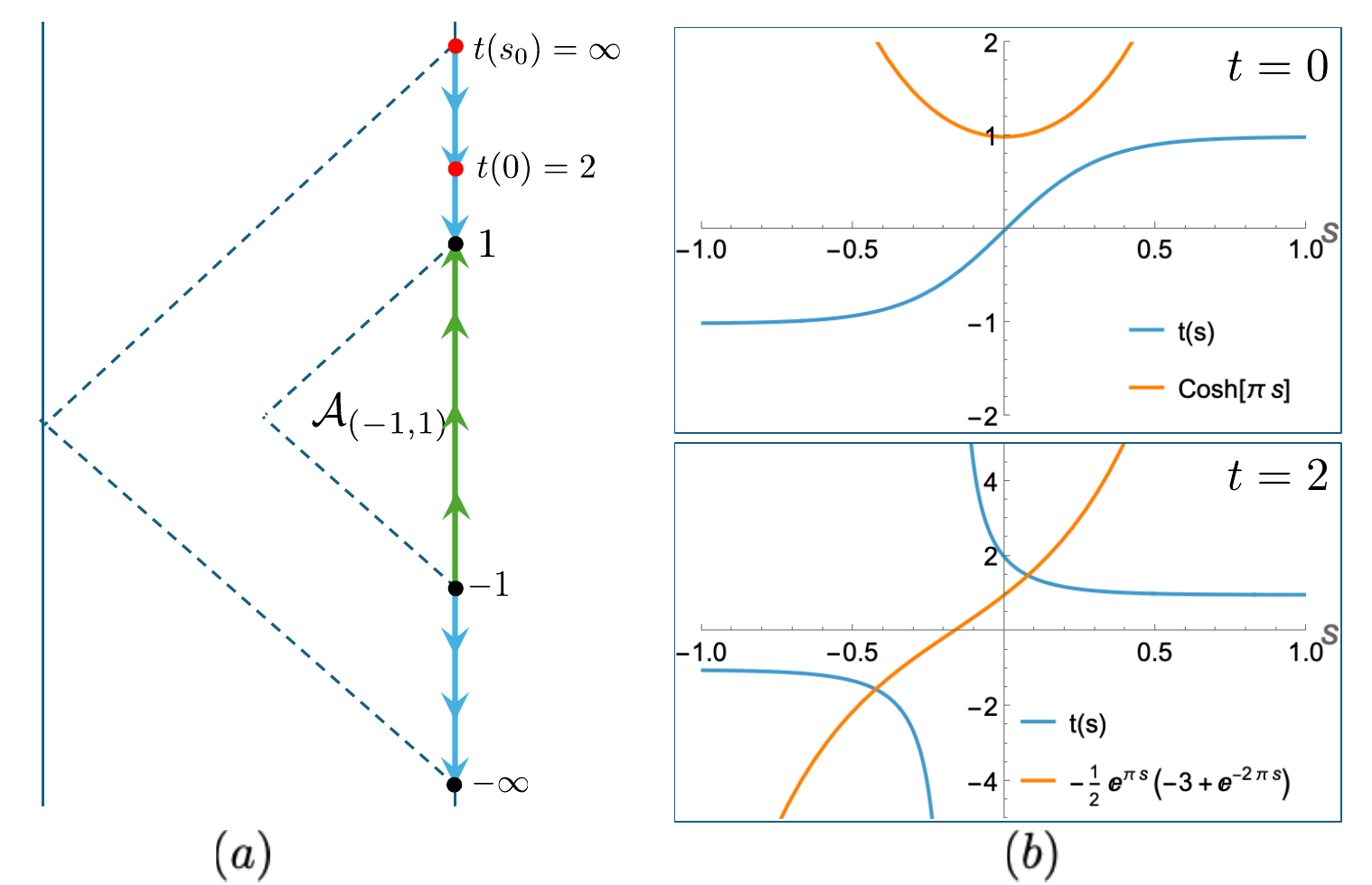}
    \caption{\small Modular flow for the time interval algebra $\mA_{(-1,1)}$: (a) The modular flow (green) is local for the operators inside $\mA_{(-1,1)}$. For an operator outside the interval, the modular flow (blue) takes it to the past/future horizon in finite modular time, $s_0 = \frac{1}{2\pi} \log\lb \frac{t-q}{t-p}\rb$. The modular flow remains local until $s_0$. (b) The top plot shows the modular flow $t(s)$ (blue) for an operator inside $\mA_{(-1,1)}$ starting at $t(0)=0$. The conformal factor (orange) remains positive for all modular time and the flow is local. The bottom plot shows the modular flow (blue) for an operator outside the interval starting at $t(0)=2$. The conformal factor (orange) becomes negative for negative modular flow at finite $s_0<0$ when the operator reaches the future horizon and there is a discontinuity. For $s<s_0$, $t(s)$ is mapped to $(-\infty, -1)$ with a frequency dependent phase. Hence the flow is non-local in general.  }
    \label{fig:interval_modular_flow}
\end{figure}
For any operator in $\mA_I$, the $i\epsilon$ is trivial and the modular flow acts as if it is M\"obius covariant. This result also applies to the (unbounded) affiliated operators that have support only in $I$. For operators not in $\mA_I$, the action is, in general, non-local; see Figure \ref{fig:interval_modular_flow}. This is summarized in the following corollary:
\begin{corollary}\label{modflowtimeintervalBoundary}
    The modular flow of the operators $\varphi(f_I)$ affiliated with $\mA_I$ is local and given by the geometric transformation
    \begin{eqnarray}
   &&\Delta_I^{-is}\varphi(f)\Delta_I^{is} = \varphi(f')\ , \nn\\
   &&f'(t)= \left(\frac{(q-t)e^{\pi s}+(t-p)e^{-\pi s}}{q-p}\right)^{2\Delta-2} f\left(\frac{(q-t)p\,e^{\pi s} + (t-p)q\,e^{-\pi s}}{(q-t)e^{\pi s} + (t-p)e^{-\pi s}}\right)     
    \end{eqnarray}
    whereas, for operators not affiliated with the algebra, the modular flow is non-local in general. The algebra $\mA_I$ is a modular future algebra. 
\end{corollary}
This locality issue has an interesting interpretation in the bulk. In the next section, we will see that, the modular flow of bulk operators outside of the algebra remains local on the boundary for a finite range of modular time parameters until the operator falls through the future or past Poincar\'e horizon. At this point, the boundary modular flow becomes highly nonlocal whereas, in terms of the bulk field, the flow remains local and geometric.

In the following corollary, we explicitly write down the action of the modular conjugation operator for a time interval algebra (see Figure \ref{fig:interval_modular_conjugation}):
\begin{corollary}[Modular conjugation for time interval] \label{thm-J_interval}
    The modular conjugation operator corresponding to the time interval algebra $\mA_I$ for $I=(p,q)$ is given by
    \begin{eqnarray}
        J_I = - \lb \mT_{\Delta}^\dagger \Theta(\tfrac{p+q}{2}-t)  + \mT_\Delta \Theta(t-\tfrac{p+q}{2})  \rb  \mR_I 
    \end{eqnarray}
    where $\mR_I$ is a local anti-unitary map that implements $\mR_I :\varphi(f) \to \varphi(R_I[f])$ with
    \begin{eqnarray}
        R_I[f](t) = \left|\frac{p+q-2t}{q-p} \right|^{2\Delta-2}  f\left(\frac{2pq-(p+q)t}{p+q-2t}\right)
    \end{eqnarray}
    and for convenience we used the notation $\Theta(t)\varphi(f(t)) := \varphi(f(t)\Theta(t))$.
\end{corollary}
\begin{proof}
    See Appendix \ref{app:proof-thm-J_interval} for a proof.
\end{proof}

Before we proceed with the bulk analysis, we end this section with the following observation concerning the modular conjugation operator of the time interval algebra:
\begin{corollary} \label{corr-inv-in-terms-of-Js}
Consider an interval $I = (-1,1)$. The unitary operator defined as $J_{I}J_{\mathbb{R}^+}$ implements the inversion map in \eqref{eq-inv-map}. 
\end{corollary}
\begin{proof}
    See Appendix \ref{app-proof-inv-corr} for a proof.  
\end{proof}
\begin{figure}[t]
    \centering
    \includegraphics[width=0.35\linewidth]{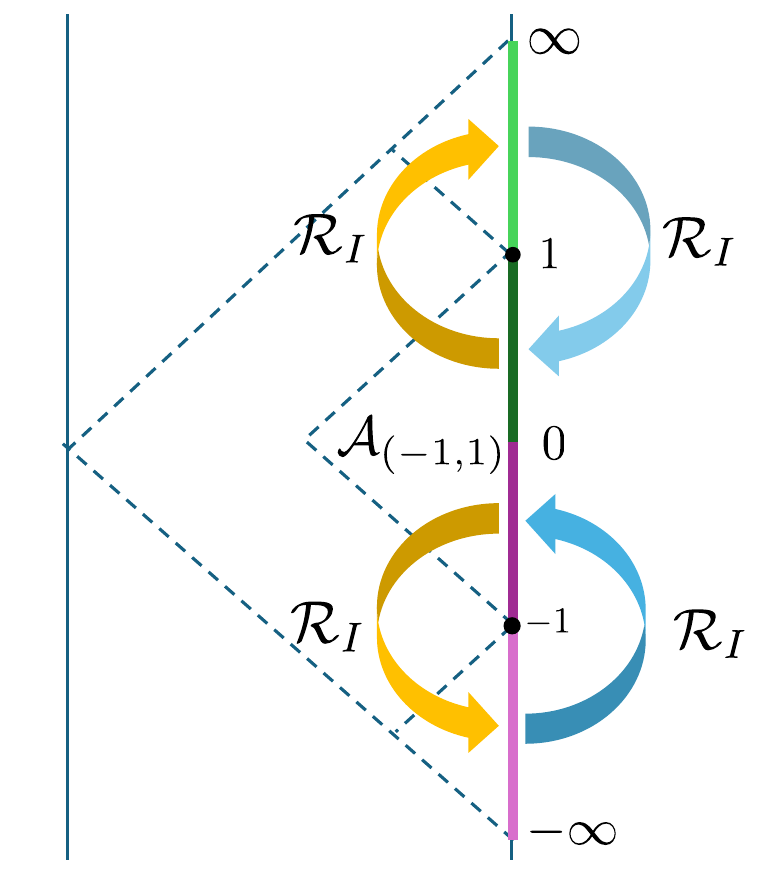}
    \caption{\small For the time interval algebra $\mA_{(-1,1)}$, the local anti-unitary map $\mR_I$ maps $f(t) \to |t|^{2(\Delta-1)} f(1/t)$. The modular conjugation of operators in the interval $(-1,0)$ is the GHT of operators in $(-\infty,-1)$ and the modular conjugation of operators in $(0,1)$ is the GHT of operators in $(1,\infty)$ and vice versa. }
    \label{fig:interval_modular_conjugation}
\end{figure}

\section{Bulk dual of conformal GFF} \label{sec:bulk-dual}

Conformal GFF in $d$-dimensional Minkowski space with spectral density $\rho(\omega)=\omega^{2\Delta-d}$ is dual to massive free fields on AdS$_{d+1}$ with mass $m^2=\Delta(\Delta-d)$. The conformal group $SO(d,2)$ coincides with the isometry group of AdS$_{d+1}$. Conformal $0+1$-dimensional GFF with spectral density $\omega^{2\Delta-1}$ is dual to free fields of mass $m^2=\Delta(\Delta-1)$ living on the boundary of Poincar\'e AdS$_2$ whose isometry group is $PSL(2,\mathbb{R})$.\footnote{The restricted conformal group $SO^+(1,2)\simeq PSL(2,\mathbb{R})\simeq PSU(1,1)$} In the special case where $\Delta=1$, the AdS$_2$ (bulk) fields are massless and conformal. This is consistent with what we found in Lemma \ref{lemma:psl2r_rep}: for $\Delta=1$, conformal GFF transforms as a representation of $PSL(2,\mathbb{R})$, whereas for non-integer $\Delta$, we only have a representation of the universal cover $\widetilde{PSL}(2,\mbR)$.

In this section, using the HKLL reconstruction of $1+1$-dimensional fields in AdS$_2$  \cite{hamilton2006local} we establish that the modular flow and conjugation of time interval algebras that were non-local on the boundary, viewed in terms of bulk fields become local. The key point is that, as we saw in Section \ref{sec:modular-flow}, the GHT was the origin of the non-locality of the modular data of the boundary. In Lemma \ref{lamma-anitpodal}, we show that 
 the bulk dual of the GHT corresponds to a local antipodal transformation.

We start by setting the notation for the global and Poincar\'e coordinates in AdS$_2$. For more details see Appendix \ref{app:ads_geometry}. In global coordinates, the AdS$_2$ metric is
\begin{eqnarray}
    ds^2=\frac{1}{\cos^2\rho}(-d\tau^2+d\rho^2)
\end{eqnarray}
where $\rho\in (-\pi/2,\pi/2)$ and $\tau\in \mathbb{R}$, and we have set the AdS radius to one. The GFF fields live on the two asymptotic boundaries at $\rho=\pm \pi/2$. For every pair of points in global AdS$_2$ it is convenient to define 
\begin{eqnarray} \label{eq-invariant-distance}
    \sigma(\tau,\rho|\tau',\rho')=\frac{\cos(\tau-\tau')-\sin\rho\sin\rho'}{\cos\rho\cos\rho'}=\cosh(s)
\end{eqnarray}
where $s$ is the geodesic proper distance between two points.

The Poincar\'e coordinates $(t,z)$ are given by 
\begin{eqnarray}
    t=\frac{\sin\tau}{\cos\tau+\sin\rho}\ ,\qquad z=\frac{\cos\rho}{\cos\tau+\sin\rho} \ .
\end{eqnarray}
The Poincar\'e patch of AdS$_2$ covered by the coordinates above has an asymptotic boundary at $z=0$ and there are Killing horizons at $(t>0, z\to \infty)$ and $(t<0,z\to \infty)$. On the boundary, the Poincar\'e time is given by 
\begin{eqnarray}
    t=\frac{\sin\tau}{\cos\tau+1}=\tan(\tau/2)\ .
\end{eqnarray}
This is precisely the change of coordinates we made on the boundary in Section \ref{sec-Mobius} to realize the action of the universal cover $\widetilde{PSL}(2,\mbR)$ and the map $\tau\to \tau+2\pi$.

\subsection{Bulk algebras and Haag's duality}

A free scalar field of mass $m$ in global AdS$_2$ can be expanded in a complete set of normalizable modes 
\begin{eqnarray} \label{eq-bulk-mode-expansion}
    \varphi(\tau,\rho)=\sum_{n=0}^\infty a_n e^{-i(n+\Delta) \tau}\cos^\Delta\rho\: C_n^\Delta(\sin\rho) + \text{h.c.}
\end{eqnarray}
    where $C_n^\Delta$ are Gegenbauer polynomials \cite{hamilton2006local}.\footnote{The Gegenbauer polynomials are a special case of the Jacobi polynomials and they are odd for odd $n$ satisfying $C_n^\Delta(-x) = (-1)^n C_n^\Delta(x)$.} The right and left GFF fields are
    \begin{eqnarray}
        &&\varphi^R(\tau)=\lim_{\rho\to \pi/2} \frac{\varphi(\tau,\rho)}{\cos^\Delta\rho}\nn\\
       && \varphi^L(\tau)=\lim_{\rho\to -\pi/2} \frac{\varphi(\tau,\rho)}{\cos^\Delta\rho}\ .
    \end{eqnarray}
The bulk dual of every boundary time interval algebra $\mA_{I}$ corresponds to the causal development of the time interval algebra; see Figure \ref{fig:intro} (a).  Note that there is only one algebra $\mA_I$ and one Hilbert space, and the bulk or boundary fields correspond to different fields (frames) we use to generate the algebra. However, we will see that the notion of locality and Haag's duality are frame-dependent.\footnote{Another frame we can use to generate this algebra is the GFF algebra of time intervals defined on a constant $z$-surface in the bulk. This corresponds to moving the boundary inside the bulk.} In Lemma \ref{lemma-psl-unitary}, we showed that the vacuum of GFF is invariant under the action of the $\widetilde{PSL}(2,\mbR)$. In the bulk, this corresponds to the invariance of the vacuum of massive free fields in Poincar\'e AdS$_2$ under the isometry group of AdS$_2$.\footnote{For completeness, in Appendix \ref{AdS2massive}, we show this invariance directly in the bulk.}

\paragraph{Antipodal symmetry is the GHT:} 
The antipodal symmetry transformation $\mT:(\tau,\rho)\to (\tau+\pi,-\rho)$ of the global AdS$_2$ commutes with the $PSL(2,\mathbb{R})$ isometry group of AdS$_2$. It sends one Poincar\'e patch to the one in the future of it in the embedding in global AdS$_2$; see Figure \ref{fig:thm4and6}. It pushes local bulk fields behind the Poincar\'e horizon and corresponds to a non-local transformation on the boundary. In fact, as we show below, this transformation is the bulk dual to the boundary GHT transform $\mT_\Delta$ we introduced in (\ref{GHT}).
Recall that $\mT_\Delta$ is not part of the universal cover of the conformal group, however, it squares to the transformation $\mT^2_\Delta:\tau\to \tau+ 2\pi$.

\begin{lemma} [Dual of Antipodal and GHT] \label{lamma-anitpodal}
The action of the transformation $\mT:(\tau,\rho)\to (\tau+\pi,-\rho)$ on global AdS$_2$ fields $\varphi(\tau,\rho)$ written in terms of the boundary GFF fields is the Generalized Hilbert Transform (GHT):
\begin{eqnarray}
    \mT_\Delta^\dagger: \varphi_\pm(\tau)\to e^{\mp i\pi \Delta}\varphi_\pm(\tau)\ .
\end{eqnarray}
\end{lemma}
\begin{proof}
From \eqref{eq-bulk-mode-expansion}, the positive and negative modes of the bulk field are
\begin{eqnarray}
    \varphi_+(\tau,\rho)&=&\sum_{n=0}^\infty a_n e^{-i(n+\Delta) \tau}\cos^\Delta\rho\: C_n^\Delta(\sin\rho)\ ,\qquad \varphi_-(\tau,\rho) = \varphi_+^\dagger(\tau,\rho)\nn\\
    \varphi_+(\tau + \pi,-\rho)&=&\sum_{n=0}^\infty a_n e^{-i(n+\Delta)(\tau+\pi)}\cos^\Delta\rho\: C_n^\Delta(-\sin\rho)\nonumber\\
    &=&\sum_{n=0}^\infty a_n e^{-i\pi\Delta}(-1)^n e^{-i(n+\Delta)\tau}\cos^\Delta\rho\: (-1)^n C_n^\Delta(\sin\rho)\nonumber\\
    &=& e^{-i\pi\Delta} \varphi_+(\tau,\rho)\ .
\end{eqnarray}
We can write the bulk field as a smeared boundary field 
\begin{eqnarray}
    \varphi(\tau,\rho) &=& \int d\tau'\, K(\tau'|\tau,\rho) \varphi(\tau')\nn\\
    &=&\int d\tau'\, \lb K_-(\tau'|\tau,\rho) \varphi_+(\tau') + K_+(\tau'|\tau,\rho) \varphi_-(\tau')\rb\ .
\end{eqnarray}
where $K(\tau'|\tau,\rho)$ is the smearing function. The action of $\mT_\Delta$ on the bulk field is
\begin{eqnarray}
    \mT_\Delta^\dagger [\varphi(\tau,\rho)] &=& \int d\tau'\, H_\Delta^{-1}[K(\tau'|\tau,\rho)] \varphi(\tau') \nonumber\\
    &=& \sum_\pm \int d\tau'\, (K_\mp(\tau'|\tau,\rho) e^{\mp i\pi \Delta})\varphi_\pm(\tau) \nonumber\\
    &=& \sum_\pm e^{\mp i\pi \Delta}\varphi_\pm(\tau,\rho)\nonumber\\
    &=& \sum_\pm \varphi_\pm(\tau + \pi,-\rho) = \varphi(\tau + \pi,-\rho)\ .
\end{eqnarray}
We have used the fact that the positive (negative) modes in the global coordinates is still the positive (negative) modes in Poincar\'e coordinates \cite{Spradlin:1999bn,Danielsson:1998wt,hamilton2006local}. Note that the functions $f_\pm$ combine with the operators $\varphi_\mp$ (see \eqref{eq:field}).
\end{proof}

\begin{figure}[t]
   \centering
    \includegraphics[width=\textwidth]{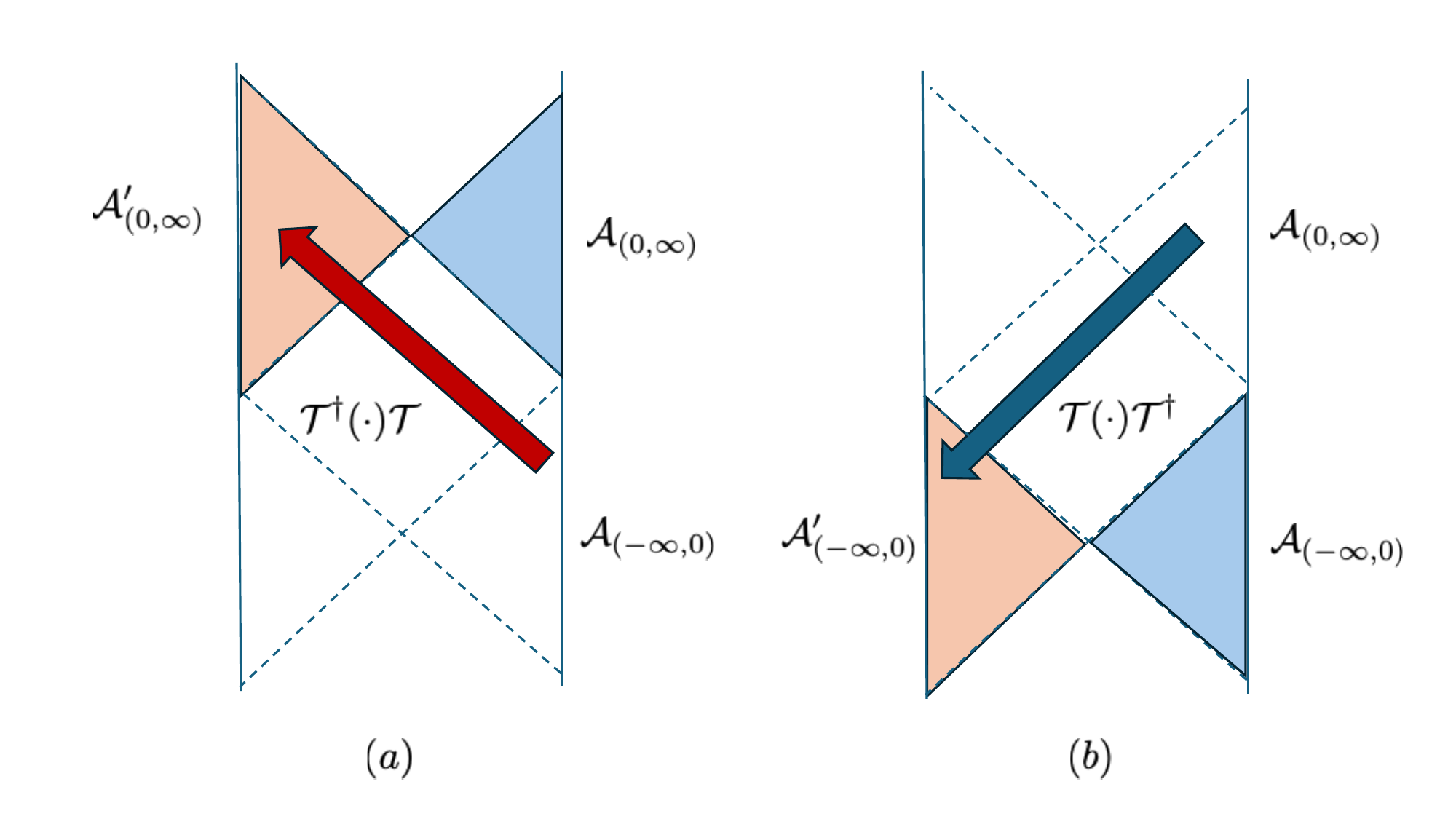}
    \caption{\small Haag's duality is restored in the bulk (a) The geometrization of the commutant relation in Theorem \ref{thm-comm-half-line-general} $(\mA_{(0,\infty)})'=\mT^\dagger \mA_{(-\infty,0)}\mT$. (b) The geometrization of the commutant relation in Theorem \ref{thm-comm-half-line-general} $(\mA_{(-\infty,0)})'=\mT \mA_{(0,\infty)}\mT^\dagger$.}
    \label{fig:thm4and6}
\end{figure}

We saw in Section \ref{sec:gff_conformal} that for conformal GFF with non-integer $\Delta$, Haag's duality is violated. In particular, we saw in Theorem \ref{thm-comm-half-line-general} that the commutant of the algebra $\mA_{(0,\infty)}$ is no longer $\mA_{(-\infty,0)}$ but its GHT: 
\begin{eqnarray}
    (\mA_{(0,\infty)})'=\mT^\dagger \lb \mA_{(-\infty,0)}\rb \mT\ .
\end{eqnarray}
In the bulk, $\mT^\dagger$ sends the patch that covers the Poincar\'e AdS$_2$ in global AdS$_2$ into the next patch by pushing operators behind the future Poincar\'e horizon. Similarly, the inverse map $\mT$ pushes operators behind the past Poincar\'e horizon. As we have depicted in Figure \ref{fig:thm4and6}, in the bulk Haag's duality is restored because $\mT^\dagger \mA_{(-\infty,0)} \mT$ corresponds to the algebra of the bulk causal complement of $\mA_{(0,\infty)}$. This is only possible because there are operators and algebras localized behind the horizon in the bulk which correspond to highly non-local operators and algebras on the boundary.

Similarly, in Theorem \ref{thm-comm-I-general}, we derived the commutant algebra of time interval $\mA_{(-1,1)}$. Once again, the bulk restores Haag's duality as depicted in Figure \ref{fig:commutant-time-interval}. Note that by the time-like tube theorem \cite{strohmaier2024timelike}, the algebraic union of the two orange regions in Figure \ref{fig:commutant-time-interval} (b) is the same as the orange region in Figure \ref{fig:commutant-time-interval} (a).

\begin{figure}[t]
   \centering
    \includegraphics[width=\textwidth]{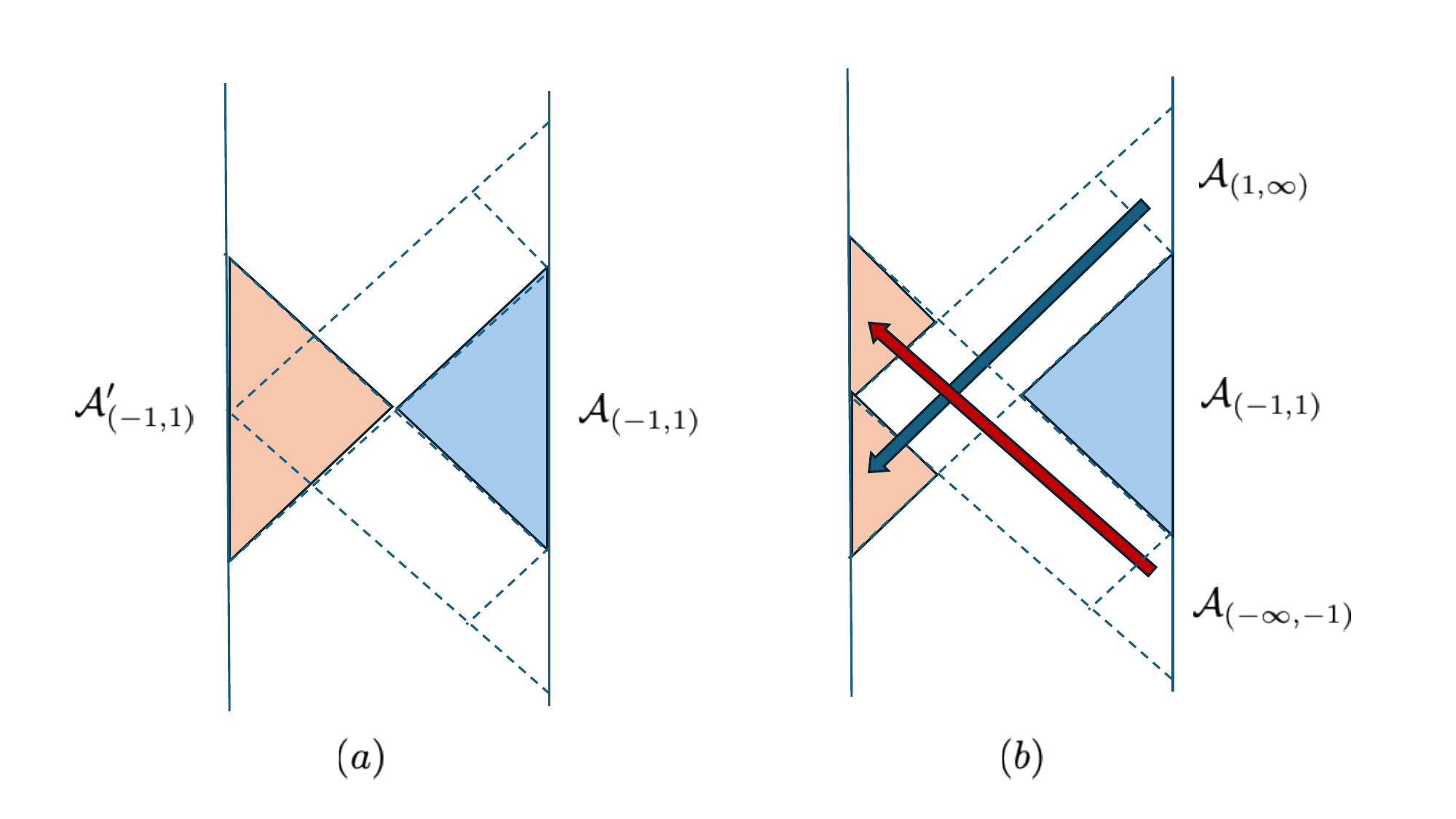}
    \caption{\small Haag's duality is restored in the bulk for time interval algebras. (a) The commutant $(\mA_{(-1,1)})'$ in the bulk. (b) Theorem \ref{thm-comm-I-general} says that $(\mA_{(-1,1)})'$ is the algebraic union of the two orange regions.}
    \label{fig:commutant-time-interval}
\end{figure}

\subsection{Modular flow in Poincar\'e AdS$_2$}

In this subsection, we explore the geometric properties of the modular flow of bulk dual of time interval algebras. 
We start by rewriting the bulk field in Poincar\'e coordinates in terms of the boundary Poincar\'e GFF fields \cite{hamilton2006local,hamilton2006holographic}. 
\begin{align}
    \varphi(t,z) &= \int_{-\infty}^\infty dt'\, K(t'|t,z) \varphi(t'),\nn\\
    K(t'|t,z) &= \frac{2^{\Delta-1}\Gamma(\Delta+1/2)}{\sqrt{\pi}\Gamma(\Delta)}\left(\frac{z^2-(t'-t)^2}{2z}\right)^{\Delta-1} \Theta(z-|t'-t|).
\end{align}
The kernel above (the step function) is non-zero only for boundary points that are spacelike separated from $(t,z)$. It is convenient to introduce the null coordinates $t_\pm=t\pm z$ with $t_+>t_-$ corresponding to the Poincar\'e patch and the smearing function as
\begin{align}
    K'(t|t_+,t_-)
     = \frac{2^{\Delta-1}\Gamma(\Delta+1/2)}{\sqrt{\pi}\Gamma(\Delta)}
    \left(\frac{(t_+-t)(t-t_-)}{t_+-t_-}\right)^{\Delta-1} \Theta\left(\frac{(t_+-t)(t-t_-)}{t_+-t_-}\right)\ .
\end{align}
We can continue the coordinates beyond the Poincar\'e horizon by considering $t_+ < t_-$; see Figure \ref{fig:poincare}. 

Here, we show that the modular flow of bulk fields is local, breaking down the argument into the following three steps:
\begin{enumerate}
    \item {\bf Lemma \ref{lemma-local-action}:} The modular flow of $\mA_I $ acting on a bulk field inside the bulk dual of the interval $I$ is always local and geometric. This is the bulk version of Corollary \ref{modflowtimeintervalBoundary} for the boundary GFF.
    \item {\bf Lemma \ref{lemma-restricted-local}:} The modular flow of $\mA_I $ acting on a bulk field outside the bulk dual of the interval $I$ is local and geometric for finite modular time until the operator reaches the Poincar\'e horizon. 
    \item {\bf Lemma \ref{lemma-geometric-bulk-action}:} As bulk modular flow pushes bulk operators behind the Poincar\'e horizon, the bulk flow remains local and geometric in global AdS$_2$, however, on the boundary it is highly non-local.
\end{enumerate}
We saw in Corollary \ref{cor:local-net-integer-delta} that GFF with general $\Delta$, generated by the boundary GFF fields, is neither local nor satisfies Haag's duality, whereas, in this subsection, we establish that the same algebra generated by the bulk massive free fields is both local and satisfies Haag's duality. We can turn this into the following slogan: {\it The bulk emerges from insisting on a local representation of time interval algebras that satisfies Haag's duality.} 

\begin{figure}[t]
   \centering
    \includegraphics[width=\textwidth]{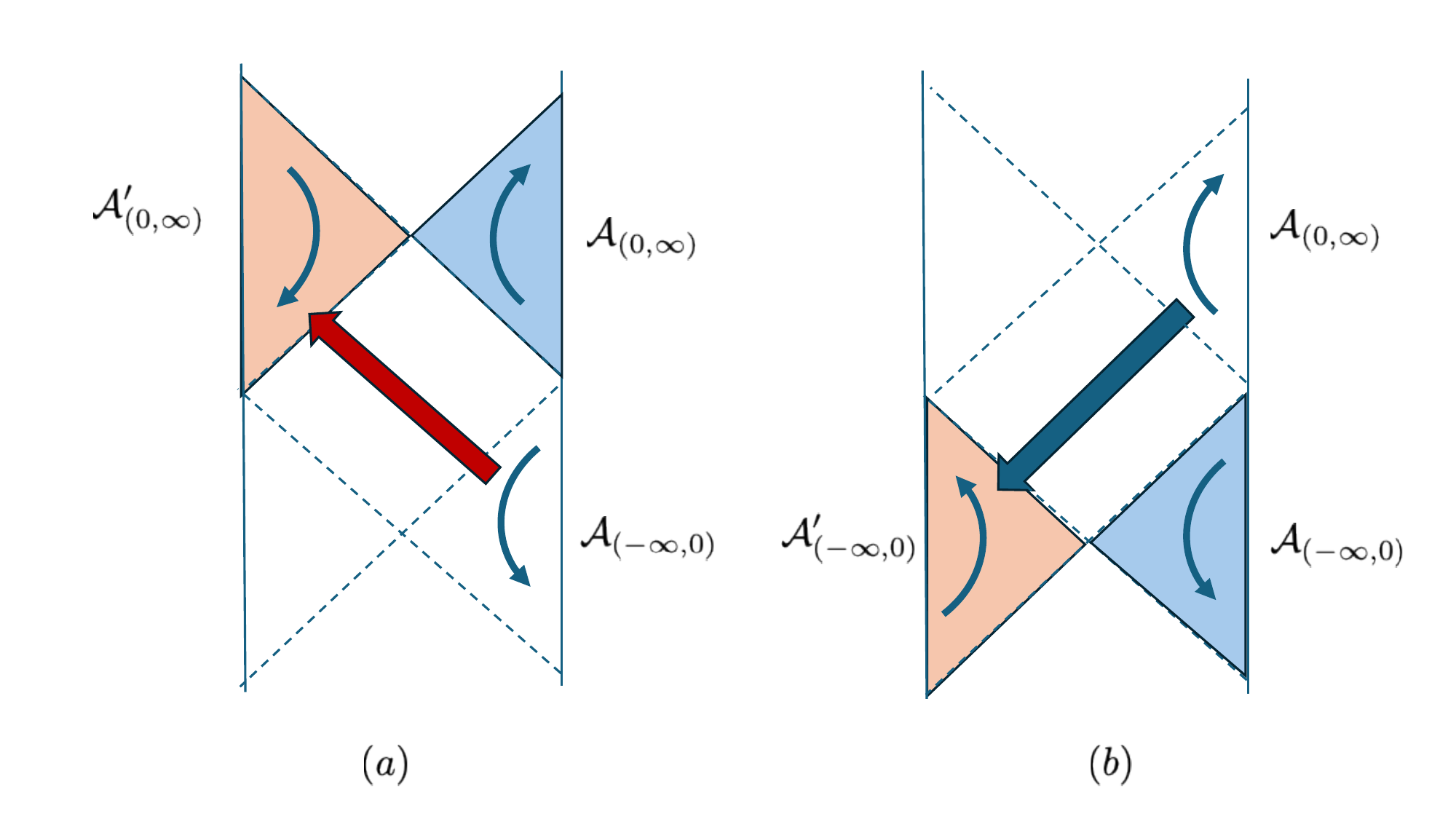}
    \caption{\small Modular flow of half-infinite interval regions $\mA_{(0,\infty)}$ and $\mA_{(-\infty,0)}$ in the bulk.}
    \label{fig:modular-future-past}
\end{figure}
\begin{figure}[t]
   \centering
    \includegraphics[width=\textwidth]{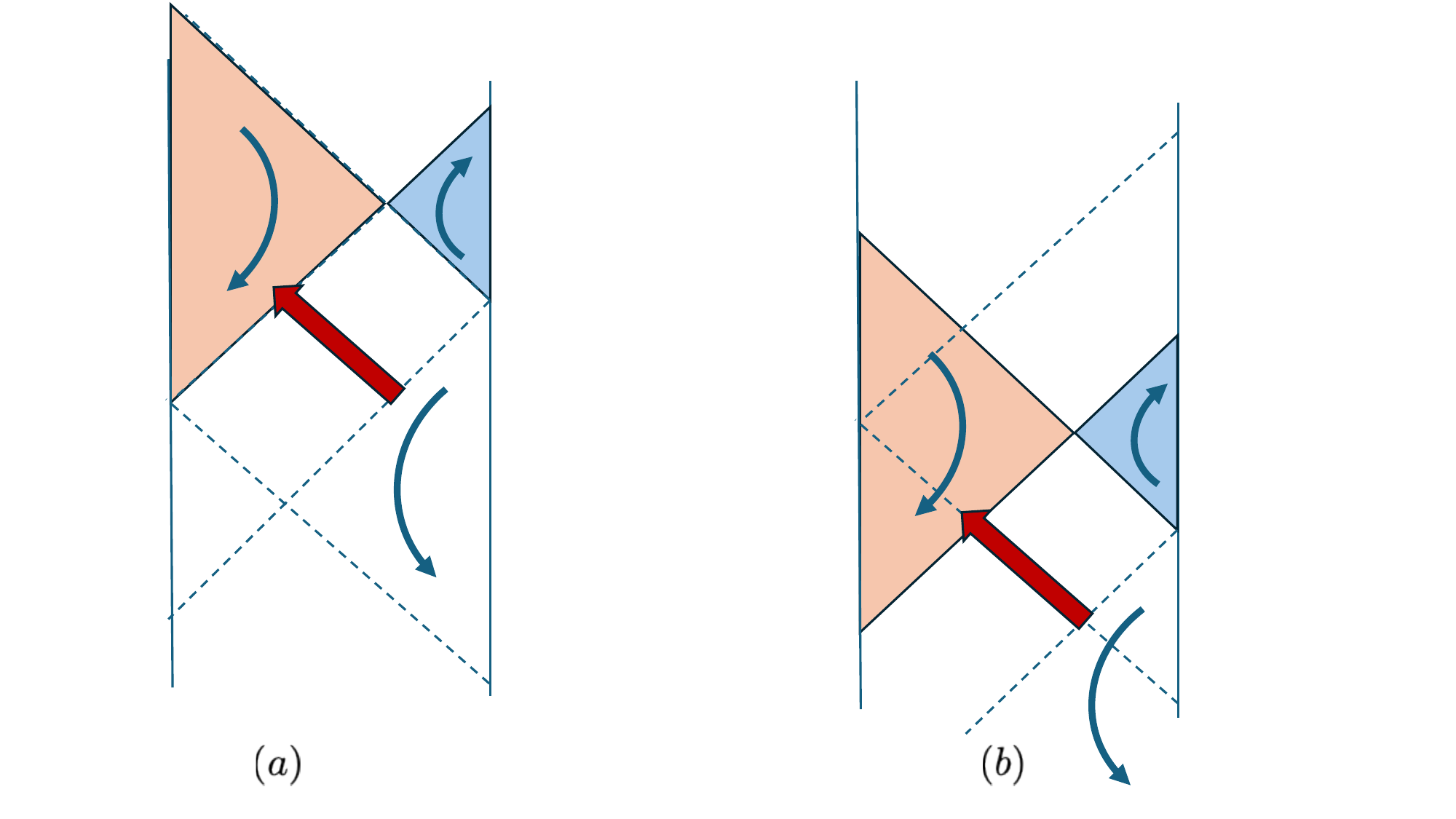}
    \caption{\small Modular flow  of time interval algebras in the bulk is found by a conformal transformation of that of half-infinite transformations.}
    \label{fig:modular-time-interval}
\end{figure}

Another key observation is that the modular flow acts geometrically on the bulk fields $\varphi(t_+,t_-)$ , for (at least) some finite range of $s$. 
\begin{lemma} \label{lemma-local-action}
    For any bulk field $\varphi(t_+,t_-)$ that is inside the bulk dual of the time interval algebra $\mathcal{A}_{(p,q)}$, the action of the modular flow is local and geometric:
    \begin{eqnarray}
        \sigma_s^{(p,q)}[\varphi(t_+,t_-)] = \varphi(t_+(s),t_-(s))\ , \quad \forall \, p < t_\pm < q\ , \quad \forall s\in \mathbb{R}
    \end{eqnarray}
    where 
    \begin{eqnarray}\label{modflowinterval}
        t_\pm(s) \equiv \frac{q(p-t_\pm)-e^{-2\pi s}p(q-t_\pm)}{(p-t_\pm)-e^{-2\pi s}(q-t_\pm)}\ .
    \end{eqnarray}
    Equivalently, using the conformal transformation that maps the interval $(p,q)$ to the interval $(0,\infty)$: 
    \begin{eqnarray}
        u=\frac{(t-p)}{(q-t)}
    \end{eqnarray}
we can write this modular flow as a dilatation map
    \begin{eqnarray}
u_\pm(s)\equiv \frac{t_\pm(s)-p}{q-t_\pm(s)}=e^{2\pi s}\frac{t_\pm-p}{q-t_\pm}=e^{2\pi s}u_\pm(0)\ .
    \end{eqnarray}
\end{lemma}
\begin{proof}
    This lemma follows immediately from the next Lemma \ref{lemma-restricted-local}.
\end{proof}
For any bulk field outside of the bulk dual of the algebra, the action is local on the boundary until the bulk field falls behind the Poincar\'e horizon:
\begin{lemma} \label{lemma-restricted-local}
For any bulk field $\varphi(t_+,t_-)$ that is outside the bulk dual of the algebra $\mA_{(p,q)}$, the modular flow is geometric and local with the same expression as in (\ref{modflowinterval}) for all $s_-<s<s_+$ for which $t_\pm(s)$ in \eqref{modflowinterval} is finite. At 
\begin{eqnarray}
s_\pm =\frac{1}{2\pi}\ln\lb \frac{q-t_\pm}{p-t_\pm}\rb
\end{eqnarray}
the operators reach the Poincar\'e horizon, and past that point the boundary modular flow is non-local.
\end{lemma}
\begin{proof}
    We consider the modular flow of $\mathcal{A}_{(p,q)}$ and a bulk field $\varphi(t_+,t_-)$. The smearing function $K(t|t_+,t_-)$ transforms as
    \begin{eqnarray}
        \sigma_s^{(p,q)}[K(t|t_+,t_-)] = \sum_\pm \Bigg[ e^{\mp i2\pi\Delta\theta(t,s)} \, K\left(t|t_+(s),t_-(s)\right)\Bigg]_\pm,
    \end{eqnarray}
    where 
    \begin{eqnarray}
        \theta(t,s) = \sgn{\left(s\right)} \Theta\left(-\frac{(q-t)e^{\pi s}+(t-p)e^{-\pi s}}{q-p}\right)
    \end{eqnarray}
    and $t_\pm(s)$ the same as \eqref{modflowinterval}. 
    We have used the fact that 
    \begin{eqnarray}
        K\left(\tfrac{(q-t)p\,e^{\pi s} + (t-p)q\,e^{-\pi s}}{(q-t)e^{\pi s} + (t-p)e^{-\pi s}} | t_+,t_-\right) = \left|\tfrac{(q-t)e^{\pi s}+(t-p)e^{-\pi s}}{q-p}\right|^{-2\Delta+2} K\left(t|t_+(s),t_-(s)\right).
    \end{eqnarray}
    As long as $\Theta\left(-\frac{(q-t)e^{\pi s}+(t-p)e^{-\pi s}}{q-p}\right) = 0$ for all $t_+(s)>t>t_-(s)$, the smearing function transforms locally. For any $p,q,t_+,t_-$, there exists $s_+>0>s_-$, such that we have $\Theta\left(-\frac{(q-t)e^{\pi s}+(t-p)e^{-\pi s}}{q-p}\right) = 0$ for all $s_+>s>s_-$ and $t_+(s)>t>t_-(s)$. This is just the connected interval of $s$ for which both $t_\pm(s)$ never reach $\pm\infty$. Without reaching infinity, this range of $s$ also guarantees that $t_+(s)>t_-(s)$.
\end{proof}
\begin{figure}[t]
    \centering
    \includegraphics[height=0.5\linewidth]{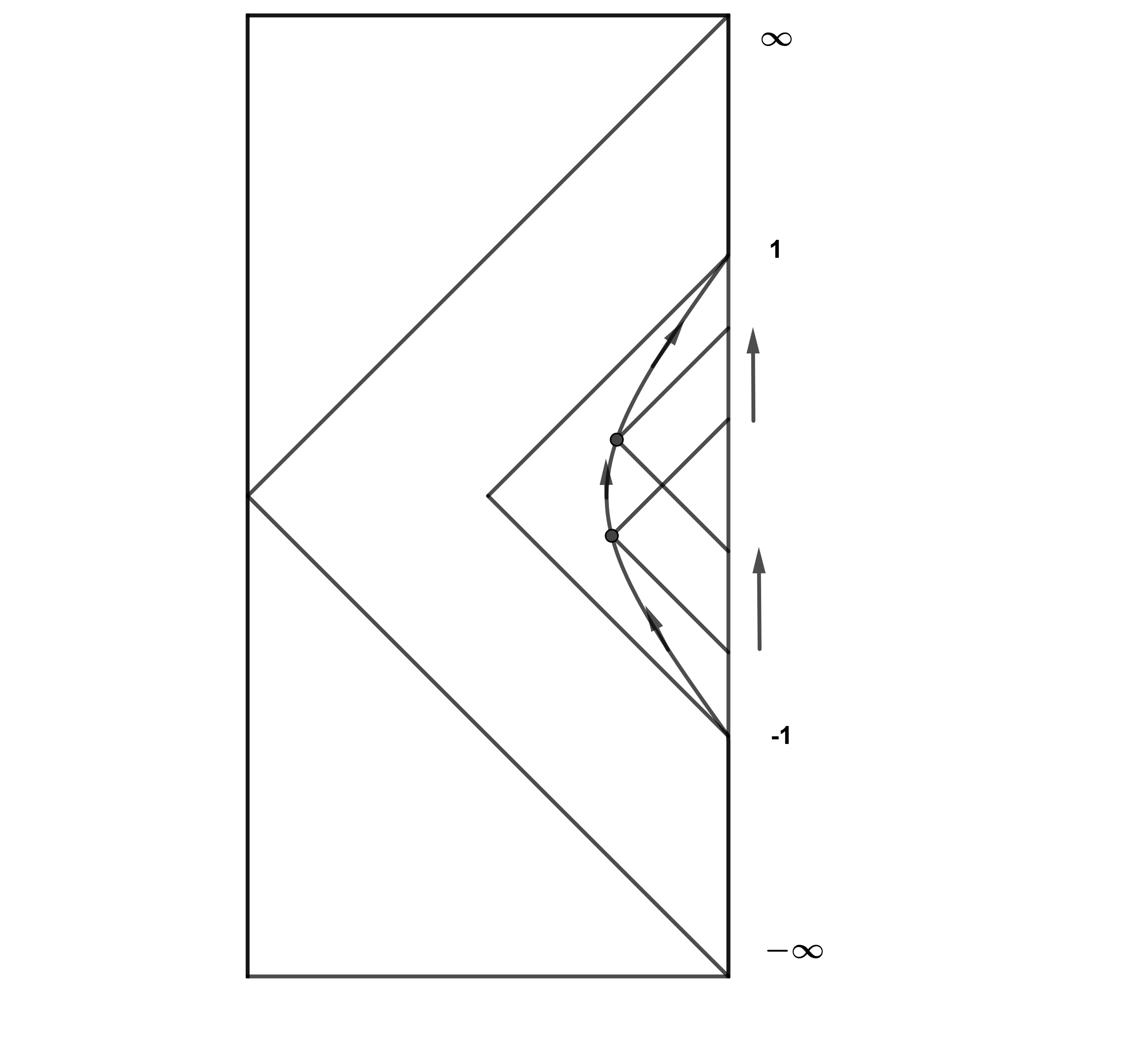}
    \caption{The geometric bulk action of the modular flow, with $p=-1$, $q=1$ as an example. The graph shows the AdS$_2$ spacetime with $\tau\in(-\pi,\pi)$. The bulk field moves along a trajectory in the bulk. This trajectory maintains a constant distance from the bulk point $(t_0,z_0) = (0, 1)$. We can also describe the trajectory with the action on the boundary points $t_\pm(s)$.}
    \label{fig:geometric_modular_flow}
\end{figure}

Once the modular flowed operator reaches the Poincar\'e horizon and past that point, the boundary modular flow is non-local, however, the bulk flow remains local. This is the result of the next lemma. An alternative way to understand is to recast the trajectory of bulk field  $(t(s),z(s))$ as the implicit expression
\begin{align}
    \frac{2t(s)t_0 + z(s)^2-t(s)^2+z_0^2-t_0^2}{2z(s)z_0} = \text{const}
\end{align}
where $q = t_0+z_0$, $p = t_0-z_0$. The left-hand side of the above expressions is the invariant distance between the bulk points $(t(s),z(s))$ and $(t_0,z_0)$, which is \eqref{eq-invariant-distance} in terms of Poincar\'e coordinate. The modular flow moves the bulk field along a trajectory while keeping the distance from $(t_0,z_0)$ fixed; see Figure \ref{fig:geometric_modular_flow}. If both $t_\pm \in (p,q)$, the action is geometric for all $s$, as expected. For other $t_\pm$, there exists a finite $s$ such that $t_+(s)$ or $t_-(s)$ reaches infinity. In the bulk, this corresponds to reaching the Poincar\'e horizon. Intriguingly, the failure of the boundary locality seems to be an indicator of the existence of a horizon. Once an operator reaches the horizon, we have to take care of the phase difference between the positive and negative modes. The expression of the trajectory also suggests that the action remains geometric while crossing the horizon, as we can extend it using the global coordinate $(\tau,\rho)$. Finally, we establish that the modular flow, indeed, acts geometrically for global AdS$_2$:
\begin{lemma} \label{lemma-geometric-bulk-action}
    For any bulk field $\varphi(\tau,\rho)$ in Poincar\'e AdS$_2$ and boundary time interval algebra $\mathcal{A}_{(p,q)}$, we have
    \begin{eqnarray}
        \sigma_s^{(p,q)}[\varphi(\tau,\rho)] = \varphi(\tau(s),\rho(s)), \qquad \forall s \in \mathbb{R}
    \end{eqnarray}
    where $\tau(s)$ and $\rho(s)$ are some continuous functions with $\tau(0) = \tau$, $\rho(0) = \rho$. The modular flow is always geometric and local in global $AdS_2$.
\end{lemma}
\begin{proof}
    We follow the HKLL reconstruction for the fields outside the Poincar\'e patch. For a bulk field outside the patch, we can write it in terms of the global coordinates
    \begin{eqnarray}
        \varphi(\tau,\rho) &=& \int_{-\infty}^\infty d\tau' K(\tau'|\tau,\rho) \varphi^{\text{global}}(\tau')\nn\\
        K(\tau'|\tau,\rho) &=& \frac{2^{\Delta-1}\Gamma(\Delta+1/2)}{\sqrt{\pi}\Gamma(\Delta)} \left(\frac{\cos(\tau'-\tau)-\sin\rho}{\cos\rho}\right)^{\Delta-1} \Theta\left(\frac{\pi}{2}-\rho-|\tau'-\tau|\right).
    \end{eqnarray}
    The global boundary field has the following property
    \begin{eqnarray}
        \varphi_\pm^{\text{global}}(\tau+2\pi) = e^{\mp i2\pi\Delta}\varphi_\pm^{\text{global}}(\tau)\ .
    \end{eqnarray}
    We first consider $\pi > \tau > \pi - (\frac{\pi}{2}-\rho) $. Let $\tau_- =\tau - (\frac{\pi}{2}-\rho)$, $\tau_+ =\tau + (\frac{\pi}{2}-\rho) - 2\pi$, $t_\pm = \tan\tau_\pm$. Note that with this shifted $\tau_\pm$, we have that $t_+ < t_-$. The bulk field can be expressed in terms of Poincar\'e boundary field as
    \begin{eqnarray}
        \varphi(\tau,\rho) &\propto& \int_{-\infty}^\infty d\tau' \left(\frac{\cos(\tau'-\tau)-\sin\rho}{\cos\rho}\right)^{\Delta-1} \Theta\lb \frac{\pi}{2}-\rho-|\tau'-\tau|\rb \varphi^{\text{global}}(\tau')\nonumber\\
        &=& \int_{\tau_-}^{\tau_+ + 2\pi} d\tau' \left(\frac{\cos(\tau'-\tau)-\sin\rho}{\cos\rho}\right)^{\Delta-1}  \varphi^{\text{global}}(\tau') \nonumber\\
        &=& \int_{\tau_-}^{\pi} d\tau' \left(\frac{\cos(\tau'-\tau)-\sin\rho}{\cos\rho}\right)^{\Delta-1}  \varphi^{\text{global}}(\tau') \nonumber\\
        &&+ \int_{-\pi}^{\tau_+} d\tau' \left(\frac{\cos(\tau'-\tau)-\sin\rho}{\cos\rho}\right)^{\Delta-1}  (e^{-i2\pi\Delta}\varphi_+^{\text{global}}(\tau') + e^{+i2\pi\Delta}\varphi_-^{\text{global}}(\tau')) \nonumber\\
        &=& \int_{t_-}^\infty dt \left(-\frac{(t-t_+)(t-t_-)}{t_+-t_-}\right)^{\Delta-1} \varphi(t) \nonumber\\
        &&+\int_{-\infty}^{t_+} dt \left(-\frac{(t-t_+)(t-t_-)}{t_+-t_-}\right)^{\Delta-1} (e^{-i2\pi\Delta}\varphi_+(t) + e^{+i2\pi\Delta}\varphi_-(t))\ .
    \end{eqnarray}
    The argument under the $\Delta-1$ power is still positive in these integral ranges because $t_+<t_-$ in this case. Thus, we can rewrite it as
    \begin{eqnarray}
        \varphi(\tau,\rho) = \int dt\, K(t|t_+,t_-) \lb e^{-i2\pi\Delta\Theta(t_+-t)}\varphi_+(t) + e^{+i2\pi\Delta\Theta(t_+-t)}\varphi_-(t)\rb\ .
    \end{eqnarray}
    For $-\pi <\tau < -\pi + (\frac{\pi}{2}-\rho) $, let $\tau_- =\tau - (\frac{\pi}{2}-\rho) + 2\pi$, $\tau_+ =\tau + (\frac{\pi}{2}-\rho)$, $t_\pm = \tan\tau_\pm$. We have 
    \begin{eqnarray}
        \varphi(\tau,\rho) = \int dt\, K(t|t_+,t_-) \lb e^{+i2\pi\Delta\Theta(t-t_-)}\varphi_+(t) + e^{-i2\pi\Delta\Theta(t-t_-)}\varphi_-(t)\rb\ .
    \end{eqnarray}
    For any $\tau \in (-\pi + 2n\pi, \pi+ 2n\pi)$ with integer $n$, we can shift $\tau$ by $2n\pi$ by including the extra phase $e^{\mp i2n\pi\Delta}$ for $\varphi_\pm(t)$. For any bulk field $\varphi(\tau + 2n\pi, \rho)$  with $\tau \in (-\pi,\pi)$, we have
    \begin{eqnarray}
        \varphi(\tau + 2n\pi, \rho) = \int dt\, K(t|t_+,t_-) \left( e^{-i2\pi\Delta n'} \varphi_+(t) + e^{+i2\pi\Delta n'} \varphi_-(t)\right)
    \end{eqnarray}
    where \footnote{To be rigorous, the expressions for $t_\pm$ are only valid when the argument of the step function is not exactly zero i.e., $(\tau - 2n\pi \pm \lb\frac{\pi}{2}-\rho)\rb \neq \pm \pi$, but we can easily handle the special cases by mapping it to the correct infinity.}
    \begin{eqnarray}
        n' &=& n + \Theta(t_--t_+)\left[\Theta(-(t_++t_-))\Theta(t_+-t)-\Theta(t_++t_-)\Theta(t-t_-)\right]\nn \\
        t_+ &=& \tan \left[ \tau - 2n\pi + \lb\frac{\pi}{2}-\rho\rb  - 2\pi \Theta\left((\tau - 2n\pi + (\frac{\pi}{2}-\rho)) -\pi\right)\right]\nn\\
        t_- &=& \tan \left[ \tau - 2n\pi - \lb\frac{\pi}{2}-\rho\rb  + 2\pi \Theta\left(-\pi - (\tau - 2n\pi - (\frac{\pi}{2}-\rho) )\right)\right]\ .
    \end{eqnarray}
    This example looks complicated, but the phase in fact matches the phase from the modular flow. We can see this by simply counting the moment $t_\pm(s)$ passing through the infinities. The constant $n$ does not matter as it only indicates which section in the bulk. The step functions capture whether the bulk field crosses the future horizon or the past horizon. This corresponds to the step function in the modular flow that captures whether the boundary point crosses the infinity with forward modular time or backward modular time. The modular flow also swaps $t_+(s) > t_-(s)$ to $t_-(s) > t_+(s)$ while maintaining the same expression of the constant distance trajectory. Thus, the modular flow indeed moves bulk field across the Poincar\'e horizon geometrically. If $\Delta$ is an integer, there is no phase difference between the positive and negative modes. In this case, it is straightforward to derive and see the geometric action.
\end{proof}

An interesting corollary is that the modular flow equations we derived above are equivalent to the HKLL reconstruction:
\begin{lemma} [Modular Zero Mode] \label{lemma-zero-mode}
    The bulk field $\varphi(t_+,t_-)$ is invariant under the modular flow $\sigma_s^{(t_-,t_+)}$ of the time interval algebra $\mathcal{A}_{(t_-,t_+)}$ i.e.,
    \begin{eqnarray}
        \sigma_s^{(t_-,t_+)}[\varphi(t_+,t_-)] = \varphi(t_+,t_-)\ .
    \end{eqnarray}
\end{lemma}
\begin{proof}
    This lemma follows from Lemma \ref{lemma-local-action}. Instead, we can also derive the bulk-field smearing function solely with the knowledge of the boundary modular flow. Consider a distribution $f$ that has support only in the interval $(p,q)$ and the corresponding affiliated operator is invariant under the modular flow of $\mathcal{A}_{(p,q)}$. We have
\begin{eqnarray}
    f &=&  \sigma_s^{(t_-,t_+)}[f] \nn\\
    &=& \sum_\pm \Bigg[\left(\tfrac{(q-t)e^{\pi s}+(t-p)e^{-\pi s}}{q-p} \mp\sgn(s)i\epsilon\right)^{2\Delta-2}\, f\left(\tfrac{(q-t)p\,e^{\pi s} + (t-p)q\,e^{-\pi s}}{(q-t)e^{\pi s} + (t-p)e^{-\pi s}}\right)\Bigg]_\pm \nn\\
    &=& \left(\tfrac{(q-t)e^{\pi s}+(t-p)e^{-\pi s}}{q-p}\right)^{2\Delta-2} \, f\left(\tfrac{(q-t)p\,e^{\pi s} + (t-p)q\,e^{-\pi s}}{(q-t)e^{\pi s} + (t-p)e^{-\pi s}}\right)\ .
\end{eqnarray}
Put $t = (p+q)/2$, we have
\begin{eqnarray}
    f\left(\frac{p+q}{2}\right) = \cosh(\pi s)^{2\Delta-2}f\left(\frac{pe^{\pi s}+qe^{-\pi s}}{e^{\pi s}+e^{-\pi s}}\right)\ .
\end{eqnarray}
Let $t' = \frac{pe^{\pi s}+qe^{-\pi s}}{e^{\pi s}+e^{-\pi s}}$, we have $e^{2\pi s} = \frac{q-t'}{t'-p}$ and
\begin{eqnarray}
    f(t') &=& f\left(\frac{q+p}{2}\right) \cosh(\pi s)^{-2\Delta+2} \nn\\
    &=& f\left(\frac{q+p}{2}\right) \left(\frac{4}{q-p}\right)^{\Delta-1}\left( \frac{(q-t')(t'-p)}{q-p}\right)^{\Delta-1}\ .
\end{eqnarray}
By putting back the step function and the normalization, it is the smearing function of the bulk field $\varphi(p,q)$. It is a modular zero mode by construction.
\end{proof}

 Passing through the horizon is a non-local transformation on the boundary. 
No time interval algebra includes the operator behind the horizon, and all time interval algebras are strict subalgebras of the whole algebra. The casual depth parameter is infinite and matches the idea of \cite{Gesteau:2024rpt}. This raises a question about the converse of our results.
 We ask under what conditions a theory that admits time interval algebra has an emergent bulk as:
\begin{enumerate}
    \item The affiliated zero modes of a time interval algebra can be interpreted as the bulk field.
    \item The modular flow of a time interval algebra acts geometrically on all these zero modes.
    \item The non-local action of the modular flows implies the existence of some bulk horizon.
\end{enumerate}
We will answer the question above in the next section.

\section{Conformal algebra, twisted inclusions and modular intersections}\label{sec:Twisted}

Up to now, we have focused on conformal GFF in $0+1$-dimensions as a theory with future algebras $\mA_{\mathbb{R}_+}$ and time interval algebras $\mA_I$. Then, using the action of the $\mathfrak{psl}(2,\mathbb{R})$ Lie algebra we derived expressions for the modular data (flow and conjugation) of $\mA_{\mathbb{R}_+}$ and $\mA_I$. We saw that the geometric action of $\mathfrak{psl}(2,\mathbb{R})$ on global AdS$_2$ resulted in local modular actions in the bulk. In this section, we prove converse results for any theory with time interval algebras.
First, in Section \ref{conftoModInt}, we show that the local modular action we derived here using the $\mathfrak{psl}(2,\mathbb{R})$ algebra results in a {\it Modular Intersection} property. Then, in Section \ref{ModInttoConf}, we show a converse that the existence of Modular Intersection property for any theory with future algebras is enough to deduce the existence of the $\mathfrak{psl}(2,\mathbb{R})$ algebra.

\subsection{Conformal algebra implies modular intersections}\label{conftoModInt}

Let us first focus on the integer $\Delta$ case and consider the algebra of half-line, $\mN = \mA_{(0,\infty)}$. We showed in Theorem  \ref{ModFlowThmInsideAlg} that the modular flow of this algebra is dilation. This means that 
\begin{eqnarray}
    \Delta_{\mN}^{-it} \mA_{(0,1)} \Delta_{\mN}^{it} \subset \mA_{(0,1)} \qquad \text{for all } t \le 0 \ . \label{eq-inter-prop-1}
\end{eqnarray}
In other words, the inclusion $\mA_{(0,1)} \subset \mN$ is a past half-sided modular inclusion (HSMI$-$); see Appendix \ref{App:Hsmi} for definitions. Now, consider the algebra of an interval $\mM = \mA_{(-1,1)}$. In Lemma \ref{lemma-interval-map}, we saw that this algebra can be mapped to that of a half-line by the unitary action of \eqref{conformalUT}. Under this map its subalgebra $\mA_{(0,1)}$ goes to $\mA_{(1,\infty)}$, and we deduce that\footnote{This follows from the fact that if $\mA_1\subset \mA_2$ is an HSMI$\pm$ inclusion and $U$ is a unitary that leaves $\ket{\Omega}$ invariant, the inclusion $U\mA_1U^\dagger\subset U\mA_2U^\dagger$ is also HSMI$\pm$.}
\begin{eqnarray}
    \Delta_{\mM}^{-it} \mA_{(0,1)} \Delta_{\mM}^{it} \subset \mA_{(0,1)} \qquad \text{for all } t \ge 0 \ , \label{eq-inter-prop-2}
\end{eqnarray}
and hence, the inclusion $\mA_{(0,1)} \subset \mM$ is a future half-sided modular inclusion (HSMI$+$). Finally, we note that for integer $\Delta$, the modular conjugation of $\mN = \mA_{(0,\infty)}$ is the reflection map ($t \to -t$), and hence,
\begin{align}
    J_{\mN} \mM J_{\mN} = \mM \ . \label{eq-inter-prop-3}
\end{align}
We can summarize the above discussion in terms of the following corollary:
\begin{corollary} \label{mod-intersection-GFF}
    In $0+1$-dimensional conformal GFF dual to Poincar\'e AdS$_2$ with integer $\Delta$, the von Neumann algebras $\mM=\mA_{(-1,1)}$ and $\mN=\mA_{(0,\infty)}$ have the following properties
\begin{enumerate}
    \item $(\mN\cap \mM)\subset \mN$ is HSMI$-$
    \item $(\mN\cap \mM)\subset \mM$ is HMSI$+$
    \item $J_\mN \mM J_\mN=\mM$. 
\end{enumerate}
\end{corollary}
This motivates us to propose the following definition:
\begin{definition}(Modular Intersection)
We say a tuple $(\mM, \mN)$ of von Neumann algebras satisfies the {\it modular intersection} property if $\mM$ and $\mN$  have a common cyclic and separating state and the following properties hold
\begin{enumerate}
    \item $(\mN\cap \mM)\subset \mN$ is HSMI$-$
    \item $(\mN\cap \mM)\subset \mM$ is HMSI$+$
    \item $J_\mN \mM J_\mN=\mM$.
\end{enumerate}
\end{definition}

Now we focus on the case of non-integer $\Delta$. Even though \eqref{eq-inter-prop-1} and \eqref{eq-inter-prop-2} are still valid for the general $\Delta$, \eqref{eq-inter-prop-3} is no longer valid. This is because the modular conjugation of the algebra of a half-line is the reflection followed by GHT $\mT_\Delta$. Therefore, the modular intersection property does not hold for $\mM=\mA_{(-1,1)}$ and $\mN=\mA_{(0,\infty)}$ for general $\Delta$. However, the generalization of \eqref{eq-inter-prop-3} for arbitrary $\Delta$ is
\begin{align}
    J_{\mN} \mM J_{\mN} \, = \, \mT^{\dagger}_{\Delta} \mM \mT_{\Delta} 
\end{align}
which follows from Theorem \ref{thm-J_half-line}. This motivates the following generalization of the modular intersection property:
\begin{definition}(Twisted modular intersection)
We say a tuple $(\mM, \mN, U)$ of von Neumann algebras $\mM$ and $\mN$ and a unitary $U$ satisfies the {\it twisted modular intersection} property if $\mM$ and $\mN$ have a common cyclic and separating state and the following properties hold:
\begin{enumerate}
    \item $(\mN\cap \mM)\subset \mN$ is HSMI$-$
    \item $(\mN\cap \mM)\subset \mM$ is HSMI$+$
    \item $U$ commutes with $\Delta_\mM^{is}$ and 
   \begin{eqnarray}
        J_\mN \mM J_\mN = U \mM U^\dagger \, .\label{eq-Jn-twisted-intersection-prop}
   \end{eqnarray}
\end{enumerate}
\end{definition}
The discussion of GFF with non-integer $\Delta$ can be summarized with the following result:
\begin{corollary} \label{twisted-intersection-GFF}
    In $0+1$-dimensional conformal GFF dual to Poincar\'e AdS$_2$ with $\Delta > 1/2$, $(\mM, \mN, \mT^{\dagger}_{\Delta})$ satisfies the twisted modular intersection property, where $\mM=\mA_{(-1,1)}$ and $\mN=\mA_{(0,\infty)}$. 
\end{corollary}

The Corollaries \ref{mod-intersection-GFF} and \ref{twisted-intersection-GFF} are results of the $\mathfrak{psl}(2,\mathbb{R})$ Lie algebra. In fact, we show below that any (twisted) modular intersection property implies the existence of a unitary representation of the $PSL(2,\mathbb{R})$ ($\widetilde{PSL}(2,\mbR)$) group. 

Before we do that, let us discuss another crucial difference between the conformal GFF with integer $\Delta$ and non-integer $\Delta$. Consider the inclusion $\mA_{(1,\infty)} \subset \mA_{(0,\infty)}$. When $\Delta$ is an integer, Haag's duality holds and hence, the relative commutant of $\mA_{(1,\infty)}$ in $\mA_{(0,\infty)}$ is $\mA_{(0,1)}$ which is cyclic (see the Reeh-Schlieder property of time interval algebras in Lemma 23 of \cite{Furuya:2023fei}).  This motivates the following definition (see also Appendix \ref{app:StandardNormal}):
\begin{definition}[Standard inclusion]
An inclusion $\mN \subset \mM$ of von Neumann algebras is called a {\it standard inclusion} if $\mN$, $\mM$, and the relative commutant $\mN' \cap \mM$ share a common cyclic and separating state. 
\end{definition}
Therefore, for integer $\Delta$ the inclusion $\mA_{(1,\infty)} \subset \mA_{(0,\infty)}$ is standard.
On the other hand, for non-integer $\Delta$ the commutant of $\mA_{(1,\infty)}$ is $\mT_{\Delta}^\dagger \mA_{(-\infty,1)} \mT_{\Delta}$ (see Theorem \ref{thm-comm-I-general}).
It is instructive to understand this result in the bulk AdS$_{2}$ picture. As shown in Figure \ref{fig:commutantbulk}, the relative commutant of $\mA_{(1,\infty)}$ in $\mA_{(0,\infty)}$ in the bulk spacetime corresponds to the null segment of the future Poincar\'e horizon. While in 1+1-dimensional quantum field theory, the algebra of null segment is non-trivial,\footnote{For example, it includes operators generated by the null components of the stress tensor.} the resulting subalgebra is not cyclic.\footnote{Note that for Reeh-Schlieder to apply, we need a region of spacetime with non-zero volume.} This provides a bulk reasoning for why for non-integer $\Delta$ the modular inclusion $\mN\cap \mM\subset \mM$ is not standard.

\begin{figure}[h]
   \centering    \includegraphics[width=\textwidth]{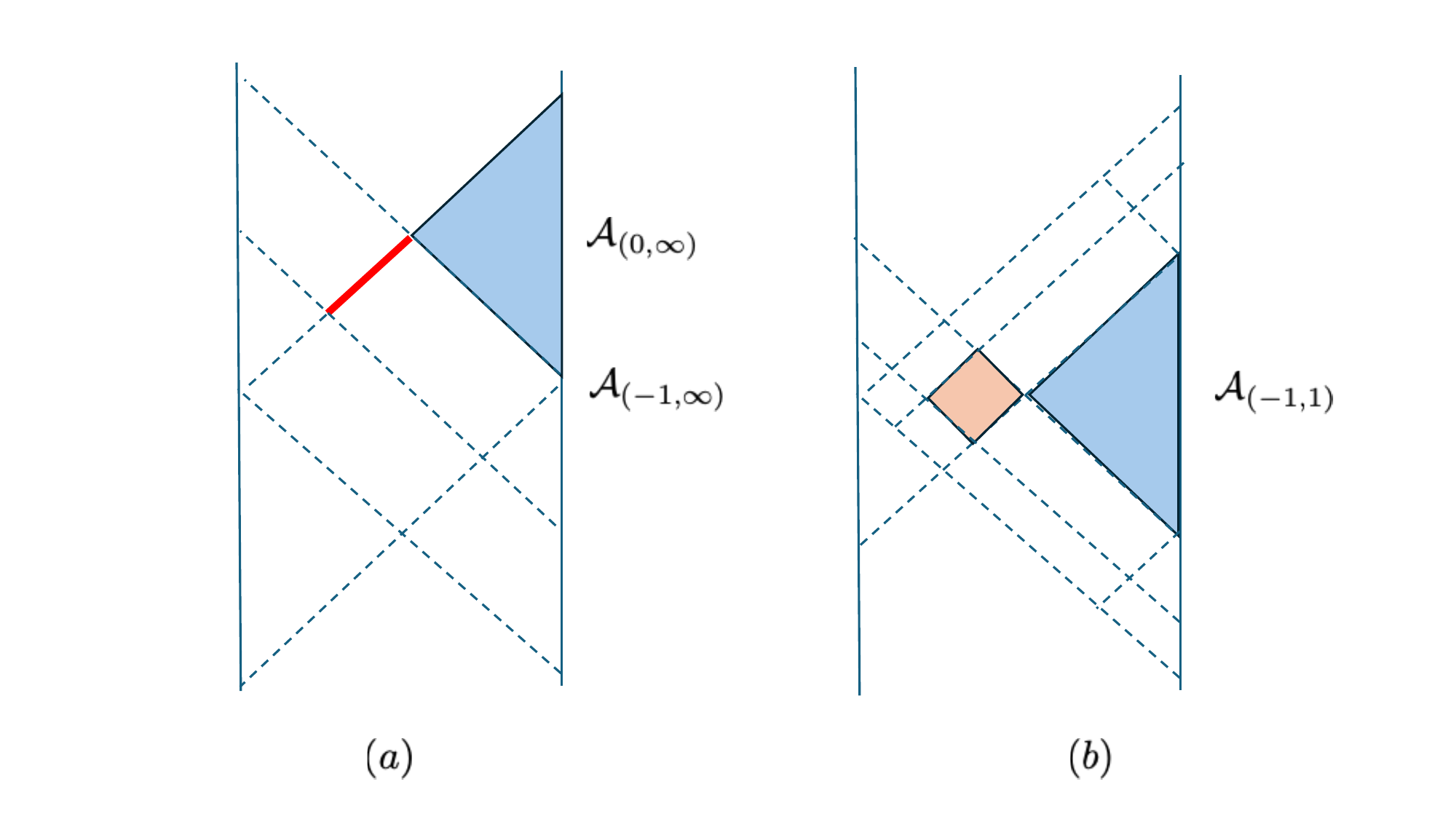}
    \caption{\small (a) The relative commutant of $\mA_{(0,\infty)}\subset \mA_{(-1,\infty)}$ is the algebra of the null segment in red in the bulk. (b) The relative commutant of $\mA_{(-1,1)}$ in $\mA_{(-2,2)}$ is the algebra of the orange region in the bulk which has finite volume, and is hence cyclic.  }
    \label{fig:commutantbulk}
\end{figure}

For non-integer $\Delta$, we have a non-standard inclusion. However, as can be seen from the geometric picture in Figure \ref{fig:modular-future-past}, the algebra
$\mT_{\Delta} (\mA_{(1,\infty)})' \mT_{\Delta}^{\dagger}$
overlaps with $\mA_{(0,\infty)}$ and the vacuum state is cyclic and separating for this overlapping algebra. This motivates defining a \textit{twisted inclusion} as a generalization of standard inclusion:
\begin{definition}[Twisted inclusion]
We call a tuple $(\mN \subset \mM, U)$ of an inclusion $\mN \subset \mM$ of von Neumann algebra and a unitary $U$ a {\it twisted inclusion} if $U$ leaves $\ket{\Omega}$ invariant, commutes with both $\Delta_{\mM}$ and $\Delta_{\mN}$, and the state $\ket{\Omega}$ is a common cyclic and separating state for $\mN$, $\mM$, and $(U^\dagger \mN' U)\cap\mM$. 
\end{definition}
With this definition, we have
\begin{corollary}
    In a theory of conformal GFF for non-integer $\Delta$, $(\mA_{(1,\infty)} \subset \mA_{(0,\infty)}, \mT_{\Delta}^\dagger)$ is a twisted inclusion. 
\end{corollary}

In the next subsection, we will demonstrate the connection between the representation of $\widetilde{PSL}(2,\mbR)$ and the twisted inclusion. In particular, we will show that any twisted half-sided inclusion leads to a representation of $\widetilde{PSL}(2,\mbR)$ group.

\subsection{Twisted modular intersection implies conformal algebra}\label{ModInttoConf}

Classic result in algebraic QFT by Wiesbrock states the following converse statements: 
\begin{enumerate}
     \item {\bf Modular intersection theorem \cite{wiesbrock1993conformal}:} Modular intersection property implies a unitary representation of the $PSL(2,\mathbb{R})$ group.
    \item {\bf Standard inclusion theorem \cite{wiesbrock1993symmetries}:} Standard half-sided modular inclusion implies a unitary representation of the $PSL(2,\mathbb{R})$ group. 
\end{enumerate}
For integer $\Delta$, these results imply the existence of a unitary representation of $PSL(2,\mathbb{R})$ in accordance with Lemma \ref{lemma:psl2r_rep}. We saw that insisting on a local representation of modular data (modular conjugation and modular flow) geometrizes the $PSL(2,\mathbb{R})$ action as a bulk AdS$_2$. These results explain the emergence of free bulk fields in Poincar\'e AdS$_2$ from GFF fields with integer $\Delta$. 

For conformal GFF with non-integer $\Delta$, neither the modular intersection property nor the standard inclusion holds. To explain the emergence of AdS$_2$ for non-integer $\Delta$, we need new results that we prove below.
As we showed in the previous subsection, for non-integer $\Delta$ we have the twisted modular intersection property and the twisted inclusion. Here, we generalize Wiesbrock's results by showing: 
\begin{enumerate}
    \item {\bf Theorem \ref{twistedModInt}:} Twisted modular intersection property implies a unitary representation of the universal cover of $PSL(2,\mathbb{R})$ group: $\widetilde{PSL}(2,\mbR)$.
    \item {\bf Theorem \ref{twistModularInclusion}:} Twisted half-sided modular inclusion implies a unitary representation of the universal cover of $PSL(2,\mathbb{R})$ group: $\widetilde{PSL}(2,\mbR)$.
\end{enumerate}
Once again, insisting on a local representation of modular flow explains the emergence of AdS$_2$. 
Note that, since these theorems do not assume the existence of any particular spectral density, or even GFF algebras, they can be applied more generally to GFF fields with non-holographic spectral densities, or potentially more general theories with modular future and past subalgebras. For this reason, we call this emergent geometric local representation AdS$_2$.

Instead of reviewing Wiesbrock's theorems, we will only discuss the proofs of the twisted statements because for the case $U=1$ they reduce to Wiesbrock's original theorems.

\begin{theorem}[Twisted Wiesbrock's Theorem ]\label{GenWiesbrockThm}
    Suppose we have von Neumann algebras $\mM, \mM_1, \mM_2$ with common cyclic and separating state $\Omega$ and a unitary $U$ that leaves $\Omega$ invariant and 
    \begin{enumerate}
        \item $(\mM_1 \subset \mM, \Omega)$ is HSMI$+$,
        \item $(\mM_2 \subset \mM, \Omega)$ is HSMI$-$,
        \item $(\widetilde{\mM_2} \equiv U\mM_2 U^\dagger \subset \mM_1', \Omega)$ is HSMI$+$,
        \item $U\Delta_\mM^{is}U^\dagger =\Delta_\mM^{is}$.
    \end{enumerate}
    Then, the operators 
    \begin{eqnarray}
        G_0 &=& \frac{1}{2\pi}\log \Delta_\mM \nn\\
        G_1 &=& \frac{1}{2\pi}(\log \Delta_{\mM_1}- \log \Delta_{\mM}) \nn\\
        \widetilde{G_2} &=& \frac{1}{2\pi}(\log \Delta_{\widetilde{\mM_2}}- \log \Delta_{\mM})
    \end{eqnarray}
    satisfy the $\mathfrak{psl}(2,\mathbb{R})$ Lie algebra:
        \begin{eqnarray}
        [G_0, G_1] = iG_1\ , \qquad [G_0, \widetilde{G_2}] = -i\widetilde{G_2}\ , \qquad [G_1, \widetilde{G_2}] = -2iG_0\ . 
    \end{eqnarray}
    They exponentiate to $e^{it G_0}$, $e^{is G_1}$, and $e^{ir \widetilde{G_2}}$ furnishing a representation of $\widetilde{PSL}(2,\mbR)$ group. In this representation, the inversion map $\mathcal{I}$ is given by 
    \begin{eqnarray}\label{inversionmodular}
        \mathcal{I}  =  J_{\mM}J_{\mM_{1}} \Delta_{\widetilde{\mM_2} }^{-i\frac{\log (2)}{2\pi}} \Delta_{\mM_1}^{-i\frac{\log (2)}{2\pi}}   =  J_{\mM} \left( \Delta_{\mM_1}^{-i\frac{\log(2)}{2\pi}} J_{\widetilde{\mM_2}} \Delta_{\mM_1}^{i\frac{\log(2)}{2\pi}} \right) \ . 
    \end{eqnarray}
\end{theorem}
\begin{proof}
Recall that according to the half-sided modular inclusion theorem (see Appendix \ref{App:Hsmi} for a review), if $\mN \subset\mM$ is HSMI$\pm$, then 
\begin{eqnarray}
    \left[\frac{1}{2\pi}\log \Delta_\mM,\frac{1}{2\pi}(\log \Delta_{\mN}- \log \Delta_{\mM})\right]= \pm i \frac{1}{2\pi}(\log \Delta_{\mN}- \log \Delta_{\mM})\ .
\end{eqnarray}
Therefore, conditions 1 and 2 imply that $[G_{0},G_{1}] = iG_{1}$ and  $[G_{0},G_{2}] = -iG_{2}$ where
\begin{eqnarray}
    G_{2} = \frac{1}{2\pi}(\log \Delta_{\mM_2}- \log \Delta_{\mM}) \ .
\end{eqnarray}
Now by definition $\widetilde{\mM_2} \equiv U\mM_2 U^\dagger$ we have that $\log \Delta_{\widetilde{\mM_2}} = U \log \Delta_{\mM_2} U^\dagger$ and from condition 4 that $\log \Delta_{\mM} = U \log \Delta_{\mM} U^\dagger$. Hence, we can write $\widetilde{G_2} = U G_2 U^\dagger$ and $G_0 = U G_0 U^\dagger$. As a result, we get
\begin{eqnarray}
    [G_0, \widetilde{G_2}] = U [G_0, G_2] U^\dagger = U(-iG_2)U^\dagger = -i\widetilde{G_2}\ .
\end{eqnarray} 
Finally, condition 3 implies
\begin{eqnarray}
    \left[\frac{1}{2\pi}\log \Delta_{\mM_1'},\frac{1}{2\pi}\log \Delta_{\widetilde{\mM_2}}\right] = i \frac{1}{2\pi}(\log \Delta_{\widetilde{\mM_2}}- \log \Delta_{\mM'_1}) 
\end{eqnarray}
    or equivalently
\begin{eqnarray}
    \left[\frac{1}{2\pi}\log \Delta_{\mM_1},\frac{1}{2\pi}\log \Delta_{\widetilde{\mM_2}}\right] = -i \frac{1}{2\pi}(\log \Delta_{\mM_1} + \log \Delta_{\widetilde{\mM_2}})\ .
\end{eqnarray}
Combining this with the first two commutation relations yields the third commutation relation
\begin{eqnarray}
    [G_1, \widetilde{G_2}] = -2iG_0\ .
\end{eqnarray}
This proves that $G_{0}$, $G_{1}$, and $\widetilde{G_2}$ satisfy the $\mathfrak{psl}(2,\mathbb{R})$ Lie algebra, and hence, the unitaries $V_{0}(a) = e^{iaG_{0}}$, $V_{1}(a) = e^{iaG_{1}}$, and $\widetilde{V_{2}}(a) = e^{ia\widetilde{G_{2}}}$ generate the $\widetilde{PSL}(2,\mbR)$ group. 

It follows from the $\mathfrak{sl}(2,\mathbb{R})$ Lie algebra that the representation of the inversion map is \cite{wiesbrock1993symmetries}
\begin{eqnarray}
    \mathcal{I}  =  V_{1}(-2) \widetilde{V_{2}}(-1/2) V_{0}(-2\log2) V_{1}(-1/2) \ .
\end{eqnarray}
Using \eqref{eq-hsmi-4}, we get $V_{1}(-2) = J_{\mM}J_{\mM_{1}}$ whereas using \eqref{eq-hsmi-5}, we get
\begin{eqnarray}
 V_{1}(-1/2) = \Delta_{\mM}^{i\frac{\log (2)}{2\pi}} \Delta_{\mM_1}^{-i\frac{\log (2)}{2\pi}} \nn\\
 V_{2}(-1/2) = \Delta_{\mM_2}^{-i\frac{\log (2)}{2\pi}} \Delta_{\mM}^{i\frac{\log (2)}{2\pi}} 
\end{eqnarray}
where $V_2(a) = e^{ia G_2}.$ Since $U$ commutes with $V_0(a) = \Delta_{\mM}^{ia/2\pi}$, we get
\begin{eqnarray}
\widetilde{V_{2}}(-1/2) = \Delta_{\widetilde{\mM_2}}^{-i\frac{\log (2)}{2\pi}} \Delta_{\mM}^{i\frac{\log (2)}{2\pi}} \ .
\end{eqnarray}
With these results the inversion map becomes
\begin{eqnarray}
    \mathcal{I} &=&  J_{\mM}J_{\mM_{1}} \Delta_{\widetilde{\mM_2} }^{-i\frac{\log (2)}{2\pi}} \Delta_{\mM_1}^{-i\frac{\log (2)}{2\pi}}  \ .
\end{eqnarray}
To further simplify this expression, we note that since $\widetilde{\mM_2} \subset \mM'_1$ is HSMI$+$, there exists a unitary $V_{3}(a) = \exp[i\frac{a}{2\pi}(\log\Delta_{\widetilde{\mM_2}}+\log\Delta_{\mM_1})] $. Using properties of this unitary, we get
\begin{eqnarray}
    \mathcal{I} &=&  J_{\mM}J_{\mM_{1}} V_{3}(-1)  \nn\\ 
    &=&  J_{\mM}J_{\mM_{1}} \Delta_{\mM_1}^{-i\frac{\log (2)}{2\pi}} V_{3}(-2) \Delta_{\mM_1}^{i\frac{\log (2)}{2\pi}} \nn\\ 
    &=&  J_{\mM} \Delta_{\mM_1}^{-i\frac{\log (2)}{2\pi}} J_{\widetilde{\mM_{2}}}  \Delta_{\mM_1}^{i\frac{\log (2)}{2\pi}} 
\end{eqnarray}
where we have used \eqref{eq-hsmi-5}, \eqref{eq-hsmi-1}, and \eqref{eq-hsmi-4} in the first, second, and third steps, respectively.

\end{proof}
Note that the theorem above is a generalization of the Lemmas $3$ and $4$ of Wiesbrock \cite{wiesbrock1993symmetries}.
Next, we prove the first main result of this section:  
\begin{theorem}[Twisted modular intersection]\label{twistedModInt}
      Suppose ($\mM$, $\mN$, $U$) satisfies the twisted modular intersection property, then we have a representation of $\widetilde{PSL}(2,\mbR)$ group, and the inversion map is given by $J_{\mM}J_{\mN}$. 
\end{theorem}
\begin{figure}[t]
    \centering
    \includegraphics[width=\linewidth]{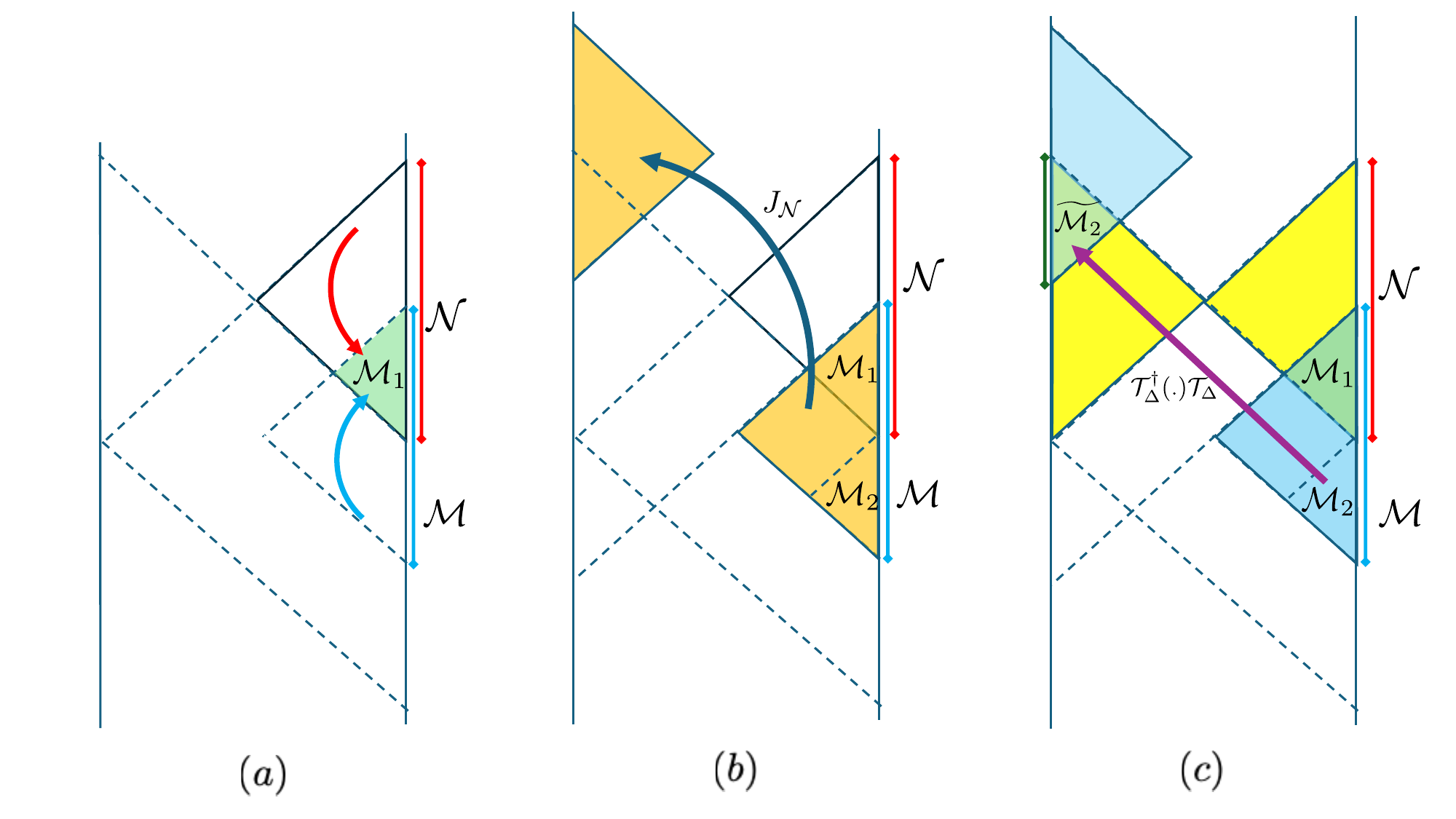}
    \caption{\small Twisted modular intersection property for 0+1-dimensional GFF dual to bulk AdS$_2$ for the unitary $U = \mT_\Delta^\dagger$ (see Theorem \ref{GenWiesbrockThm} and \ref{twistedModInt}). (a) $\mM \cap \mN \subset \mM$ is HSMI$+$ and $\mM \cap \mN \subset \mN$ is HSMI$-$. (b) $J_\mN$ maps $\mM$ to the corresponding antipodal region in the bulk $\mT_\Delta^\dagger \mM \mT_\Delta$. (c) The antipodal map takes $\mM_2$ to the region $J_\mN \mM_1 J_\mN$ showing that $\mM_2 \subset \mM$ is HSMI$-$ and $\widetilde{\mM_2} \subset \mM_1'$ is HSMI$+$. }
    \label{fig:twisted_intersection}
\end{figure}
\begin{proof}
    To prove this theorem, we need to show that if we define $\mM_{1} \equiv (\mM \cap \mN)$ and $\mM_2 \equiv \mM \cap (U^{\dagger} \mN' U)$, then the unitary $U$ and algebras $\mM$, $\mM_1$, and $\mM_2$ satisfy conditions $1$-$4$ of Theorem \ref{GenWiesbrockThm}; see Figure \ref{fig:twisted_intersection}. 
    By definition of the twisted modular intersection, we have $\mM_1 \subset\mM$ is HSMI$+$ and that the unitary $U$ commutes with $\Delta_{\mM}$. Therefore conditions 1 and 4 of Theorem \ref{GenWiesbrockThm} are satisfied. To show condition 2, we note that 
    \begin{eqnarray}
        \widetilde{\mM_{2}} = U \mM_2 U^\dagger = (U\mM U^\dagger) \cap \mN' = J_{\mN} (\mM\cap \mN) J_{\mN} = J_{\mN} \mM_{1} J_{\mN}  \label{eq-m2-m1-jn}
    \end{eqnarray}
    where we have used \eqref{eq-Jn-twisted-intersection-prop}. Moreover, \eqref{eq-Jn-twisted-intersection-prop} also implies that 
    \begin{eqnarray}
        J_{\mN} \Delta_{\mM}^{it} J_{\mN} = U \Delta_{\mM}^{-it} U^\dagger = \Delta_{\mM}^{-it} \ .
    \end{eqnarray}
    Using the above two results, we get
    \begin{eqnarray}
        \Delta_{\mM}^{-it} \widetilde{\mM_2} \Delta_{\mM}^{it} = J_{\mN} \left( \Delta_{\mM}^{it} \mM_1 \Delta_{\mM}^{-it} \right) J_{\mN} 
        \subset  J_{\mN} \mM_{1} J_{\mN} 
        = \widetilde{\mM_2}  \qquad \text{for all }  t \le 0\ .
    \end{eqnarray}
    Thus, $\mM_{2} \subset \mM$ is HSMI$-$ and condition $2$ is also satisfied.
    In order to see condition 3, we note that $\widetilde{\mM_{2}} \subset \mM'_{1}$ since
    \begin{eqnarray}
        \widetilde{\mM_{2}}  =  (U\mM U^\dagger) \cap \mN' \subset \mN'  \subset \mM'_{1} \ .
    \end{eqnarray}
    Now we need to show that $\widetilde{\mM_{2}} \subset \mM'_{1}$ is in fact HSMI$+$. 
    From \eqref{eq-m2-m1-jn}, we get
    \begin{eqnarray}
        \Delta_{\mM'_{1}}^{-it} \widetilde{\mM_2} \Delta_{\mM'_{1}}^{it}  =  \Delta_{\mM_{1}}^{it} J_{\mN} \mM_1 J_{\mN} \Delta_{\mM_{1}}^{-it} \ .\label{eq-int-thm-cond-3-int-1}
    \end{eqnarray}
    Now since by definition of twisted modular intersection, $\mM_{1} \subset \mN$ is HSMI$-$, using Theorem \ref{thm:hsmi} there exists a one-parameter unitary $V(a)$ such that
    \begin{eqnarray}
        \mM_1 = V(-1) \mN V(1) \ .
    \end{eqnarray}
    With this, \eqref{eq-int-thm-cond-3-int-1} becomes
    \begin{eqnarray}
        \Delta_{\mM'_{1}}^{-it} \widetilde{\mM_2} \Delta_{\mM'_{1}}^{it}  &=&  V(-1)\Delta_{\mN}^{it} V(1) J_{\mN} V(-1) \mN V(1) J_{\mN} V(-1) \Delta_{\mN}^{-it} V(1) \nn\\
        &=& J_{\mN} V(-1) V(2-2e^{2\pi t}) \mN V(-2+2e^{2\pi t}) V(1) J_{\mN}  \label{eq-int-thm-cond-3-int-2}
    \end{eqnarray}
    where we have used \eqref{eq-hsmi-1} and \eqref{eq-hsmi-2} to commute $J_{\mN}$ and $\Delta_{\mN}$ through $V(a)$. Finally, we note that the half-sided translation theorem (see Appendix \ref{App:Hsmi}) dictates that 
    \begin{eqnarray}
        V(a) \mN V(-a) \subset \mN \qquad \text{for all } a\le 0 \ .
    \end{eqnarray}
    Combining this with \eqref{eq-int-thm-cond-3-int-2}, we find that for all $t \ge 0$, 
    \begin{eqnarray}
        \Delta_{\mM'_{1}}^{-it} \widetilde{\mM_2} \Delta_{\mM'_{1}}^{it}  &\subset& J_{\mN} V(-1) \mN V(1) J_{\mN} \nn\\
        &=&  J_{\mN}  \mM_1 J_{\mN} \nn\\
        &=&  \widetilde{\mM_2} \ .
    \end{eqnarray}
    Therefore, $\widetilde{\mM_{2}} \subset \mM'_1$ is HSMI$+$ and condition $3$ is also satisfied. Therefore, we have a representation of $\widetilde{PSL}(2,\mbR)$.

    Finally, we discuss the inversion map. From Theorem \ref{GenWiesbrockThm}, we have
    \begin{eqnarray}
        \mathcal{I}  &=&  J_{\mM}J_{\mM_{1}} \Delta_{\widetilde{\mM_2} }^{-i\frac{\log (2)}{2\pi}} \Delta_{\mM_1}^{-i\frac{\log (2)}{2\pi}}  \ .
    \end{eqnarray}
    Using \eqref{eq-m2-m1-jn}, we get
    \begin{eqnarray}
        \mathcal{I}  &=& J_{\mM}J_{\mM_{1}}J_{\mN} \Delta_{\mM_1}^{i\frac{\log (2)}{2\pi}} J_{\mN} \Delta_{\mM_1}^{-i\frac{\log (2)}{2\pi}}  \nn\\
        &=& J_{\mM}J_{\mM_{1}}J_{\mN} \Delta_{\mM_1}^{i\frac{\log (2)}{2\pi}} J_{\mN}J_{\mM_1} \Delta_{\mM_1}^{-i\frac{\log (2)}{2\pi}} J_{\mM_1} \nn\\
        &=&  J_{\mM} V(-2) \Delta_{\mM_1}^{i\frac{\log (2)}{2\pi}} V(2) \Delta_{\mM_1}^{-i\frac{\log (2)}{2\pi}} J_{\mM_1} \nn \\
        &=&  J_{\mM} V(2) J_{\mM_1} \nn \\
        &=& J_{\mM}J_{\mN}
    \end{eqnarray}
    where we have used \eqref{eq-hsmi-1} and \eqref{eq-hsmi-4}.     
\end{proof}

The above theorem establishes our claim that the twisted modular intersection property leads to a representation of $\widetilde{PSL}(2,\mbR)$ group. In the special case $U=1$ we have the modular intersection property, then $J_{\mN}$ and $J_{\mM}$ commute (since $J_{\mN}\mM J_{\mN} = \mM$), the inversion map in (\ref{inversionmodular}) squares to identity, i.e. $\mathcal{I}^2 = 1$, and hence, we obtain a representation of the $PSL(2,\mathbb{R})$. 

Next, we discuss the second result of this section:
\begin{theorem}[Twisted inclusion]\label{twistModularInclusion}
If $(\mN \subset \mM , U)$ is a twisted inclusion and $\mN \subset \mM$ is HSMI$+$, then we have a representation of $\widetilde{PSL}(2,\mbR)$, and the inversion map is $J_{\mM}J_{\widehat{\mN}}$ where $\widehat{\mN} = \mN \vee \left( UJ_{\mM}\mN J_{\mM}U^\dagger\right) $.
\end{theorem}
\begin{figure}
    \centering
    \includegraphics[width=\linewidth]{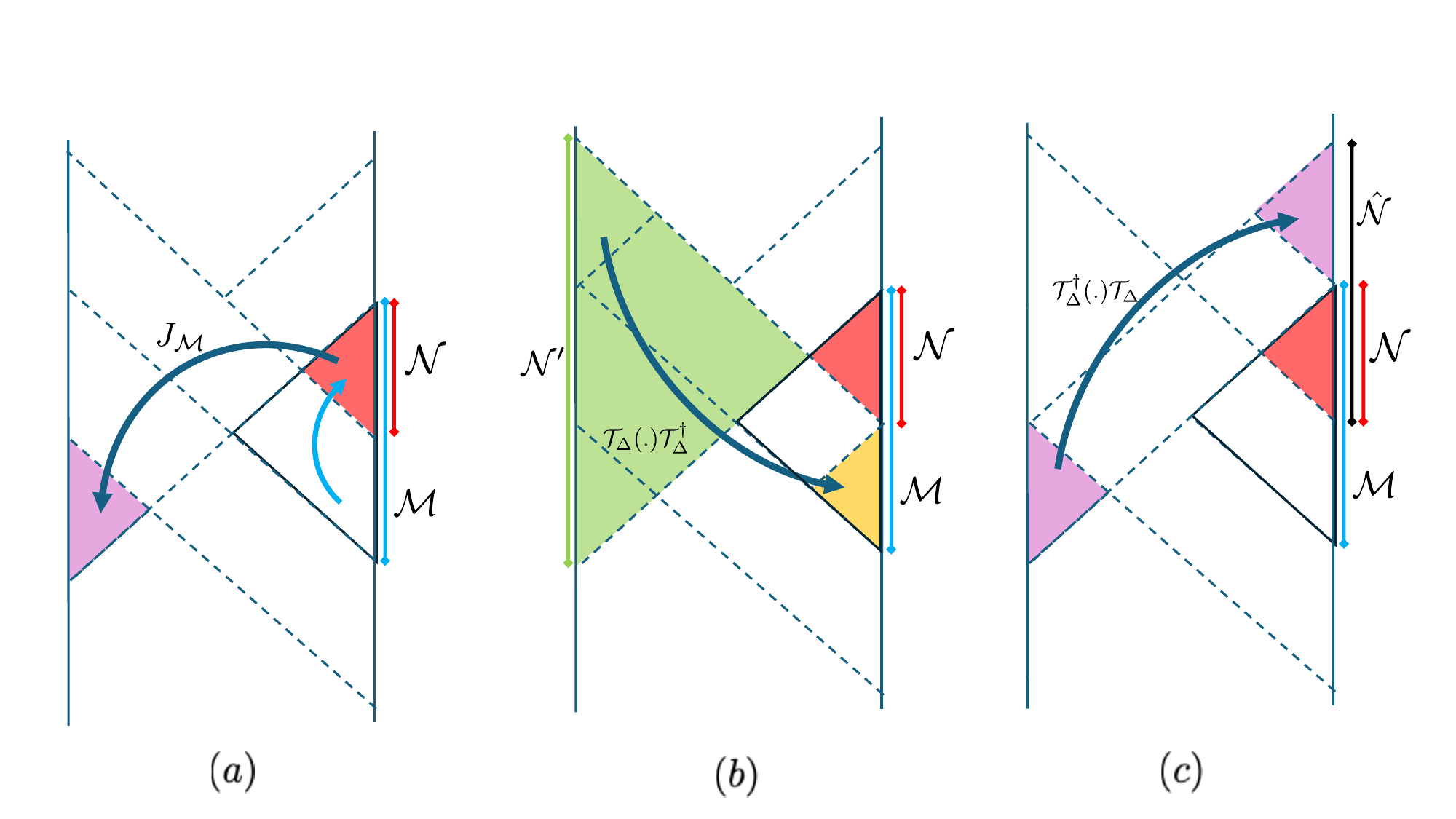}
    \caption{\small Twisted inclusion for 0+1-dimensional conformal GFF dual to AdS$_2$ bulk with unitary $U=\mT_\Delta^\dagger$ (see Theorems \ref{GenWiesbrockThm} and \ref{twistModularInclusion}). (a) $\mN\subset \mM$ is HSMI$+$. (b) The green region is $\mN'$ and the orange region is $U^\dagger\mN' U \cap \mM=\mT_\Delta \mN' \mT_\Delta^\dagger \cap \mM$ which is cyclic and separating. (c) The purple region on the right side is $UJ_{\mM}\mN J_{\mM}U^\dagger=\mT_\Delta^\dagger J_\mM \mN J_\mM \mT_\Delta$.}
    \label{fig:twisted_inclusion}
\end{figure}
\begin{proof}
Let us take $\mM_1 = \mN$ and $\mM_2 =  (U^{\dagger}\mN' U) \cap \M$ (see Figure \ref{fig:twisted_inclusion}). To prove this theorem, we need to show the unitary $U$ and algebras $\mM$, $\mM_1$, $\mM_2$ satisfy conditions $1$-$4$ of Theorem \ref{GenWiesbrockThm}. 
Note that condition $1$ of Theorem \ref{GenWiesbrockThm} is true by assumption whereas condition $4$ is satisfied due to the definition of twisted inclusion. To verify condition $2$, note that
\begin{eqnarray}
    \Delta_{\mM}^{-it} \mM_2 \Delta_{\mM}^{it}  &=&  \Delta_{\mM}^{-it} \left(\mM \cap (U^\dagger \mN' U) \right) \Delta_{\mM}^{it} \nn\\
    &=& \mM \cap \left(\Delta_{\mM}^{-it} (U^\dagger \mN' U) \Delta_{\mM}^{it}\right) \nn\\
    &=& \mM \cap \left( U^\dagger ( \Delta_{\mM}^{-it} \mN' \Delta_{\mM}^{it}) U\right) 
\end{eqnarray}
where we have used the fact that $\Delta_\mM^{-it}$ commutes with $U$. Since $\mN \subset \mM$ is assumed to be HSMI$+$, we have that $\Delta_{\mM}^{-it} \mN' \Delta_{\mM}^{it} \subset \mN'$ for all $t \le 0$. As a result, we get
\begin{eqnarray}
    \Delta_{\mM}^{-it} \mM_2 \Delta_{\mM}^{it} \subset \mM_2 \qquad \text{for all } t\le 0 \ .
\end{eqnarray}
Therefore, $\mM_2 \subset \mM$ is HSMI$-$ and so condition $2$ is satisfied. Next, note that $\widetilde{\mM_2} = U\mM_2 U^\dagger = \mM'_1 \cap (U\mM U^\dagger) \subset \mM'_1 $. To verify condition $3$, we note that 
\begin{eqnarray}
    \Delta_{\mM'_1}^{-it} \widetilde{\mM_2} \Delta_{\mM'_1}^{it} &=&  \Delta_{\mM'_1}^{-it} \left(\mM'_1 \cap (U \mM U^\dagger) \right) \Delta_{\mM'_1}^{it} \nn\\
    &=& \mM'_1 \cap \left(\Delta_{\mM'_1}^{-it} (U \mM U^\dagger) \Delta_{\mM'_1}^{it}\right) \nn\\
    &=& \mM'_1 \cap \left( U ( \Delta_{\mM'_1}^{-it} \mM \Delta_{\mM'_1}^{it}) U^\dagger\right)
\end{eqnarray}
where we have used the fact that $U$ commutes with $\Delta_{\mN}^{-it}$ and hence with $\Delta_{\mM_1}^{-it}$. Since $\mM_1 \subset \mM$ is HSMI$+$, Corollary \ref{kwing-corollary} implies that $\mM' \subset \mM'_1$ is HSMI$-$, which means $\Delta_{\mM'_1}^{it} \mM \Delta_{\mM'_1}^{-it} \subset \mM$ for all $t\ge 0$. Therefore
\begin{eqnarray}
    \Delta_{\mM'_1}^{it} \widetilde{\mM_2} \Delta_{\mM'_1}^{-it} \subset \widetilde{\mM_2} 
\end{eqnarray}
for all $t\ge 0$ and hence condition $3$ is satisfied. 
Therefore, Theorem \ref{GenWiesbrockThm} implies the existence of a unitary representation of $\widetilde{PSL}(2,\mbR)$ group and the inversion acts as
\begin{eqnarray}
    \mathcal{I}  =  J_{\mM} \Delta_{\mM_1}^{-i\frac{\log (2)}{2\pi}} J_{\widetilde{\mM_{2}}}  \Delta_{\mM_1}^{i\frac{\log (2)}{2\pi}} \ .
\end{eqnarray}
To simplify this expression, we note that 
\begin{eqnarray}
    \Delta_{\mM_1}^{-i\frac{\log (2)}{2\pi}} \widetilde{\mM_{2}}  \Delta_{\mM_1}^{i\frac{\log (2)}{2\pi}} &=& \Delta_{\mM_1}^{-i\frac{\log (2)}{2\pi}} \left(\mM'_{1} \cap (U\mM U^{\dagger})\right)  \Delta_{\mM_1}^{i\frac{\log (2)}{2\pi}} \nn\\
    &=&  \mM'_{1} \cap  \left(U (\Delta_{\mM_1}^{-i\frac{\log (2)}{2\pi}}\mM \Delta_{\mM_1}^{i\frac{\log (2)}{2\pi}}) U^\dagger\right) \label{eq-inv-twst-inc-int-1}
\end{eqnarray}
where we have used the fact that $U$ commutes with $\Delta_{\mM_1}$. Since $\mM_1 \subset \mM$ is HSMI$+$, we get
\begin{eqnarray}
    \Delta_{\mM_1}^{-i\frac{\log (2)}{2\pi}}\mM \Delta_{\mM_1}^{i\frac{\log (2)}{2\pi}} &=&  \Delta_{\mM_1}^{-i\frac{\log (2)}{2\pi}} V_{1}(-1)\mM_{1}V_{1}(1) \Delta_{\mM_1}^{i\frac{\log (2)}{2\pi}} \nn\\
    &=& V_{1}(-2)\mM_{1}V_{1}(2) \nn\\
    &=& J_{\mM} \mM'_{1} J_{\mM} 
\end{eqnarray}
where $V_{1}(a)$ is the same unitary defined in the proof of Theorem \ref{GenWiesbrockThm} and where we have used \eqref{eq-hsmi-3}, \eqref{eq-hsmi-1} and \eqref{eq-hsmi-4} in the first, second and third equalities, respectively. With this result, \eqref{eq-inv-twst-inc-int-1} becomes
\begin{eqnarray}
    \Delta_{\mM_1}^{-i\frac{\log (2)}{2\pi}} \widetilde{\mM_{2}}  \Delta_{\mM_1}^{i\frac{\log (2)}{2\pi}} &=&  \mM'_{1} \cap \left( U J_{\mM} \mM'_{1} J_{\mM} U^\dagger \right) \nn\\
    &=&  \left(\mN \vee ( UJ_{\mM}\mN J_{\mM}U^\dagger)\right)' 
\end{eqnarray}
where we have substituted $\mN = \mM_1$. By defining 
\begin{eqnarray}
    \widehat{\mN}  = \mN \vee \left( UJ_{\mM}\mN J_{\mM}U^\dagger\right) 
\end{eqnarray}
we get
\begin{eqnarray}
    \Delta_{\mM_1}^{-i\frac{\log (2)}{2\pi}} J_{\widetilde{\mM_{2}}}  \Delta_{\mM_1}^{i\frac{\log (2)}{2\pi}} =  J_{\widehat{\mN}} 
\end{eqnarray}
and hence, the inversion map becomes
\begin{eqnarray}
    \mathcal{I} =  J_{\mM}J_{\widehat{\mN}} \ .
\end{eqnarray}
\end{proof}
This ensures that the twisted inclusion leads to a representation of $\widetilde{PSL}(2,\mbR)$. If we take $U=1$ in the above theorem, then $\widehat{\mN}$ becomes invariant under $J_{\mM}$ i.e., $J_{\mM} \widehat{\mN} J_{\mM} = \widehat{\mN}$. This implies that $J_{\mM}$ and $J_{\widehat{\mN}}$ commute and the inversion squares to identity. Therefore, we recover Wiesbrock's result that the standard inclusion implies a representation of $PSL(2,\mathbb{R})$.

\subsection{Twisted modular intersection/inclusion implies locality in emergent AdS$_2$}\label{emergentlocalQFT}

Consider abstract triples $(\mathcal{M},\mathcal{N},U)$ that satisfy the conditions of twisted modular intersection and/or twister inclusion theorems. Having established such triples imply a unitary representation of  $\widetilde{PSL}(2,\mbR)$ group, in this subsection, we explicitly construct a {\it local} and covariant net of algebras in an emergent global AdS$_{2}$ spacetime.


Our result is the generalization of Wiesbrock's earlier results where he showed that the modular intersection property or the standard half-sided modular inclusion leads to a conformal net on a sphere \cite{wiesbrock1993conformal,wiesbrock1993symmetries}. For simplicity, we only focus on the twisted HSMI$+$ $(\mN \subset \mM, U)$ and present the construction of such a net. 

A global AdS$_2$ spacetime is parameterized by coordinates $\tau \in \mbR$ and $-\pi/2 \le \rho \le \pi/2$. The right and left boundaries are at $\rho = \pi/2$ and $\rho = -\pi/2$, respectively. Let us first concentrate on the region $-\pi/2 \le \tau \le \pi/2$. For this region, it is convenient to use the Kruskal coordinates:
\begin{eqnarray}
    u = \tan \frac{\tau+\rho}{2} \quad\quad\quad v = \tan \frac{\tau-\rho}{2} \ .
\end{eqnarray}
Note that $-\pi/2 \le \rho \le \pi/2$ implies that $uv> -1$ and the boundaries are at $uv = -1$. Now corresponding to each point $(u,v)$, we associate a right and left wedge. In particular, a right (left) wedge associated with a point $(u,v)$ is the set of all the points that are spacelike separated to $(u,v)$ and are towards the right (left) boundary. That is,
\begin{eqnarray}
    W_{R}(u_1,v_1) =& \, \{ (u,v) | u > u_1 , v < v_1 \} \ \nn\\
    W_{L}(u_1,v_1) =& \, \{ (u,v) | u < u_1 , v > v_1 \} \ .
\end{eqnarray}
For example, $W_{R}(0,0)$ is the right Rindler wedge whereas $W_{L}(0,0)$ is the left Rindler wedge.

Our goal is to associate to each of these wedges $W_{R/L}(u,v)$ an algebra $\mA_{R/L}(u,v)$ such that if two wedges are related by a bulk $\widetilde{PSL}(2,\mbR)$ transformation, then the associated algebras are also related by a unitary representation of the $\widetilde{PSL}(2,\mbR)$ group. 
Recall that the bulk isometries are generated by (see Appendix \ref{app:ads_geometry})
\begin{eqnarray}
    \hat{\eta}_{0} = 2\pi( u \partial_{u} -  v \partial_{v})\ , \qquad
    \hat{\eta}_{1} = \partial_{u} + v^2 \partial_{v}\ ,  \qquad
    \hat{\eta}_{2} = u^2 \partial_{u} + \partial_{v} \ .
\end{eqnarray}
Under the flow generated by these vector fields, a point $(u,v)$ flows to
\begin{eqnarray}
    \eta_{0}(s)  :  (u,v) \to (u_s,v_s) &=& (e^{2\pi s} u , e^{-2\pi s} v)  \nonumber\\
    \eta_{1}(s) :  (u,v) \to (u_s,v_s) &=& \left(u+s , \frac{v}{1-sv}\right)  \nonumber\\
    \eta_{2}(s)  : (u,v) \to (u_s,v_s) &=& \left(\frac{u}{1-su} , v+s\right) \ .
\end{eqnarray}
Therefore, the flow generated by $\eta_{0}$ preserves the wedge $W_{R}(0,0)$, but the flow generated by $\eta_{1}$ ($\eta_2$) maps $W_{R}(0,0)$ to a smaller wedge when $s>0$ ($s<0$). 

This motivates that we associate algebra $\mM$ with $W_{R}(0,0)$ and we relate the transformation generated by $\eta_{0}$, $\eta_{1}$ and $\eta_2$ with unitaries $V_{0}(a) = e^{ia G_0}$, $V_{1}(a) = e^{ia G_1}$ and $V_2(a) = e^{ia G_0}$, where (see Theorems \ref{GenWiesbrockThm} and \ref{twistModularInclusion})
\begin{eqnarray}
    G_{0} &=& \frac{1}{2\pi} \log\Delta_{\mM}  \nonumber\\
    G_{1} &=&  \frac{1}{2\pi} \left(\log\Delta_{\mN}-\log\Delta_{\mM}\right) \nonumber\\
    G_{2} &=&  \frac{1}{2\pi} \left(\log\Delta_{(U^\dagger \mN' U)\cap\mM}-\log\Delta_{\mM}\right) \ .
\end{eqnarray}
Now consider a general point $(u,v)$. Note that the wedge $W_{R}(0,0)$ is related to the wedge $W_{R}(u,v)$ by the following transformation
\begin{eqnarray}
    \eta_{1}(u) \cdot \eta_{2}\left(\frac{v}{1+uv}\right) : W_{R}(0,0) \to W_{R}(u,v) \ .\label{eq-wedge-R-map1}
\end{eqnarray}
This motivates the definition
\begin{eqnarray}
    \mA_{R}(u,v) &=&  V_{1}(u) V_{2}\left(\frac{v}{1+uv}\right) \mA_{R}(0,0) \, V_{2}\left(-\frac{v}{1+uv}\right) V_{1}(-u) \nonumber\\ 
    &=& V_{1}(u) V_{2}\left(\frac{v}{1+uv}\right) \mM \, V_{2}\left(-\frac{v}{1+uv}\right) V_{1}(-u) \ .\label{eq-def-Ar-net-1}
\end{eqnarray}
As an example, note that $\mA_{R}(1,0) = \mN$ and $\mA_{R}(0,-1) = (U^\dagger \mN' U)\cap\mM.$ 

Note that the map in \eqref{eq-wedge-R-map1} is not unique at all. In fact, the following map also maps $W_R(0,0)$ to $W_R(u,v)$:
\begin{eqnarray}
    \eta_{2}(v) \cdot \eta_{1}\left(\frac{u}{1+uv}\right) :  W_{R}(0,0) \to W_{R}(u,v) \ .\label{eq-wedge-R-map2}
\end{eqnarray}
This leads to another definition of $\mA_{R}(u,v)$
\begin{eqnarray}
    \mA_{R}(u,v)  = V_{2}(v) V_{1}\left(\frac{u}{1+uv}\right) \mM \, V_{1}\left(-\frac{u}{1+uv}\right) V_{2}(-v) \ .\label{eq-def-Ar-net-2}
\end{eqnarray}
However, the above two definitions of $\mA_{R}(u,v)$ are equivalent since 
\begin{eqnarray}
    V_{2}(v)V_{1}\left(\frac{u}{1+uv}\right) = V_{1}(u)V_{2}\left(\frac{v}{1+uv}\right) V_{0}(-2\log(1+uv)) 
\end{eqnarray}
and $V_{0}(a) \mM V_{0}(-a) = \Delta_{\mM}^{ia/2\pi} \mM \Delta_{\mM}^{-ia/2\pi} = \mM $. In the following, we will use both \eqref{eq-def-Ar-net-1} and \eqref{eq-def-Ar-net-2} interchangeably. 

Next, we show that the map $W_{R}(u,v) \to \mA_{R}(u,v)$ satisfies isotony and hence, is a net.  
In particular, we show that if $W_{R}(u_1,v_1) \subset W_{R}(u_2,v_2)$, then $\mA_{R}(u_1,v_1) \subset \mA_{R}(u_2,v_2)$. In other words, if $u_1 > u_2$ and $v_1 < v_2$, then $\mA_{R}(u_1,v_1) \subset \mA_{R}(u_2,v_2)$. From \eqref{eq-def-Ar-net-1}, we get that if $v_1 < v_2$, then 
\begin{eqnarray}
    \mA_{R}(u,v_1) &=&  V_{1}(u) V_{2}\left(\frac{v_1}{1+uv_1}\right) \mM \, V_{2}\left(-\frac{v_1}{1+uv_1}\right) V_{1}(-u)  \nonumber\\
    &\subset&  V_{1}(u) V_{2}\left(\frac{v_2}{1+uv_2}\right) \mM \, V_{2}\left(-\frac{v_2}{1+uv_2}\right) V_{1}(-u) \nonumber\\
    &=& \mA_{R}(u,v_2)  \label{eq-ads-net-isotony-1}
\end{eqnarray}
where we have used the half-sided translation theorem which implies
\begin{eqnarray}
    V_{2}(v_1) \mM \, V_{2}(-v_1) \subset V_{2}(v_2) \mM \, V_{2}(-v_2) \, \qquad \text{for } v_1 < v_2 
\end{eqnarray}
since $(U^\dagger \mN' U)\cap\mM \subset\mM$ is HSMI$-$ (see Theorem \ref{twistModularInclusion}). Similarly, from \eqref{eq-def-Ar-net-2}, we find that if $u_1>u_2$, then 
\begin{eqnarray}
    \mA_{R}(u_1,v) &=&  V_{2}(v) V_{1}\left(\frac{u_1}{1+u_1 v}\right) \mM \, V_{1}\left(-\frac{u_1}{1+u_1 v}\right) V_{2}(-v) \nonumber\\
    &\subset&  V_{2}(v) V_{1}\left(\frac{u_2}{1+u_2 v}\right) \mM \, V_{1}\left(-\frac{u_2}{1+u_2 v}\right) V_{2}(-v) \nonumber\\ 
    &=&  \mA_{R}(u_2,v)  \label{eq-ads-net-isotony-2}
\end{eqnarray}
where we have again used the half-sided translation theorem
\begin{eqnarray}
    V_{1}(u_1) \mM \, V_{1}(-u_1) \subset V_{1}(u_2) \mM \, V_{1}(-u_2) \qquad \text{for } u_1 > u_2 
\end{eqnarray}
since $\mN \subset \mM$ is HSMI$-$ (see Theorem \ref{twistModularInclusion}). By combining \eqref{eq-ads-net-isotony-1} and \eqref{eq-ads-net-isotony-2}, we get that if $u_1>u_2$ and $v_1<v_2$, then
\begin{eqnarray}
    \mA_{R}(u_1,v_1) \subset \mA_{R}(u_1,v_2) \subset \mA_{R}(u_2,v_2) 
\end{eqnarray}
and hence, the map $W_{R}(u,v) \to \mA_{R}(u,v)$ satisfies isotony.

Now we associate algebras to the left wedges, $W_{L}(u,v)$. First, note that the left wedge $W_{L}(0,0)$ is related to the right wedge $W_{R}(0,0)$ by the reflection and time-reversal (RT) symmetry: $(u,v) \to (-u,-v)$. Therefore, we relate the RT transformation by the anti-unitary $J_{\mM}$. This implies that
\begin{eqnarray}
    \mA_{L}(0,0) = J_{\mM}\mA_{R}(0,0) J_{\mM} = J_{\mM}\mM J_{\mM} = \mM' \ .
\end{eqnarray}
Moreover, analogous to \eqref{eq-def-Ar-net-1}, we define
\begin{eqnarray}
    \mA_{L}(u,v) &=&  V_{1}(u) V_{2}\left(\frac{v}{1+uv}\right) \mA_{L}(0,0) \, V_{2}\left(-\frac{v}{1+uv}\right) V_{1}(-u)  \nonumber\\ 
    &=& \, V_{1}(u) V_{2}\left(\frac{v}{1+uv}\right) \mM' \, V_{2}\left(-\frac{v}{1+uv}\right) V_{1}(-u) \ .
\end{eqnarray}
Therefore, the net satisfies the `wedge duality'
\begin{eqnarray}
    (\mA_{R}(u,v))' = \mA_{L}(u,v)  \ .
\end{eqnarray}

We now extend the net for the wedge regions to a net for the diamond-shaped regions. A general diamond-shaped region is parameterized by two points $(u_1,v_1)$ and $(v_1,v_2)$, and is defined as
\begin{eqnarray}
    \Diamond(u_1,v_1|u_2,v_2) = W_{R}(u_1,v_1) \cap W_{L}(u_2,v_2) \ .
\end{eqnarray}
We associate with this region the algebra
\begin{eqnarray}
    \mA_{\Diamond}(u_1,v_1|u_2,v_2) = \mA_{R}(u_1,v_1) \cap \mA_{L}(u_2,v_2) \ .
\end{eqnarray}
This net satisfies Haag's duality since
\begin{eqnarray}
    \left(\mA_{\Diamond}(u_1,v_1|u_2,v_2)\right)' &=&   \left(\mA_{R}(u_1,v_1) \cap \mA_{L}(u_2,v_2) \right)'  \nonumber\\
    &=&  \left(\mA_{R}(u_1,v_1)\right)' \vee \left(\mA_{L}(u_2,v_2)\right)' \nonumber\\
    &=&  \mA_{L}(u_1,v_1) \vee \mA_{R}(u_2,v_2) \ .
\end{eqnarray}
As a result, the net satisfies locality. 

Until now, we have only restricted to the region $-\pi/2 \le \tau \le \pi/2$ of the global AdS$_{2}$ spacetime. We extend the net beyond this region by using the unitary $U$. In particular, if two wedge/diamond regions are related to each other by the antipodal transformation $(\tau,\rho) \to (\tau+\pi,-\rho)$, then the associated algebras are related by the unitary $U$.

This finishes our construction of a local and $\widetilde{PSL}(2,\mbR)$ covariant net of algebras in a global AdS$_{2}$ spacetime only from the data of twisted inclusion $(\mN \subset \mM, U)$.

\section{Ergodic hierarchy in GFF and Stringy spacetimes}\label{Sec:Stringy spacetime}

We saw in Section \ref{sec:gff} that a theory of GFF in a quasi-free state is fixed with the knowledge of the spectral density $\rho(\omega)$ and inverse temperatures $\beta(\omega)$ (or equivalently $\gamma(\omega)$). For simplicity, we will focus on the case where $\beta(\omega)$ is independent of frequency i.e., we are in a thermal state of inverse temperature $\beta$. Then, the KMS relation in (\ref{KMS}) ties the spectral density function to the two-point function:
\begin{eqnarray}
f,g\in L^2(\rho):\quad G(f,g)=\braket{\varphi(f)\varphi(g)}=\braket{f|g}_\rho= \sqrt{2\pi} \int_0^\infty d\omega \rho(\omega)b_+^2(\omega)f(-\omega)g(\omega)\ .
\end{eqnarray}
Since the thermal odd-point functions vanish in any GFF theory, the two-point function is the same as its connected piece. The long-time behavior of the thermal two-point function depends crucially on the spectral density function. We can define the following ergodic classes based on how well at late times the system forgets its initial perturbations based on the spectral density:
 
\begin{enumerate}

\item {\bf Unforgetful (Almost periodic):} A local perturbation at time $t=0$ will never be truly forgotten because at late enough times, there are always local operators that are sensitive to initial perturbations. This occurs when the spectral density is supported on a countable set of points
\begin{eqnarray}\label{diracdelta}
\rho(\omega)=\sum_{n\in\mathbb{Z}} a_n \delta(\omega-\omega_n)\ .
\end{eqnarray}
In real time, the thermal two-point function $G(f,g)$ becomes a ($B^2$-Besicovitch) almost periodic function of time, almost repeats itself and does not have any late time limit \cite{Furuya:2023fei}. 

For a theory of $0+1$-dimensional GFF that corresponds to single trace operators of conformal field theory on a spatial sphere, we expect $\omega_n=\omega_0+2n$ which implies that the von Neumann algebra of an interval of length $2\pi$ i.e., $\mA_{(-\pi,\pi)}$, is isomorphic to $\mA_{(-\infty,\infty)}$ \cite{Gesteau:2024rpt}. This means that this spectral density is exponential type $\pi$. 

\item {\bf Strong-mixing (clustering in time):} 
It was observed in \cite{Furuya:2023fei} that when the spectral density is Lebesgue measurable the thermal two-point function clusters in time i.e., it is strongly mixing, and the resulting von Neumann algebra is type III$_1$. Since this can occur in GFF with non-holographic spectral densities, we say we have {\it Stringy information loss}.

Note that there is no relation between infinite-exponential type and strong mixing. We saw that a compactly supported spectral density that is measurable shows strong mixing but has vanishing causal depth parameter $T_{dep}=0$. Physical examples of such systems include the IOP model \cite{iizuka2010matrix,Gesteau:2024rpt}. We also saw that the primon gas has an infinite causal depth parameter, $T_{dep}=\infty$, but has a discrete spectrum which implies almost periodicity.

\item {\bf Exponentially-mixing:} From quasi-normal mode physics, it is expected that in a large class of systems, the decay of the GFF thermal correlator is exponential. In this case, we say we have {\it Stringy quasi-normal modes}.
The necessary and sufficient condition for the exponential decay was identified in Lemma 5 of \cite{Furuya:2023fei}.
Define the two-point function $f_{gh}(t)=\braket{\varphi(g)\varphi(h)}$ of the GFF theory. Then,
\begin{lemma}\label{exponential}
There exists a small $\epsilon>0$ such that for all $z<\epsilon$ we have
 \begin{eqnarray}\label{expfalloff}
     \int dt\, e^{2z |t|}|f_{gh}(t)|^2 <\infty
 \end{eqnarray}
if and only if $f_{gh}(\omega)$ is analytic in the strip $|\Im(\omega)|<\epsilon$ and inside the strip we have: 
\begin{enumerate}
    \item $f(\omega)\in L^2(\Re(\omega))$ 
    \item The supremum over these $L^2$-norms over the strip remains bounded.
\end{enumerate}
  \end{lemma}

\item {\bf Quantum K-system (Stringy horizon):} A stronger ergodic condition is that we have {\it future algebras} in the sense of Section \ref{time-intervalvNalgebras}.
It was shown in \cite{Furuya:2023fei,leutheusser2023emergent,ouseph2024local,gesteau2024emergent,Gesteau:2024rpt} the existence of a future algebra implies half-sided modular inclusion and the existence of a {\it Stringy horizon}.
Note that the existence of future algebras resembles a generalization of the causal depth parameter to half-infinite time intervals \cite{Gesteau:2024rpt}.

\item {\bf Conformal Anosov system:} The subalgebras of the quantum K-system transform under a representation of the universal cover of $\mathfrak{psl}(2,\mathbb{R})$. Such systems are locally hyperbolic. These are quantum analogs of the widely studied example of classical Anosov flows on the tangent bundle of a Riemann surface of negative curvature.

\end{enumerate}

In this work, we studied the time interval algebras of $0+1$-dimensional conformal GFF with spectral density $\rho(\omega)=\omega^{2\Delta-1}$, and constructed the modular flow of time interval algebras explicitly. In this case, viewed as a map in Poincar\'e coordinates, we have 
\begin{itemize}
    \item The map $I\to \mathcal{A}_I$ is injective, and has infinite causal depth $T_{dep}=\infty$.
    \item We have strong mixing because the connected two-point function decays as a function of modular time.
    \item We have a quantum K-system. There is a horizon: the Poincar\'e horizon.
    \item We have a conformal Anosov system because we have the action of the Lie algebra $\mathfrak{psl}(2,\mathbb{R})$.
\end{itemize}
Note that even though the two-point function itself does not decay exponentially fast, we still have a dense set of observables with smoothing in time and space that decay exponentially fast, as expected in an Anosov system. Viewed in terms of global coordinates, our system is almost-periodic and has finite causal depth $\pi$.

\section{Discussion}\label{sec:discussion}

In this work, we derived the modular flow of time interval algebras for conformal $0+1$-dimensional GFF. For integer $\Delta$, we 
proved that the time interval algebras form a local net, but we did not discuss an explicit construction of this net. The theory with $\Delta=1$ is a massless free scalar field in AdS$_2$. In Poincar\'e coordinates the boundary time interval algebras are a chiral massless scalar field: a $U(1)$ current algebra CFT. The lowest dimension primary is the current $J=\p\varphi$ with scaling dimension one. We also have vertex operators $e^{i \alpha \varphi(f)}$ in the spectrum \cite{buchholz1988current,weiner2007conformal}. The theory with integer $\Delta>1$ corresponds to projecting the $U(1)$ current theory to subalgebras generated by a fundamental field with a higher scaling dimension. For more detail, see \cite{buchholz1988current}.

\begin{figure}[t]
   \centering
    \includegraphics[width=.8\textwidth]{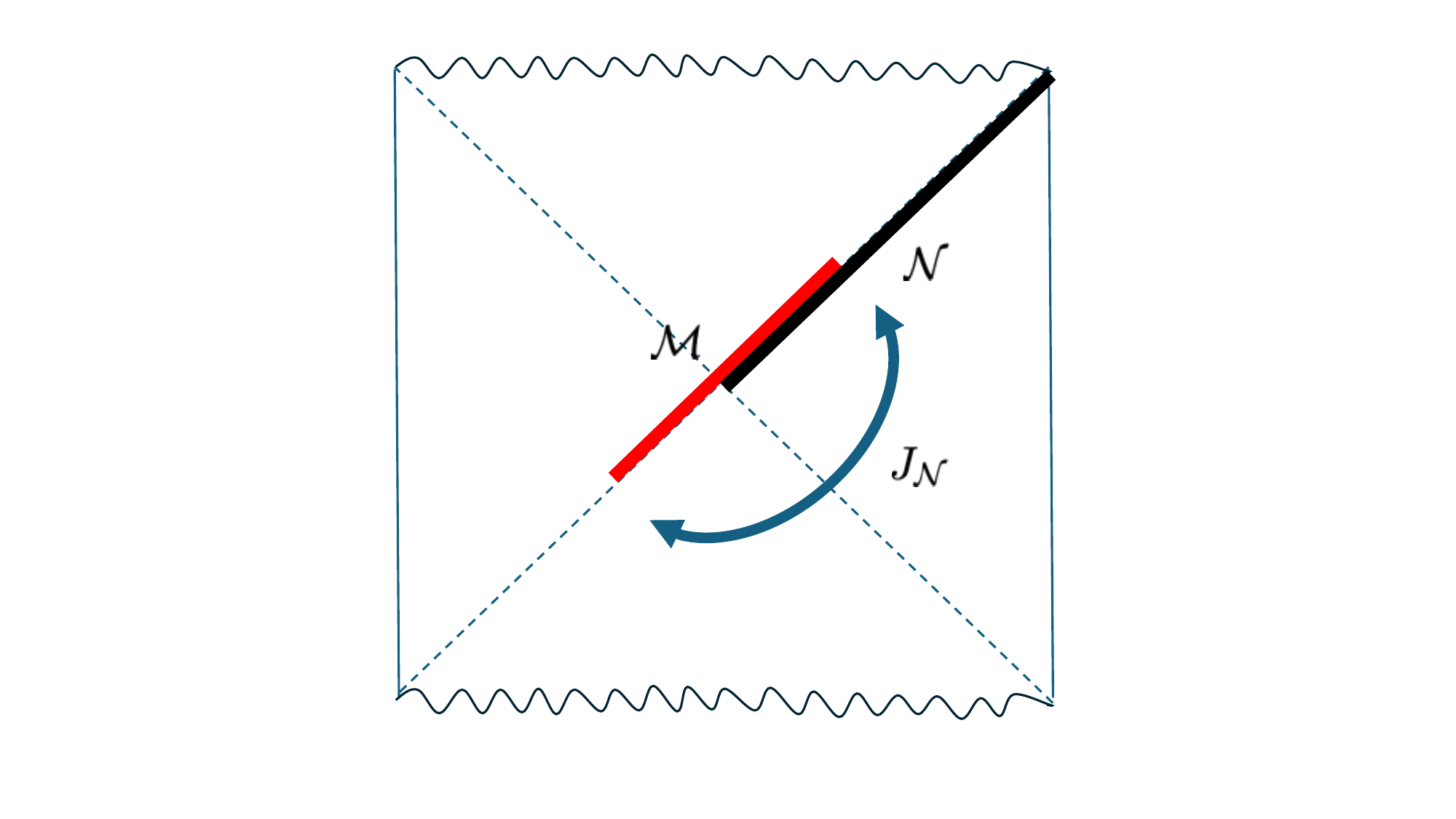}
    \caption{\small The choice of algebras $\mM$ and $\mN$ localized on the bifurcate future Killing horizon of a higher dimensional AdS eternal black hole satisfies the conditions of the modular intersection theorem. We obtain two copies of $PSL(2,\mathbb{R})$ each acting on the future and past Killing horizons.}
    \label{fig:eternalbh}
\end{figure}

In this work, we focused on AdS$_2$ in Poincar\'e and global coordinates. AdS$_2$-black holes or AdS$_2$-Rindler are quotients of global AdS$_2$. Repeating the analysis of this paper for these cases is straightforward and follows from the following logic: 
We saw in (\ref{linearbog}) that a linear Bogoliubov transformation sends the vacuum state killed by annihilation operators $a(\omega)$ to the KMS state killed by annihilation operators $c_L$ and $c_R$ in (\ref{newannihilation}). It follows from the Stone-von Neumann theorem that for a finite number of oscillators, this transformation is represented by a unitary in the Hilbert space. For an infinite number of oscillators, we are not guaranteed that the resulting Hilbert spaces are unitarily equivalent. However, in the special case where the transformation sends positive (negative) frequency modes to positive (negative) frequency modes, the vacua are the same, and the Hilbert spaces are unitarily equivalent. It is a well-known fact that the vacuum of Poincar\'e AdS$_2$, global AdS$_2$, and the AdS$_2$ black hole match \cite{strominger1999ads2,hamilton2006local}. Therefore, the modular flows of their net of von Neumann algebras are related to those we derived here by the unitary map that represents the conformal coordinate transformation. Another generalization of the current work is to consider higher-dimensional AdS and their quotients as topological black holes. We postpone careful treatment of these cases to upcoming work.

An interesting application of our results is the case of the conformal algebra of the Killing horizons of a higher-dimensional eternal black hole. Consider the GFF dual to a higher-dimensional eternal black hole, and the subalgebra of observables localized on the future or past Killing horizons. Choosing $\mM$ and $\mN$ as the algebra of null regions in Figure \ref{fig:eternalbh}, we find that the conditions of the modular intersection theorem i.e., Theorem \ref{twistedModInt} for $U=1$, are satisfied.\footnote{Note that the algebra of $\mM$ is dense in the Hilbert space obtained by the action of all operators localized only on the Killing horizon. This Hilbert space is strictly smaller than the full Hilbert space of the $\ket{TFD}$ state.} We obtain the action of a copy of $\mathfrak{psl}(2,\mathbb{R})$ on the future Killing horizon and another copy of $\mathfrak{psl}(2,\mathbb{R})$ on the past Killing horizon. 

Another potential generalization we would like to highlight is the local modular action work of \cite{buchholz2000geometric} and the approach of Summers \cite{summers1996geometric} that focuses on the action of modular conjugation of time interval algebras as an order-preserving bijection in the net of time interval algebras. This allows for an alternate approach to generalizing from the $0+1$-dimensional conformal group $PSL(2,\mathbb{R})$ to the higher dimensional conformal group $SO(d,2)$. We postpone a careful treatment of this approach to upcoming work.

\paragraph{Acknowledgements:}  We would like to thank Matthew Dodelson, Keiichiro Furuya, Hong Liu and Aron Wall for insightful conversations. NL's work is supported through DOE grants DE-SC0007884 and
DE-SC0025547. MM's work was supported in part by QuantISED $1.0$ grant No. DE-SC0020360. 

\appendix

\section{$PSL(2,\mbR)$ and the universal cover $\widetilde{PSL}(2,\mathbb{R}) $} \label{app:psl2r}

The $PSL(2,\mbR)$ group consists of linear transformations on the complex plane $\mbC \ni t$ 
\begin{eqnarray}
    g  :  t  \to  \frac{at+b}{ct+d}  ,\label{eq-psl-map-app}
\end{eqnarray}
where $a,b,c,d \in \mathbb{R}$, and $ad-bc=1$. They preserve the upper and lower halves of the complex $z$-plane. Each element $g$ of the $PSL(2,\mbR)$ group is represented as a $2\times 2$ matrix of determinant 1 quotiented by $\mbZ_2$ 
\begin{eqnarray}
    g = \begin{pmatrix}
        a & b\\
        c & d
    \end{pmatrix} \sim \begin{pmatrix}
        -a & -b\\
        -c & -d
    \end{pmatrix} \ .
\end{eqnarray}
The Iwasawa (KAN) decomposition of the $PSL(2,\mbR)$ group
\begin{eqnarray}
    R_\theta =  \begin{pmatrix}
        \cos\frac{\theta}{2} & \sin\frac{\theta}{2} \\
        -\sin\frac{\theta}{2} & \cos\frac{\theta}{2}
    \end{pmatrix}\ , \qquad
    D_s  =  \begin{pmatrix}
        e^{s/2} & 0 \\
        0 & e^{-s/2}
    \end{pmatrix} \ , \qquad
    T_b = \begin{pmatrix}
        1 & b \\
        0 & 1
    \end{pmatrix}
\end{eqnarray}
where $R_\theta, D_s, T_b$ correspond to rotations, dilatations and translations respectively. The corresponding generators satisfy the Lie algebra
\begin{eqnarray}\label{eq-psl-algebra-app}
   && \mathfrak{r}  = \frac{1}{2} \begin{pmatrix}
        0 & 1\\
        -1 & 0
    \end{pmatrix}  , \qquad    
    \mathfrak{d}  = \frac{1}{2} \begin{pmatrix}
        1 & 0\\
        0 & -1
    \end{pmatrix}   , \qquad 
    \mathfrak{t}  =    \begin{pmatrix}
        0 & 1\\
        0 & 0
    \end{pmatrix}    \nn \\
    && \left[ \mathfrak{r},\mathfrak{d}\right]  = \mathfrak{r} -  \mathfrak{t} \ , \qquad \left[ \mathfrak{d},\mathfrak{t}\right] = \mathfrak{t} \ ,\qquad \left[ \mathfrak{r},\mathfrak{t}\right] = \mathfrak{d}\ .
\end{eqnarray}
This prompts us to define the generator corresponding to special conformal transformations $S_c$ giving another decomposition of $PSL(2,\mbR)$ group
\begin{eqnarray}
    &&\mathfrak{s} = \mathfrak{t} - 2 \mathfrak{r} = \begin{pmatrix}
        0 && 0\\1 && 0
    \end{pmatrix}\ , \qquad S_c  =  e^{c\mathfrak{s}} = \begin{pmatrix}
        1 && 0\\
        c && 1
    \end{pmatrix} \nn \\
    && \left[ \mathfrak{s},\mathfrak{d}\right]  =  \mathfrak{s} \ , \qquad \left[ \mathfrak{d},\mathfrak{t}\right] =  \mathfrak{t} \ ,\qquad \left[ \mathfrak{t},\mathfrak{s}\right] =2 \mathfrak{d}\ .
\end{eqnarray}
Note that $R_\theta$ is an element of the SO(2) group which is the maximal compact subgroup of $PSL(2,\mbR)$. Moreover, the element $R_{2\pi} = e^{2\pi \mathfrak{r}}  \sim \mbI $ corresponding to rotation by $2\pi$ is in the center of the group.

\subsection*{The universal cover $\widetilde{PSL}(2,\mathbb{R}) $}

The topology of the $PSL(2,\mbR)$ group can be deduced from the Iwasawa decomposition since any element of the $PSL(2,\mbR)$ group can be decomposed as a product of rotation, dilatation, and translation \cite{longo}. For any $g \in PSL(2,\mathbb{R})$, we can find $\theta$, $s$, and $b$ such that
\begin{eqnarray}
    g = \begin{pmatrix}
        a && b\\ c&& d
    \end{pmatrix} = R_\theta \cdot D_s \cdot T_b  \ .
\end{eqnarray}
Thus the topology of the $PSL(2,\mbR)$ group is $S^{1}\times \mathbb{R}^{2}  $. The universal covering group  $\widetilde{PSL}(2,\mathbb{R})$ decompactifies the circle $S^1$ and therefore has the topology $\mathbb{R}^{3}$. The universal cover $\widetilde{PSL}(2,\mathbb{R})$ is then the group of transformations with generators satisfying the Lie algebra of $PSL(2,\mbR)$ in \eqref{eq-psl-algebra-app}, but however, the rotation $R_\theta$ is no longer periodic in $\theta$. Since the universal cover has the same group structure, the element $R_{2\pi}$ is still in the center of the group and is proportional to the identity operator. In a unitary representation of $\widetilde{PSL}(2,\mathbb{R})$, the element $R_{2\pi}$ is equal to a phase $e^{i 2\pi\mu}$ where $\mu \in \mathbb{R}/\mathbb{Z}$. \cite{Kitaev:2017hnr}

\section{Generalized Hilbert Transform (GHT)} \label{sec-hilbert}

The Hilbert transform of a function $f(t)$ is defined as the Cauchy principal value of the convolution
\begin{eqnarray}
    H(f)(t)=-\sqrt{\frac{2}{\pi}}\,  f(t) \ast \frac{1}{t}  =- \sqrt{\frac{2}{\pi}}  \text{ p.v.} \int_{-\infty}^\infty d\tau\,\frac{f(\tau)}{t-\tau}
\end{eqnarray}
where $\frac{1}{t}$ is well-defined as a distribution against Schwartz functions. In the frequency space, the Hilbert transform acts as
\begin{eqnarray}\label{eq-Hilbert-old}
    H(f)(\omega)  = -i \sgn(\omega) f(\omega) =   e^{-i\pi \text{sgn}(\omega)/2}  f(\omega) \ .
\end{eqnarray}
Consider the functions $f_\pm (t) =  f(t)\pm i H(f)(t)$. The Fourier transform only contains positive or negative modes
\begin{eqnarray}
    f_\pm(\omega) = f(\omega) \pm \sgn(\omega) f(\omega)\ .
\end{eqnarray}
Let us define a generalization of the Hilbert transform for any $\Delta \in \mbR^+$
\begin{eqnarray}\label{eq-mod-Hilbert-def}
    H_\Delta(f)(\omega)  &=&  e^{-i\pi \Delta  \text{sgn}(\omega)}  f(\omega) \nn \\
     H_\Delta(f)(t) &=& f(t) \cos(\pi\Delta) + H(f)(t) \sin(\pi \Delta)\ .
\end{eqnarray}
The inverse generalized Hilbert transform is 
\begin{eqnarray}\label{eq-mod-Hilbert-inv}
    H^{-1}_\Delta(f)(\omega)  = e^{i\pi \Delta  \text{sgn}(\omega)}  f(\omega) \ . 
\end{eqnarray}
We find that the generalized Hilbert transform (GHT) interpolates between the identity operator when $\Delta$ is an integer and the Hilbert transform when $\Delta$ is a half-integer up to an overall sign.
Consider a  reflection operator $R : f(t) \to f_{R}(t) = f(-t)  $. Note that
 \begin{eqnarray}
     H_{\Delta}\circ R (f)(\omega) &=& H_{\Delta} (f_R)(\omega) = e^{-i\pi\Delta \text{sgn}(\omega)} f_{R}(\omega) \nn\\ 
     &=& e^{-i\pi\Delta \text{sgn}(\omega)} f_{R}(\omega) = e^{-i\pi\Delta \text{sgn}(\omega)} f(-\omega) \nn\\ 
     &=& R\circ H_{\Delta}^{-1} (f)(\omega) \ .
 \end{eqnarray}
This shows that $R\circ H_{\Delta}\circ R = H_{\Delta}^{-1} .$

\section{Proofs for Sections \ref{sec:gff_conformal} and \ref{sec:modular-flow}}\label{app:proofs-sec-gff-conformal}

\subsection{Proof for Lemma \ref{lemma:psl2r_rep}}\label{app:proof_lemma_psl2r_rep}

\begin{proof}
    Let us consider two $PSL(2,\mathbb{R})$ transformation:
    \begin{align*}
        g_1 = \begin{pmatrix}
        a_1 && b_1\\ c_1 && d_1
    \end{pmatrix} \quad\quad\quad\quad\quad\quad\quad\quad g_2 = \begin{pmatrix}
        a_2 && b_2 \\ c_2 && d_2 
    \end{pmatrix} \, .
    \end{align*}
    Let us denote their composition by $g_3$, which is also a $PSL(2,\mathbb{R})$ transformation:
    \begin{align}\label{eq-g3}
        g_{3} = g_2 \cdot g_1 \, = \begin{pmatrix}
        a_3 && b_3 \\ c_3 && d_3 
    \end{pmatrix} \, = \, \begin{pmatrix}
        a_1 a_2+b_2c_1 \quad && a_2 b_1 + b_2 d_1 \\ a_1 c_2 + c_1 d_2 \quad&& c_2 b_1 + d_1 d_2 
    \end{pmatrix} \, .
    \end{align}
    Our goal in the following is to show that the following transformation (copied from \eqref{universal_psl2r_action}) also satisfies the above group structure:
    \begin{eqnarray} \label{eq-psl-ansatz}
    g \, : \,   (P_\pm f)(t)\to (P_\pm f_g)(t)= P_\pm\lb |a-ct|^{2(\Delta-1)} e^{\mp 2\pi i\Delta\sgn(c) \Theta(ct-a)}  f\lb \frac{d\,t-b}{a-ct}\rb\rb
\end{eqnarray}
That is, we need to show that 
\begin{align} \label{eq-psl-proof-goal}
   P_\pm f_{g_3} = P_\pm \left( e^{\pm i 2\pi\Delta n_{g_3}} f_{g_2 \cdot g_1} \right) \, ,
\end{align}
for some integer $n_{g_3}$.

From \eqref{eq-psl-ansatz}, we get
\begin{align}\label{psl-rep-proof-int-1}
    (P_\pm f_{g_2 \cdot g_1})(t)= P_\pm\lb |a_2-c_2 t|^{2(\Delta-1)} e^{\mp 2\pi i\Delta\sgn(c_2) \Theta(c_2 t-a_2)}f_{g_1}\lb \frac{d_2\,t-b_2}{a_2-c_2 t}\rb\rb
\end{align}
Note that since
\begin{align*}
\chi_\pm(t) = |a-ct|^{2(\Delta-1)}  e^{\pm i 2\pi \Delta \sgn(c) \Theta(ct-a)} = \lb a - c(t \mp i0^+) \rb^{2(\Delta-1)} \, , 
\end{align*}
$\chi_{\pm}(t)$ are analytic in $\pm \text{Im}(t) < 0$ half-plane. Therefore, $\chi_{+}(t)$ ($\chi_{-}(t)$ ) is a positive (negative) frequency function. Therefore, for any function $f(t)$, we have
\begin{align}
    P_\pm\lb |a-ct|^{2(\Delta-1)} e^{\mp 2\pi i\Delta\sgn(c) \Theta(ct-a)} P_{\mp}f(t)\rb = 0 \, ,
\end{align}
and hence, 
\begin{align} \label{eq-proj-inner}
    P_\pm\lb |a-ct|^{2(\Delta-1)} e^{\mp 2\pi i\Delta\sgn(c) \Theta(ct-a)} f(t)\rb = P_\pm\lb |a-ct|^{2(\Delta-1)} e^{\mp 2\pi i\Delta\sgn(c) \Theta(ct-a)} P_{\pm}f(t)\rb \, .
\end{align}
With this observation, we write \eqref{psl-rep-proof-int-1} as
\begin{align}
    (P_\pm f_{g_2 \cdot g_1})(t)= P_\pm\lb |a_2-c_2 t|^{2(\Delta-1)} e^{\mp 2\pi i\Delta\sgn(c_2) \Theta(c_2 t-a_2)} P_{\pm} f_{g_1}\lb \frac{d_2\,t-b_2}{a_2-c_2 t}\rb\rb \, .
\end{align}
We further simplify this equation by using \eqref{eq-psl-ansatz} for $f_{g_1}$ to get
\begin{align}
(P_\pm f_{g_2 \cdot g_1})(t)= P_\pm \Bigg(& |a_2-c_2 t|^{2(\Delta-1)} e^{\mp 2\pi i\Delta\sgn(c_2) \Theta(c_2 t-a_2)} \nn \\ \times& P_{\pm} \lb \left| \frac{a_3 - c_3 t}{a_2-c_2 t}\right|^{2(\Delta-1)} e^{\mp 2\pi i \Delta \text{sgn}(c_1) \Theta\big(\frac{c_3\,t-a_3}{a_2-c_2 t}  \big)} f\lb \frac{d_3\,t-b_3}{a_3-c_3 t}\rb \rb\Bigg) \, ,
\end{align}
where we have used the expressions of $a_3$, $b_3$, $c_3$, and $d_3$ from \eqref{eq-g3}. After using \eqref{eq-proj-inner} again, we get
\begin{align}\label{psl-rep-proof-int-2}
    (P_\pm f_{g_2 \cdot g_1})(t)= P_\pm \Bigg(& |a_3-c_3 t|^{2(\Delta-1)} e^{\mp 2\pi i\Delta\sgn(c_3) \Theta(c_3 t-a_3)} e^{\pm 2\pi i\Delta n_{g_3}(t)} f\lb \frac{d_3\,t-b_3}{a_3-c_3 t}\rb \Bigg) \, ,
\end{align}
where $n_{g_3}(t)$ is an integer given by
\begin{align}
    n_{g_3}(t) = \sgn(c_3) \Theta(c_3 t-a_3) - \sgn(c_2) \Theta(c_2 t-a_2) - \text{sgn}(c_1) \Theta\lb\frac{c_3\,t-a_3}{a_2-c_2 t}\rb \, .
\end{align}

To show that \eqref{psl-rep-proof-int-2} satisfy \eqref{eq-psl-proof-goal}, we just need to show that $n_{g_3}(t)$ is a constant.
When $c_1 = c_2 =0$, $n_{g_3}(t)$ is trivially constant. So let us first consider the case when $c_1 = 0$ but $c_2 \ne 0$. In this case, $c_3 = a_1 c_2$ and $a_3 = a_1 a_2$, and hence,
\begin{align}
    n_{g_3}(t) = \sgn(a_1 c_2) \Theta(a_1 (c_2 t-a_2)) - \sgn(c_2) \Theta(c_2 t-a_2) = - \Theta(-a_1) \text{sgn}(c_2) \, ,
\end{align}
which is a constant. Finally, we focus on the case where neither $c_1$ nor $c_2$ vanish. Taking the time-derivative of $n_{g_3}(t)$, we get
\begin{align*}
    \frac{d n_{g_3}}{dt} \, = \, \lb \delta(t - a_3/c_3) - \delta(t-a_2/c_2) \rb \times \left\{ 1 - \text{sgn}(c_1)\text{sgn}(a_2 c_3 - c_2 a_3)\right\} \, .
\end{align*}
From \eqref{eq-g3}, we find that $a_2 c_3 - c_2 a_3 = c_1$. Therefore, we get
\begin{align}
    \frac{d n_{g_3}}{dt} \, = \, 0 \, ,
\end{align}
and hence, $\psi_{g_2,g_1}(t)$ is a constant integer. It is straightforward to generalize to $g_1$, $g_2$ in $n_{g_1}$-th, $ n_{g_2}$-th copy, with non-zero $n_{g_1}, n_{g_2} \in \mathbb{Z}$. This finishes our proof that \eqref{eq-psl-ansatz} form a representation of $\widetilde{PSL}(2,\mbR)$.
\end{proof}

\subsection{Proof for Lemma \ref{lemma-psl-unitary}}\label{app:proof_lemma_psl2r_unitary} 
\begin{proof}
 Consider $0+1$-dimensional conformal GFF in the vacuum state with the correlator in (\ref{conformal0p1}). The two-point function of the smeared fields is 
 \begin{eqnarray}\label{gff_smeared_twopoint}
    \braket{\varphi(f) \varphi(h) }  &=& \int_{-\infty}^\infty d\omega \, G(\omega)  f(-\omega) h(\omega)  =  \int_{0}^{\infty} d\omega \, \omega^{2\Delta-1}  f(-\omega) h(\omega)  \nn \\
    &=&  \int_{0}^{\infty} d\omega \, \omega^{2\Delta-1}  (P_-f)(-\omega) (P_+h)(\omega) \nn \\
    &=& \braket{\varphi(P_-f) \varphi(P_+h) } \ .
 \end{eqnarray}

\begin{enumerate}
\item To see the invariance of the two-point function of the smeared fields under the action of the universal cover of $PSL(2,\mbR)$ in \eqref{universal_psl2r_action}, consider the functions $\chi_\pm(t)$ with only positive or negative frequencies 
\begin{eqnarray}
     \chi_\pm(t)&=& |a-ct|^{2(\Delta-1)}  e^{\pm i 2\pi \Delta \sgn(c) \Theta(ct-a)} \nn \\
     \chi_\pm(\omega) &=&  \frac{\Gamma(2\Delta-1) \sin(2\pi\Delta) e^{i\pi\Delta\sgn(c)}}{\sqrt{2\pi} |c|^{1-2\Delta}} \frac{\Theta(\pm\omega)}{|\omega|^{2\Delta-1}} e^{i \frac{a}{c} \omega} \ .
\end{eqnarray}
In the vacuum state for conformal GFF, under the coordinate action of $PSL(2,\mbR)$ in \eqref{psl2r_action}, the two-point function transforms as
\begin{eqnarray}\label{psl2r_action_2point}
    G(t_1'-t_2') = G(t_1-t_2) \frac{\chi_+(t_1) \chi_-(t_2)}{\lb (a-ct_1)(a-ct_2)\rb^{-2}}
\end{eqnarray}
where 
\begin{eqnarray}
    &&t' = \frac{d \,t-b}{a-ct}\ , \qquad t'_1 - t'_2 = \frac{t_1-t_2}{(a-ct_1)(a-ct_2)} \nn \\
    &&\text{sgn}(t'_{1}-t'_{2})= \text{sgn}(t_{1}-t_{2}) \, -2 \text{sgn}(c) \Theta(ct_{1}-a)+ 2 \text{sgn}(c) \Theta(ct_{2}-a) \ .
\end{eqnarray}
However, the action on the smeared fields in \eqref{universal_psl2r_action} leaves the two-point function of the smeared fields invariant:
\begin{eqnarray}
     \braket{\varphi(f)\varphi(h)} &=& \int_{-\infty}^\infty dt_1'dt_2'\, G(t_1'-t_2') f(t_1') h(t_2')  \nn \\
     &=& \int_{-\infty}^\infty dt_1dt_2\,  G(t_1-t_2)  \lb \chi_+(t_1) f(t_1) \rb \lb \chi_-(t_2) h(t_2)\rb \nn \\
     &=& \int_{-\infty}^\infty dt_1dt_2\,  G(t_1-t_2)  P_-\lb \chi_+(t_1) f(t_1) \rb P_-\lb \chi_-(t_2) h(t_2)\rb \nn\\
     &=& \braket{\varphi(P_-f')\varphi(P_+h')} = \braket{\varphi(f') \varphi(h') } 
\end{eqnarray}
where we have repeatedly used \eqref{gff_smeared_twopoint}.

\item    Since the one-point function of a GFF vanishes in the vacuum state and the higher-point functions factorize into two-point functions
\begin{eqnarray}\label{eq-Weyl-expectation}
    \braket{\varphi(f)} = 0\ , \qquad \braket{ \underbrace{\varphi(f)\cdots \varphi(f)}_\text{$2n$ times}} = \frac{(2n)!}{2^n n!} \braket{\varphi(f)\varphi(f)} \ ,
\end{eqnarray}
the correlation functions of the Weyl operators depend only the two-point function of the smeared fields
\begin{eqnarray}
    \braket{ W(f) }  &=&  e^{-\frac{1}{2}\braket{ \varphi(f) \varphi(f) } }\nn \\
     \braket{W(f)W(h)} &=& e^{-\frac{1}{2}[\varphi(f), \varphi(h) ] } \braket{W(f+h)}
\end{eqnarray}
where the commutator $[\varphi(f), \varphi(h) ]$ is the antisymmetric part of the two-point function of the smeared fields. Thus the invariance of the two-point function of the smeared fields also implies the invariance of the correlation functions of the Weyl operators under the action in \eqref{universal_psl2r_action}. Since all the correlation functions of the Weyl operators are invariant, the action preserves the vacuum state and can be represented unitarily on the Weyl algebra. The assertion then follows using Lemma \ref{lemma:psl2r_rep}.

\end{enumerate}

\end{proof}

\subsection{Proof for Theorem \ref{thm-loc-mobcov-net-integer-delta}}\label{proof:thm-loc-mobcov-net-integer-delta}
\begin{proof}

We start by first establishing Haag's duality for each algebra $\mA_I$. Let $W(f) \in \mA_{I}$. As discussed in section \ref{sec-gff-comm-choice}, Weyl operators $W(f)$ and $W(g)$ commute if and only if the commutator of smeared fields satisfy $[\varphi(f),\varphi(g)] = 2\pi i n$ for some integer $n$. However, there cannot exist any function $g(t)$ such that $[\varphi(f),\varphi(g)] = 2\pi i n$ (where $n$ is non-zero integer) for all $f(t)$ supported in $I$. This can be seen simply by doing scalar multiplication, i.e.  $f(t) \to \alpha f(t)$ for $\alpha \in \mathbb{R}$, which does not change the support of $f(t)$ but changes $[\varphi(f),\varphi(g)] \to i2\pi \alpha n \,$. Therefore, we find that $W(g) \in \mA'_{I}$ if and only if $[\varphi(f),\varphi(g)] = 0$ for all $f(t)$ supported in $I$. Therefore, $W(g) \in \mA'_{I}$ if and only if
\begin{eqnarray}
   0 &=& [\varphi(f) , \varphi(g) ]  =   \int_{-\infty}^{\infty}  d\omega\,  \omega^{2\Delta-1}  f(-\omega)  g(\omega) \nn \\
   &=&  \int_{-\infty}^{\infty} dt  \left( \partial_{t}^{2\Delta-1} f(t) \right)  g(t) \nn \\
   &=&  \int_{I} dt  \left( \partial_{t}^{2\Delta-1} f(t) \right)  g(t) 
\end{eqnarray}
where we have used the fact that $f(t)$ has support in the interval $I$. Since $f(t)$ is arbitrary, the above integral can only vanish if $g(t)$ vanishes in the interval $I$. We thus conclude that $g(t)$ is supported on the complement $I^{c}$. Hence for integer $\Delta$, each algebra in the map $I\to\mA_I$ satisfies Haag's duality 
\begin{eqnarray}\label{eq-W-comm-integer}
    (\mA_I)'  =  \mA_{I^{c}} \ . 
\end{eqnarray}
Similarly we find that 
\begin{eqnarray}
    (\mA_{I^{c}})'  =  \mA_I \ ,\qquad \mA_I  =  (\mA_I)'' 
\end{eqnarray}
and hence, $\mA_I$ is a von Neumann algebra.

Now consider two intervals $I_{1} \subset I_{2} $. Then it is clear that $\mA_{I_{1}} \subset \mA_{I_{2}} $ satisfying isotony. If $I_{1}$ and $I_{2}$ are such that $I_{2} \subset I_{1}^{c}$, then isotony implies that
\begin{align}
    \mA_{I_{2}} \subset \mA_{I_{1}^{c}} = (\mA_{I_{1}})'\ .
\end{align}
Therefore the map $I \to \mA_I$ defines a local net of von Neumann algebras. Now consider the unitary representation of $PSL(2,\mbR)$ on the net. When $\Delta$ is an integer, this $PSL(2,\mathbb{R})$ transformation in \eqref{eq-psl2r-rep-loc} acts locally. That is, if a $PSL(2,\mathbb{R})$ transformation maps interval $I$ to $I'$, then \eqref{eq-psl2r-rep-loc} maps functions supported in $I$ to those supported in $I'$. Moreover, the vacuum state is invariant under the unitary representation of the $PSL(2,\mathbb{R})$ (see Lemma \ref{lemma-psl-unitary}.) This unitary representation is the positive energy representation since the generator of the translation (i.e., the Hamiltonian) is positive. This establishes that the map $I \to \mA_I$ is a local M\"obius covariant net of von Neumann algebras and the net satisfies Haag's duality. 
\end{proof}

\subsection{Proof for Proposition \ref{lemma-support-1}}\label{proof:lemma-support-1}
\begin{proof}
    \begin{enumerate}
        \item Let $f(t)$ has support in $(a,\infty)$. Then 
        \begin{eqnarray}
            \Theta(t) \ast f(t)  =  \int_{-\infty}^{\infty} dt' \Theta(t-t') f(t')  =  \int_{a}^{\infty} dt'\Theta(t-t') f(t')  .
        \end{eqnarray}
        Since $\Theta(t-t') = 0$ for $t<t'$, the above integral vanishes for $t<a$. Thus, $\Theta(t) \ast f(t)$ has support in $(a,\infty)  $.

        Now let $\Theta(t) \ast f(t)$ has support in $(a,\infty)  $. Then for all $t < a$
        \begin{eqnarray}
            \int_{-\infty}^{\infty} dt' \Theta(t-t') f(t')  =  \int_{-\infty}^{t} dt' f(t')  =  0\ .
        \end{eqnarray}
        This can only be satisfied if $f(t) = 0$ for all $t < a$. Therefore, $f(t)$ has support in $(a,\infty)  $.
        
        \item This can be proved similarly. 
        \item This is just a special case of points $1$ and $2$, and hence, easily follows from these points.
    \end{enumerate}
\end{proof}

\subsection{Proof for Theorem \ref{thm-comm-I-general}}\label{proof:thm-comm-I-general}
\begin{proof}
We will follow similar steps as in the proof for Theorem \ref{thm-comm-half-line-general} by showing inclusions in both directions. Let $W(f)\in\mM_{(a,b)}$. From the structure of the Weyl algebra \eqref{weyl-algebra}, we can assume without loss of generality that $W(g) \in \lb \mT_\Delta^\dagger \mM_{(-\infty,a)} \mT_\Delta \rb \vee \lb \mT_\Delta \mM_{(b,\infty)} \mT_\Delta^\dagger \rb$ with $g(t)$ of the form given by
\begin{eqnarray}\label{eq-g-ansatz-comm-I}
    g(t)  =  H^{-1}_{\Delta}(g_{1})(t)  + H_{\Delta}(g_{2})(t)  
\end{eqnarray}
where $g_{1}(t)$ has support in $(-\infty,a)$ and $g_{2}(t)$ has support in $(b,\infty)$. Using \eqref{piplusminus}, we can see that the commutator \eqref{commutatorfg} for the smeared fields vanishes
\begin{eqnarray}
    [\varphi(f),\varphi(g)] 
    &=&  \int_{-\infty}^\infty d\omega\, \omega |\omega|^{2(\Delta-1)} \lb e^{i\pi\Delta\sgn(\omega)} g_1(\omega) + e^{-i\pi\Delta\sgn(\omega)} g_2(\omega)  \rb f(-\omega) \nn \\
    &=&  \int_{-\infty}^\infty dt\, \lb \p_t(\theta_- \ast g_1)(t) + \p_t(\theta_+ \ast g_2)(t)  \rb f(t) = 0
\end{eqnarray}
since by Proposition \ref{lemma-support-1}, $\theta_- \ast g_1$ is supported only in $(-\infty, a)$ and $\theta_+ \ast g_2$ is supported only in $(b,\infty)$ and therefore each term in the integral vanishes. Thus we have the inclusion $\lb \mT_\Delta^\dagger \mM_{(-\infty,a)} \mT_\Delta \rb \vee \lb \mT_\Delta \mM_{(b,\infty)} \mT_\Delta^\dagger \rb \subseteq (\mM_{(a,b)})'$. To see the reverse inclusion assume that $W(f) \in \mM_{(a,b)}$ and again without loss of generality let $W(g_1 + g_2) \in (\mM_{(a,b)})'$. Then 
\begin{eqnarray}
    0 &=& \lbrack \varphi(f), \varphi(g_1+g_2)\rbrack =  \int_{-\infty}^\infty d\omega\,\omega  |\omega|^{2(\Delta-1)} \lb  g_1(\omega) + g_2(\omega)\rb f(-\omega) \nn \\
    &=&  \int_{-\infty}^\infty d\omega\, \omega   \lb |\omega|^{2(\Delta-1)} e^{i\pi\Delta\sgn(\omega)} \rb \lb e^{-i\pi\Delta\sgn(\omega)}g_1(\omega)\rb f(-\omega) \nn \\ 
    && +  \int_{-\infty}^\infty d\omega\, \omega   \lb |\omega|^{2(\Delta-1)} e^{-i\pi\Delta\sgn(\omega)} \rb \lb e^{i\pi\Delta\sgn(\omega)}g_2(\omega)\rb f(-\omega) \nn \\ 
    &=& - \int_{-\infty}^\infty dt\, \lb  (\theta_- \ast H_\Delta (g_1))(t) + (\theta_+ \ast H_\Delta^{-1} (g_2))(t) \rb \p_t f(t)\ .
\end{eqnarray}
Again by using Proposition \ref{lemma-support-1}, we see that since the integral vanishes, $H_\Delta (g_1)$ has support in $(-\infty, a)$ and $H_\Delta^{-1} (g_2)$ has support in $(b,\infty)$.  Hence using \eqref{GHT3}, we have that $W(H_\Delta (g_1)) = \mT_\Delta W(g_1) \mT_\Delta^\dagger \in\mM_{(-\infty, a)}$ and $W(H_\Delta^{-1} (g_2)) = \mT_\Delta W(g_2) \mT_\Delta^\dagger \in\mM_{(b, \infty)}$. Therefore we have shown the reverse inclusion $(\mM_{(a,b)})' \subset\lb \mT_\Delta \mM_{(-\infty,a)} \mT_\Delta^\dagger \rb \vee \lb \mT_\Delta \mM_{(b,\infty)} \mT_\Delta^\dagger \rb $ which proves the equality.  

To see that $\mM_{(a,b)}$ is a von Neumann algebra, we take the double commutant
\begin{eqnarray}
    (\mM_{(a,b)})''  = \lb \mT_\Delta^\dagger \mM_{(-\infty,a)} \mT_\Delta \rb' \wedge \lb \mT_\Delta \mM_{(b,\infty)} \mT_\Delta^\dagger \rb'\ .
\end{eqnarray}
Now using Theorem \ref{thm-comm-half-line-general}, we get
\begin{eqnarray}
    \mT_\Delta^\dagger \mM_{(-\infty,a)} \mT_\Delta &=& (\mM_{(a,\infty)})' \nn\\
    \mT_\Delta \mM_{(b,\infty)} \mT_\Delta^\dagger &=& (\mM_{(-\infty,b)})'
\end{eqnarray}
which are both von Neumann algebras. Hence
\begin{eqnarray}
     (\mM_{(a,b)})''  = \mM_{(a,\infty)} \wedge \mM_{(-\infty,b)} = \mM_{(a,b)}
\end{eqnarray}
and is therefore a von Neumann algebra.
\end{proof}

\subsection{Proof for Corollary \ref{cor:local-net-integer-delta}}\label{proof:cor-local-net-integer-delta}
\begin{proof}
    Note that by definition, we have
    \begin{align}
        \mA_{I_{1}} \subset \mA_{I_{2}} 
    \end{align}
    if $I_{1} \subset I_{2} $. Thus we have isotony and the map $I \to \mA_{I}$ defines a net of von Neumann algebras.
    However, for non-integer $\Delta$ locality is violated. From the definition of locality, it is enough to show that $\mA_{\mathbb{R}^{-}} \not\subseteq (\mA_{\mathbb{R}^{+}})' $. Using Theorem \ref{thm-comm-half-line-general}, we know that $(\mA_{\mathbb{R}^{+}})' = \mT_\Delta^\dagger \mA_{\mathbb{R}^{-}} \mT_\Delta$. 
   However we cannot write any function $f(t)$ with support in $\mathbb{R}^{-}$ as $H_{\Delta}^{-1}(g)(t)$ with support of $g(t)$ in $\mathbb{R}^{-1}$ since
    \begin{eqnarray}
        H_{\Delta}^{-1}(g)(t)  &=& \cos(\pi\Delta) g(t) - \sin(\pi\Delta) \, H(g)(t) \nn\\
        &=&\cos(\pi\Delta)  g(t) + \sin(\pi\Delta) \int_{-\infty}^{0} dt' \, \frac{g(t')}{t-t'}
    \end{eqnarray}
    does not vanish for $t>0$. Thus, locality is violated.

    Now we discuss whether the net is M\"obius covariant. 
    Consider two intervals $I$ and $I_{c}$ which are related by a special conformal transformation $\mathcal{S}_{c}$ 
    \begin{align}
            S_c = t \to t' = \frac{t}{1+ct} \, .
    \end{align}
    As discussed in Section \ref{sec-Mobius}, a special conformal transformation acts on the positive and negative frequency part of a function with a different phase. As a result, the action of a $\mathcal{S}_{c}$ is generally non-local in the sense that functions supported on $I$ do not generally map to functions supported on $I_{c}$. This implies that the algebra $\M_{I_{c}}$ and $\M_{I}$ are not generally related by the unitary representation of $\mathcal{S}_{c}$, and hence, the net is not M\"obius covariant.

    To make the discussion less abstract, let us present an example where the special conformal transformation is non-local. For concreteness, let us consider the special conformal transformation with $c=-1$ and let us consider a function $f_I(t)$ with support in the $I =(1,2)$. This interval maps to $I_{-1} = (-\infty,-2)$ under $\mathcal{S}_{-1}$. Now using \eqref{universal_psl2r_action}, we find that the action of $\mathcal{S}_{-1}$ on $f_I(t)$ is given by
    \begin{align}
            \mS_{-1} : f_{I}(t) \to f'_{I}(t) =& P_{+}\left(|1+t|^{2(\Delta-1)} e^{i2\pi\Delta \Theta(-1-t)} f_{I}\left(\frac{t}{1+t} \right)\right) \nn\\ +& P_{-}\left(|1+t|^{2(\Delta-1)} e^{-i2\pi\Delta \Theta(-1-t)} f_{I}\left(\frac{t}{1+t} \right)  \right) \, .
    \end{align}
    Since $f_{I}(t)$ has support in $I = (1,2)$, we deduce that $f_{I}(t/(1+t))$ is non-zero only for $t<-2$. In this range, $\Theta(-1-t) = 1$, and hence, the above expression simplifies to
    \begin{align}
        f'_{I}(t) =& e^{i2\pi\Delta} P_{+}\left((-1-t)^{2(\Delta-1)}  f_{I}\left(\frac{t}{1+t} \right)\right) + e^{-i2\pi\Delta} P_{-}\left((-1-t)^{2(\Delta-1)}  f_{I}\left(\frac{t}{1+t} \right)  \right) \, \nn
    \end{align}
    or equivalently, 
    \begin{align}
f'_{I}(t)
        =& \, H_{2\Delta}^{-1}\Big[ (-1-t)^{2(\Delta-1)}  f_{I}\left(\frac{t}{1+t}\right) \Big] \ ,
    \end{align}
    where $H_{2\Delta}^{-1}$ is the inverse GHT. Since GHT and its inverse are non-local transformations, the support of $f'_{I}(t)$ is not in $I_{-1} = (-\infty,-2) $, which implies that $S_{-1}$ does not map functions supported in $I = (1,2)$ to functions supported in $I_{-1} = (-\infty,-2)$.  
    \end{proof}

\subsection{Proof for Proposition \ref{prop-unitary-related-algebra}}\label{proof:prop-unitary-related-algebra}
\begin{proof}
    Let $S_{\mA_2}$ be the Tomita operator for $\mA_2$. For $a_2 \in \mA_2$, we have
    \begin{align}
        S_{\mA_2} a_2 \ket{\Omega} = a^\dagger_2 \ket{\Omega} \, .
    \end{align}
    Let $a_2 = U a_1 U^\dagger$ where $a_1 \in \mA_1$. With this, the above equation becomes
    \begin{align}
        S_{\mA_2} U a_1 \ket{\Omega} = U a^\dagger_1 \ket{\Omega} 
    \end{align}
    where we have used $U^\dagger \ket{\Omega} = \ket{\Omega}$. In terms of the Tomita operator of $\mA_{1}$, the above equation becomes
    \begin{align}
        S_{\mA_2} U a_1 \ket{\Omega} =  U S_{\mA_1} a_1 \ket{\Omega} ,
    \end{align}
    which implies 
    \begin{align}
        S_{\mA_1} = U^\dagger S_{\mA_2} U \, .
    \end{align}
    Doing a polar decomposition of the above equation yields the desired result.
\end{proof}

\subsection{Proof for Theorem \ref{thm-interval-mod}} \label{sec-proof-interval-mod}
\begin{proof}
We have found the unitary $U_I$ that maps $\mA_I$ to $\mA_{\mathbb{R}^+}$ in Lemma \ref{lemma-interval-map}. The modular flow of $\mA_I$ is $\Delta_I^{is} = U_I \Delta_{\mathbb{R}^+}^{is} U_I^\dagger$. We have proved in appendix  \ref{app:proof_lemma_psl2r_rep} that the actions form a representation of $\widetilde{PSL}(2,\mathbb{R})$. The composition of actions is similar to matrix multiplication but with an extra constant phase $e^{\pm i 2\pi\Delta \psi}$. $\Delta_I^{-is} = U_I \Delta_{\mathbb{R}^+}^{-is} U_I^\dagger$ corresponds to the matrix multiplication 
\begin{align}
    \begin{pmatrix}
         \frac{1}{\sqrt{q-p}}&\frac{-p}{\sqrt{q-p}}\\-\frac{1}{\sqrt{q-p}}&\frac{q}{\sqrt{q-p}}
    \end{pmatrix}^{-1}
    \begin{pmatrix}
        e^{\pi s} & 0 \\ 0 & e^{-\pi s} 
    \end{pmatrix}
    \begin{pmatrix}
         \frac{1}{\sqrt{q-p}}&\frac{-p}{\sqrt{q-p}}\\-\frac{1}{\sqrt{q-p}}&\frac{q}{\sqrt{q-p}}
    \end{pmatrix}=
    \frac{1}{q-p}
    \begin{pmatrix}
        {qe^{\pi s}-pe^{-\pi s}}& {qp(e^{-\pi s}-e^{\pi s})} \\ {e^{\pi s}-e^{-\pi s}}& {qe^{-\pi s}-pe^{\pi s}}
    \end{pmatrix}
\end{align}
with an extra phase 
\begin{align}
    \psi =& \, \sgn(e^{\pi s}-e^{-\pi s}) \Theta((e^{\pi s}-e^{-\pi s}) t-(qe^{\pi s}-pe^{-\pi s})) \nonumber\\
    &-\sgn(1) \Theta(t-q) - \text{sgn}(-1) \Theta\lb\tfrac{(e^{\pi s}-e^{-\pi s}) t-(qe^{\pi s}-pe^{-\pi s})}{q-t}\rb \, .
\end{align}
The phase is independent of $t$. We can put $t=0$ to simplify it. We have
\begin{align}
    \psi = \, \sgn(s) \Theta(-(qe^{\pi s}-pe^{-\pi s}))-\Theta(-q) + \Theta\lb\tfrac{-(qe^{\pi s}-pe^{-\pi s})}{q}\rb.
\end{align}
We only have to consider $4$ cases, which are when $qe^{\pi s}-pe^{-\pi s}< 0$ and when $q >0$. If $qe^{\pi s}-pe^{-\pi s}> 0$, $\psi = -\Theta(-q) + \Theta(-q) = 0$. If $qe^{\pi s}-pe^{-\pi s}< 0$, we have $\psi = \,\sgn(s) + \sgn(q)$. If $q>0$, $qe^{\pi s}-pe^{-\pi s}< 0$ implies $s<0$, and vice versa. Thus, there is no extra phase in this case. Therefore, $\Delta_I^{is}$ has action
\begin{eqnarray}
       &&\Delta_I^{-is}\varphi(f)\Delta_I^{is} = \varphi(f')\ , \nn\\
        &&f_\pm'(t) = \left[\left(\tfrac{(q-t)e^{\pi s}+(t-p)e^{-\pi s}}{q-p} \mp\tfrac{e^{\pi s}-e^{-\pi s}}{q-p}i\epsilon\right)^{2\Delta-2} f\left(\tfrac{(q-t)p\,e^{\pi s} + (t-p)q\,e^{-\pi s}}{(q-t)e^{\pi s} + (t-p)e^{-\pi s}}\right)\right]_\pm.
    \end{eqnarray}
\end{proof}

\subsection{Proof for Corollary \ref{thm-J_interval}}\label{app:proof-thm-J_interval}
\begin{proof}
Consider the unitary $U_I$ corresponding to the conformal map
\begin{eqnarray}
    U_I^\dagger \mA_{(p,q)} U_I = \mA_{\mbR^+}\ .
\end{eqnarray}
It's action on Weyl operators is $U_I^\dagger W(f) U_I = W(f')$ where
\begin{eqnarray}
    (P_\pm f')(t) = P_\pm \lb \left|\frac{t + 1}{\sqrt{q-p}} \right|^{2\Delta-2} e^{\pm i 2\pi\Delta\Theta(-t-1)} f\left(\frac{qt+p}{t+1}\right) \rb\ .
\end{eqnarray}
The inverse action on Weyl operators is $U W(f') U^\dagger = W(f)$ where
\begin{eqnarray}
    (P_\pm f)(t) = P_\pm \lb \left|\frac{q-t}{\sqrt{q-p}} \right|^{2\Delta-2} e^{\mp i 2\pi\Delta\Theta(t-q)} f'\left(\frac{t-p}{q-t}\right)  \rb\ .
\end{eqnarray}
Using Proposition \ref{prop-unitary-related-algebra}, the modular conjugation operator corresponding to the time interval algebra $\mA_{(p,q)}$ is
\begin{eqnarray}\label{unitary_rot_modular_conj}
    J_I = U_I J_\mbR^+ U_I^\dagger = - U_I \mT_\Delta^\dagger \mR U_I^\dagger = - \mT_\Delta^\dagger U_I \mR U_I^\dagger 
\end{eqnarray}
where we have used that the GHT commutes with $U_I$. The action of $U_I \mR U_I^\dagger$ on the Weyl operators can be decomposed as
\begin{eqnarray}\label{interval_reflection}
    f(t) &\xrightarrow{U_I^\dagger}& f'(t)= \sum_\pm \left[ \left|\frac{t + 1}{\sqrt{q-p}} \right|^{2\Delta-2} e^{\pm i 2\pi\Delta\Theta(-t-1)} f\left(\frac{qt+p}{t+1}\right) \right]_\pm \nn \\
    &\xrightarrow{\mR}&  f'_R(t)= \sum_\pm \left[ \left|\frac{1-t}{\sqrt{q-p}} \right|^{2\Delta-2} e^{\mp i 2\pi\Delta\Theta(t-1)} f\left(\frac{p-qt}{1-t}\right) \right]_\pm  \nn \\
    &\xrightarrow{U_I}&  f''(t) = \sum_\pm \left[ \left|\frac{p+q-2t}{q-p} \right|^{2\Delta-2} e^{\mp i 2\pi\Delta\Theta(2t- (p+q))} f\left(\frac{2pq-(p+q)t}{p+q-2t}\right) \right]_\pm 
\end{eqnarray}
where in the second line we have used $P_\pm \circ R = R \circ P_\mp$ for the reflection map $R$ and in the third line we used $\Theta(t-q)+\Theta(\frac{2t-(p+q)}{q-t}) = \Theta(2t- (p+q))$.
Note that in \eqref{interval_reflection}, for $2t < p+q$, there is no phase implying $f''(t) = R_I[f](t)$. However for $2t >p+q$, there is a phase of $e^{\mp i2\pi\Delta}$ corresponding to the action of $\mT_\Delta^2$. Combining this with $-\mT_\Delta^\dagger$ in \eqref{unitary_rot_modular_conj} completes the proof. This matches with our previous result for the commutant of $\mA_{(p,q)}$ in Theorem \ref{thm-comm-I-general}.
\end{proof}

\subsection{Proof for Corollary \ref{corr-inv-in-terms-of-Js}} \label{app-proof-inv-corr}
\begin{proof}
 Due to Proposition \ref{prop-unitary-related-algebra} and Theorem \ref{thm-J_half-line}, we get
    \begin{eqnarray}
        J_{I}J_{\mathbb{R}^+}  &=&  U_IJ_{\mathbb{R}^+} U_I^\dagger J_{\mathbb{R}^+}  \nn\\
        &=&  U_I \mT_{\Delta}^\dagger \mR U_I^\dagger \mT_{\Delta}^\dagger \mR 
    \end{eqnarray}
    where $U_I$ maps the interval the algebra $\mA_{(-1,1)}$ to $\mA_{\mathbb{R}^{+}}$. Now using \eqref{eq-comm-T-R} and the fact that GHT commutes with $U_I$, we find that $\mT_{\Delta}$ drops out from the above equation and we get
    \begin{eqnarray}
        J_{I}J_{\mathbb{R}^+} =  U_I  \mR U_I^\dagger \mR \ .
    \end{eqnarray}
    Therefore, we get
    \begin{eqnarray}
        J_{I}J_{\mathbb{R}^+} W(f) (J_{I}J_{\mathbb{R}^+})^\dagger  =  U_I  \mR U_I^\dagger \mR  W(f) \mR U_I  \mR U_I^\dagger  \ .\label{eq-inver-corr-int-1}
    \end{eqnarray}
    We now gather some results that will allow us to simplify the above equation. First, by taking $q=-p=1$ in \eqref{eq-U-act}, we find that
    \begin{eqnarray}
        U_I^\dagger W(f) U_I =  W(f') 
    \end{eqnarray}
    where
    \begin{eqnarray}
                (P_{\pm}f')(t) = P_{\pm} \left[ \left|\frac{t+1}{\sqrt{2}}\right|^{2(\Delta-1)} e^{\pm i2\pi \Delta \Theta(-t-1)} (P_{\pm} f)\left(\frac{t-1}{t+1}\right)  \right] \ .
    \end{eqnarray}
    Similarly, we can find the inverse map so that
    \begin{eqnarray}
        U_I W(f) U_I^\dagger =  W(f')
    \end{eqnarray}
    where
    \begin{eqnarray}
                (P_{\pm}f')(t) = P_{\pm} \left[ \left|\frac{t-1}{\sqrt{2}}\right|^{2(\Delta-1)} e^{\mp i2\pi \Delta \Theta(t-1)} (P_{\pm} f)\left(\frac{t+1}{1-t}\right)  \right] \ .
    \end{eqnarray}
    Moreover, the reflection operator acts as
    \begin{eqnarray}
        \mR W(f) \mR &=& W(f')  \nn \\
        f'(t) = R[f](t) &=& f(-t)  
    \end{eqnarray}
    where $R$ is the reflection map. Note that $R$ satisfies the following commutation relation with the projectors $P_{\pm}$:
    \begin{eqnarray}
        R \circ P_{\pm} = P_{\mp} \circ R \ .
    \end{eqnarray}
    By combining these results, we can write \eqref{eq-inver-corr-int-1} as
    \begin{eqnarray}
        J_{I}J_{\mathbb{R}^+} W(f) (J_{I}J_{\mathbb{R}^+})^\dagger  =  W(f') 
    \end{eqnarray}
    where
    \begin{eqnarray}
        (P_{\pm} f')(t)  =  P_{\pm} \left[ |t|^{2(\Delta-1)} e^{\mp 2\pi i \Delta \Theta(t)}(P_{\pm}f)(-1/t) \right] \  .
    \end{eqnarray}
    This precisely matches the inversion map in \eqref{eq-inv-map}. 
\end{proof}

\section{Geometry of AdS$_2$}\label{app:ads_geometry}
In the embedding coordinates, AdS$_2$ is the surface
\begin{eqnarray}
  {\bf g}_{AB} X^AX^B=-(X^{-1})^2-(X^0)^2+(X^1)^2=-1
\end{eqnarray}
where ${\bf g}_{AB} =\text{diag}(-1,-1,1)$. The isometry group is $O(1,2)$ with four disconnected components; $\pi_0(O(1,2))=\mathbb{Z}_2\times \mathbb{Z}_2$. Its fundamental group is $\mathbb{Z}$. 
\paragraph{Global coordinates:}
The global coordinates are
\begin{eqnarray}
    &&X^{-1}=\cosh\mu \cos\tau, \qquad X^0=\cosh\mu \sin\tau, \qquad X^1=\sinh\mu \nn \\
    && ds^2 = -\cosh^2\mu\, d\tau^2 + d\mu^2
\end{eqnarray}
where $\tau \in \mbR$ and $\mu\in (-\infty,\infty)$. Using the change of variables $\sinh\mu=\tan\rho$ makes it clear that AdS$_2$ with the metric
\begin{eqnarray}
&&X^{-1}=\frac{\cos\tau}{\cos\rho},\qquad X^0=\frac{\sin\tau}{\cos\rho},\qquad X^1=\tan\rho \nn \\
&&ds^2 = \frac{1}{\cos^2\rho} (-d\tau^2+d\rho^2)\ .
\end{eqnarray}
has two boundaries at $\rho=\pm \pi/2$. Thus AdS$_2$ is conformally flat and we have the Penrose diagram as shown in Figure \ref{fig:ads2_penrose} (a).

The connected piece of the isometry group is $SO^+(1,2)\simeq PSL(2,\mathbb{R})$ and the Lie-algebra generators are related by 
\begin{eqnarray}
Q^A=\frac{1}{2}\epsilon^{ABC}J_{BC},\qquad [Q^A,Q^B]=\epsilon^{ABC} {\bf g}_{CD}Q^D
\end{eqnarray}
The charge $Q^{1}$ corresponds to rotations in $(X^{-1},X^0)$-direction, and corresponds to time-translations in global coordinates $\p_\tau$. Similarly, $Q^0$ generates rotations in $(X^{-1},X^1)$-plane and $Q^{-1}$ generates a rotation in $(X^0,X^1)$-plane.
Since $Q^A$ generates a translation that fixes the vector $X^A$, for every vector $V^A$ we can define the charge $Q_V=V_A Q^A$ that fixes the vector $V^A$. The generators for the isometries can be written as \cite{lin2019symmetries}
\begin{eqnarray}\label{isometries_embedding}
    J_{-1,0} =& i \lb -X^{-1} \p_{X^0} + X^0 \p_{X^{-1}}  \rb &=- i \p_\tau\nn \\
    J_{-1,1} =& -i(X^{-1} \p_{X^1}+X^1 \p_{X^{-1}} ) &=  i \lb\sin\tau \sin\rho \,\partial_\tau -\cos\tau\cos\rho \, \partial_\rho  \rb  \nn \\
    J_{0,1} =& -i(X^0 \p_{X^1} + X^1 \p_{X^0})  &=-i\lb\cos\tau \sin\rho \p_\tau+ \sin\tau \cos\rho \p_\rho  \rb \ . 
\end{eqnarray}
The universal cover of AdS$_2$ decompactifies the time-like circle and a rotation by $J_{-1,0}$ generates a translation $\tau\to \tau+2\pi$ in global coordinates. 
Another commonly used form of the global coordinates is to substitute $r = \sinh\mu$ with the metric
\begin{eqnarray}\label{ads2-schwarzchild}
    ds^2 = -(1+r^2)d\tau^2 + \frac{1}{1+r^2} dr^2\ .
\end{eqnarray}

\begin{figure}[t]
   \centering
    \includegraphics[width=\textwidth]{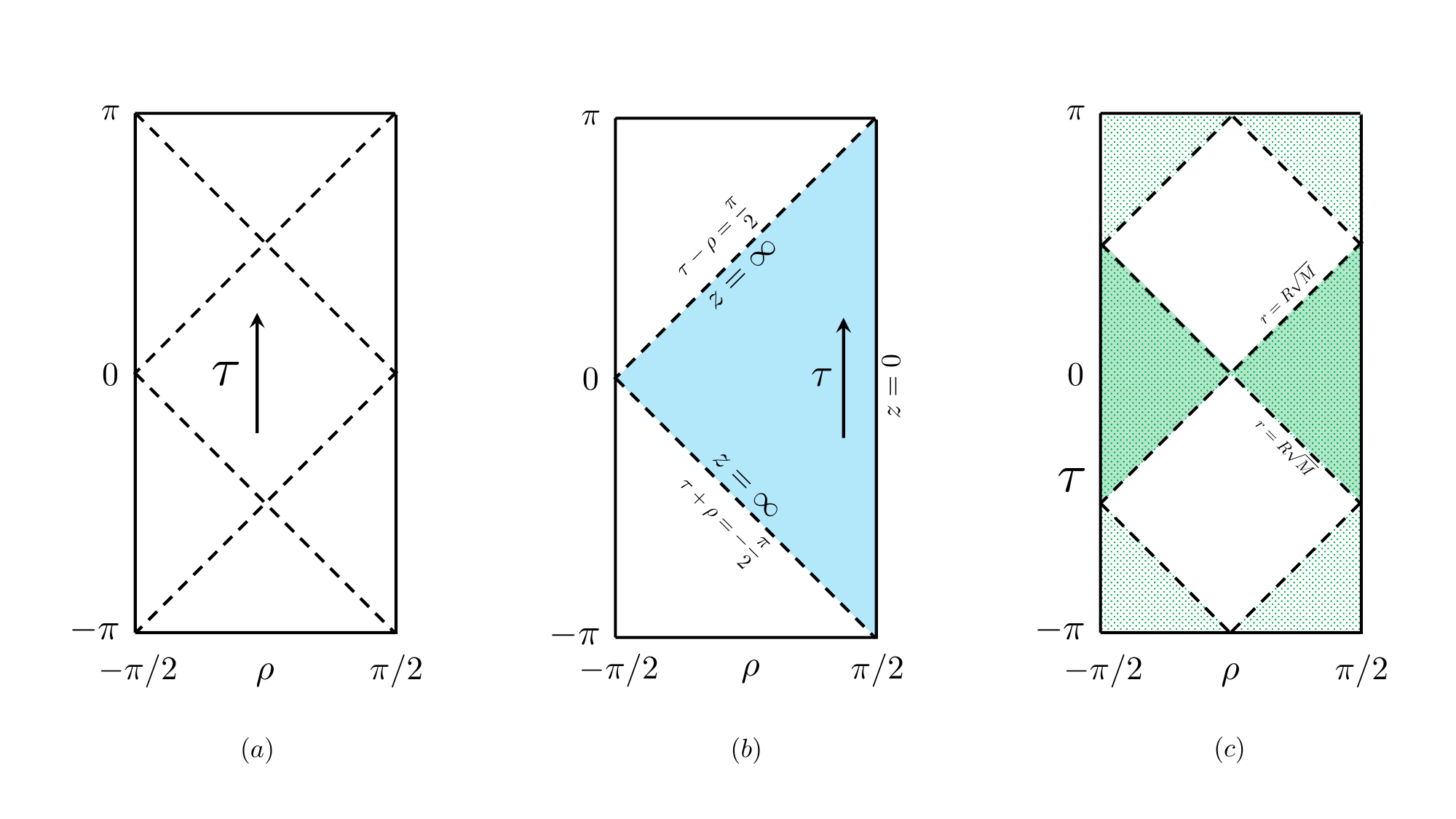}
    \caption{\small (a) Penrose diagram for AdS$_2$ spacetime in global coordinates. AdS$_2$ is conformally flat with light rays along 45$^\circ$ angles. (b) The Poincar\'e patch of AdS$_2$ is shaded in blue. (c) The AdS$_2$ Rindler region is shaded in green. }
    \label{fig:ads2_penrose}
\end{figure}

\paragraph{Poincar\'e coordinates:}
\begin{figure}[ht]
    \centering
    \includegraphics[width=.4\textwidth]{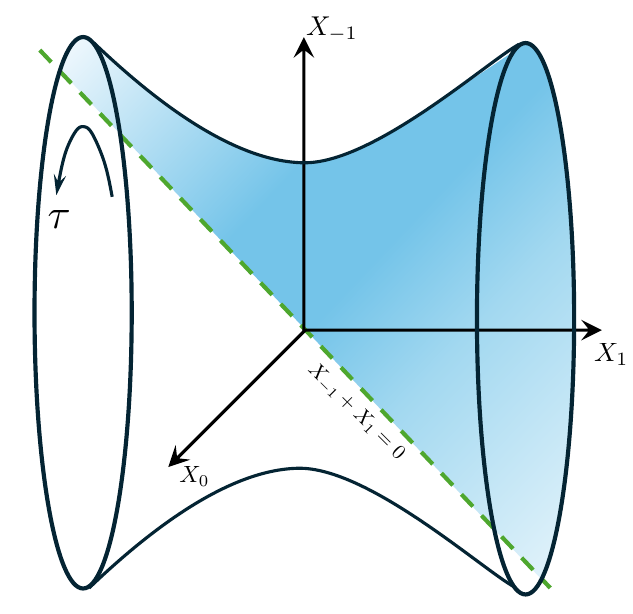}
    \caption{\small The hyperboloid corresponding to AdS$_2$ in embedding coordinates. The blue shaded region corresponds to $X_{-1}+X_1>0$ describing the Poincar\'e patch.}
    \label{fig:poincare_patch}
\end{figure}
We can describe the Poincar\'e patch using the Poincar\'e coordinates
\begin{eqnarray}\label{poincare_coord_ads2}
    &&X^{-1} = \frac{1-t^2+z^2}{2z},\qquad X^0 = \frac{t}{z},\qquad X^1= \frac{1+t^2-z^2}{2z} \nn \\
    &&ds^2 = \frac{-dt^2+dz^2}{z^2}\ ,\qquad t\in \mbR , \qquad z\geq 0\ .
\end{eqnarray}
The Poincar\'e coordinates can be inverted to give
\begin{eqnarray}\label{poincare_to_global}
    t = \frac{X^0}{X^{-1}+X^1}= \frac{\sin\tau}{\cos\tau+\sin\rho},\qquad z = \frac{1}{X^{-1}+X^1} = \frac{\cos\rho}{\cos\tau+\sin\rho}\ .
\end{eqnarray}
The Poincar\'e patch covers half of the hyperboloid with $X_{-1}+X_1>0$, see Figure \ref{fig:poincare_patch}. In terms of the global coordinates, this translates to $\cos\tau + \sin\rho >0$. Thus the Poincar\'e patch covers the region (see Figure \ref{fig:ads2_penrose} (b))
\begin{eqnarray}
    -\frac{\pi}{2}&<&\tau+\rho < \frac{3\pi}{2} \nn\\
    -\frac{3\pi}{2} &<& \tau-\rho < \frac{\pi}{2} \ .
\end{eqnarray}
The local isometries of the Poincar\'e patch are the same as the global AdS$_2$ isometries given by the Killing vectors
\begin{eqnarray}\label{isometries_poincare}
    J_{-1,0}+ J_{0,1} &=& -i \,\p_t \ , \nn \\
    J_{-1,1} &=&  i \lb t \, \partial_t  +z \,\partial_z\rb , \nn \\
    J_{-1,0}- J_{0,1} &=&-i\lb (t^2+z^2) \,\p_t + 2t z \,\p_z \rb\ .
\end{eqnarray}
The Poincar\'e patch is also sometimes described using the change in coordinate $r  = \frac{1}{z}$ with the metric
\begin{eqnarray}\label{ads2_poincare_2}
    ds^2 = -r^2dt^2 + \frac{1}{r^2} dr^2\ , \qquad r\geq 0\ .
\end{eqnarray}
For future reference, we also explicitly write the global coordinates in terms of the Poincar\'e coordinates 
\begin{eqnarray}
    \tau = \tan^{-1}\lb\frac{2t}{1 - t^2 +z^2} \rb \ ,\qquad \rho = \tan^{-1}\lb \frac{1+t^2-z^2}{2z}\rb\ ,
\end{eqnarray}
where $\text{sgn}(\tau) = \text{sgn}(t/z)$.
On the boundary $z=0$, the boundary Poincar\'e time $t$ is related to the global time by
\begin{eqnarray}\label{bdry_poincare_time}
    t = \tan\lb\frac{\tau}{2} \rb\ .
\end{eqnarray}
\begin{figure}[t]
    \centering
    \includegraphics[height=0.6\linewidth]{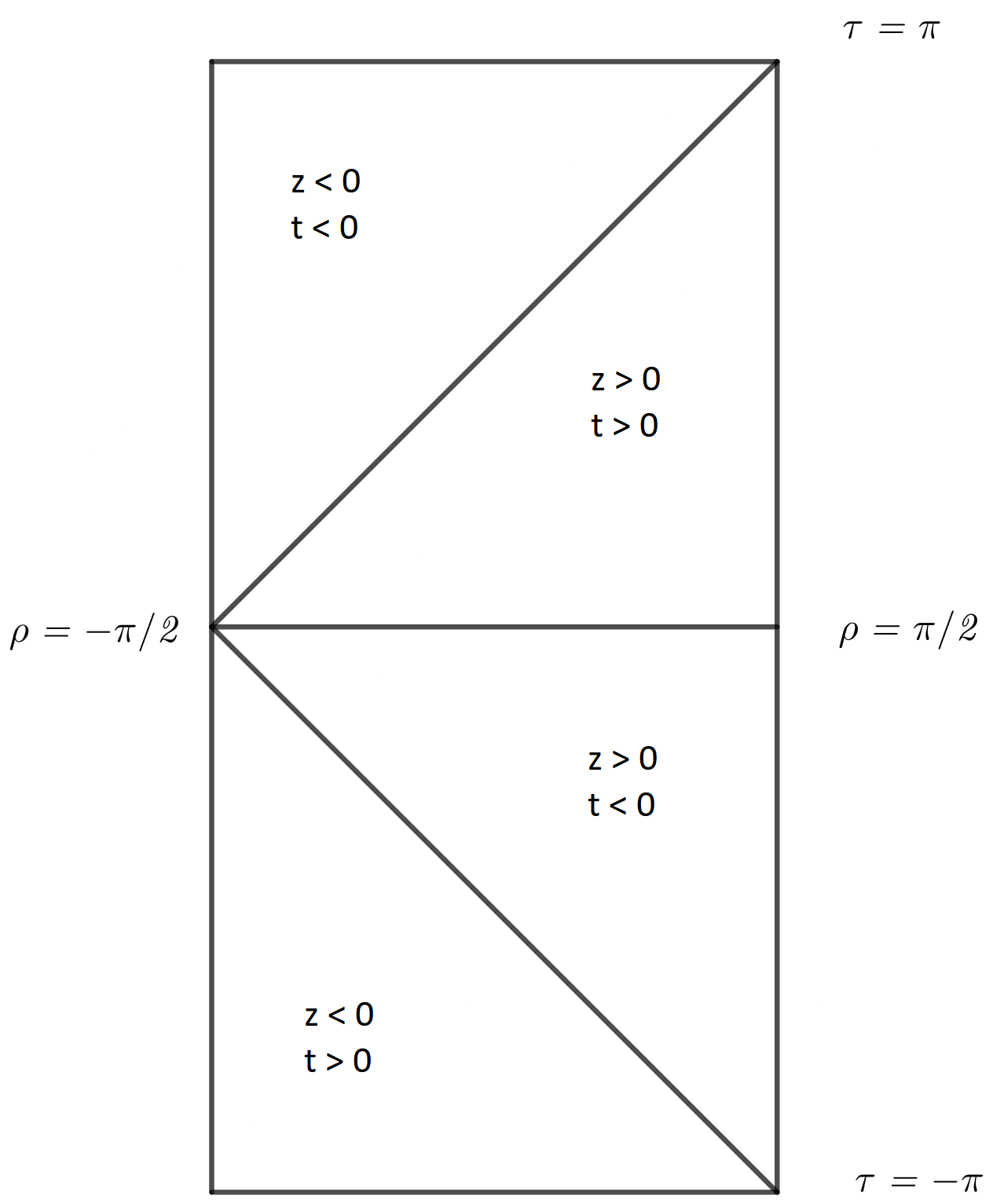}
    \caption{\small{The coordinate transformation is only valid for $z>0$. For $z<0$, it is still a bijective change of variable that allows us to label the points in the region.}}
    \label{fig:poincare}
\end{figure}
The expression in \eqref{poincare_to_global} is a valid coordinate transformation only in the Poincar\'e patch. However, it is still a well-defined change of variables and a bijective map for the region $\tau \in (-\pi,\pi)$. The $\tan^{-1}$ function gives two solutions if we neglect the sign. In fact, we can use the negative $z$ to label the region outside the Poincar\'e patch, see Figure \ref{fig:poincare}. To be explicit, we have
\begin{align}
    -\frac{\cos\rho}{\cos\tau+\sin\rho} = \frac{\cos(-\rho)}{\cos(\tau\pm\pi)+\sin(-\rho)}\nn\\
    \frac{\sin\tau}{\cos\tau+\sin\rho} = \frac{\sin(\tau\pm\pi)}{\cos(\tau\pm\pi)+\sin(-\rho)}\nn\\
    (t,z) \leftrightarrow (t,-z) \Leftrightarrow (\tau,\rho) \leftrightarrow (\tau\pm\pi, -\rho)\ .
\end{align}

\paragraph{Kruskal coordinates:} In Section \ref{sec:Twisted}, we have used the Kruskal coordinate system, defined as
\begin{eqnarray}
    u = \tan \frac{\tau+\rho}{2} \ , \qquad v = \tan\frac{\tau-\rho}{2} \ .
\end{eqnarray}
In this coordinate system, the metric has the form
\begin{eqnarray}
    ds^2 = \frac{-4 dudv}{(1+uv)^2} 
\end{eqnarray}
whereas, the generators of isometries are
\begin{eqnarray}
    J_{0,1} &=& -i(u\partial_{u} - v\partial_{v}) \nn\\
    J_{-1,0}-J_{-1,1} &=& -i (u^2 \partial_u + \partial_v) \nn\\
    J_{-1,0}+J_{-1,1} &=& -i ( \partial_u + v^2 \partial_v)  \ .
\end{eqnarray}

\paragraph{AdS$_2$ Rindler coordinates:}
The Rindler patch in AdS$_2$ is described by the coordinate transformation
\begin{eqnarray}
    X_{-1} = \frac{r}{\sqrt{M}},\quad X_0 = \sqrt{\frac{r^2}{M}-1} \sinh\sqrt{M}t_R, \quad X_1 = \sqrt{\frac{r^2}{M}-1} \cosh\sqrt{M}t_R 
\end{eqnarray}
with $r>\sqrt{M}$ and $t\in(-\infty,\infty)$. The metric simplifies to
\begin{eqnarray}\label{rindler-ads2}
    ds^2 = - \lb r^2 - M \rb dt_R^2 + \frac{1}{\lb r^2 - M \rb } dr^2 \ .
\end{eqnarray}
The global coordinates can be related to the Rindler coordinates as
\begin{eqnarray}
    \tanh\sqrt{M}t_R = \frac{\sin \tau}{\sin\rho},\qquad \frac{r}{\sqrt{M}} = \frac{\cos\tau}{\cos\rho}\ .
\end{eqnarray}
The AdS$_2$ Rindler patch covers the region where $\cos\tau > \cos\rho$ which is the shaded region in Figure \ref{fig:ads2_penrose} (c) and the horizons are labeled by $r=\sqrt{M}$. The Killing isometries of AdS$_2$ Rindler are
\begin{eqnarray}\label{isometries_rindler}
    J_{-1,0} &=&- i \lb \frac{r  \cosh\lb\sqrt{M}t_R \rb}{\sqrt{M}\sqrt{r^2-M}} \, \p_{t_R} - \sqrt{r^2-M} \sinh\lb\sqrt{M}t_R\rb \, \p_r \rb , \nn \\
    J_{-1,1} &=&  i \lb \frac{r \sinh\lb\sqrt{M}t_R \rb}{\sqrt{M}\sqrt{r^2-M}} \, \p_{t_R} - \sqrt{r^2-M} \cosh\lb\sqrt{M}t_R\rb \, \p_r \rb , \nn \\
    J_{0,1} &=&  -\frac{i}{\sqrt{M}}\,\p_{t_R} \ .
\end{eqnarray}
Another commonly used form for the Rindler AdS$_2$ is the metric
\begin{eqnarray}
    &&X_{-1} = \sqrt{r'^2 + 1},\qquad X_0 = r' \sinh t_R', \qquad X_1 = r' \cosh t_R' \nn \\
    &&ds^2 = - r'^2 dt_R'^2 + \frac{1}{1+ r'^2} dr'^2 \ , \qquad r'\geq 0 , \qquad t_R'\in \mbR\ .
\end{eqnarray}
The global coordinates can be related to them as
\begin{eqnarray}
    \tanh t_R' = \frac{\sin \tau}{\sin\rho},\qquad r' = \frac{1}{\cos\rho}\sqrt{\sin^2\rho - \sin^2 \tau}\ .
\end{eqnarray}
It covers the region where $\sin^2\rho > \sin^2\tau$ which is the dotted region in Figure \ref{fig:ads2_penrose} (c). Notably, this includes two copies of the Rindler wedge and the horizons are now labelled by $r'=0$.

\subsubsection*{Dilatation of forward lightcone in Poincar\'e AdS$_2$}

\begin{figure}
    \centering
    \includegraphics[width = 0.8\textwidth]{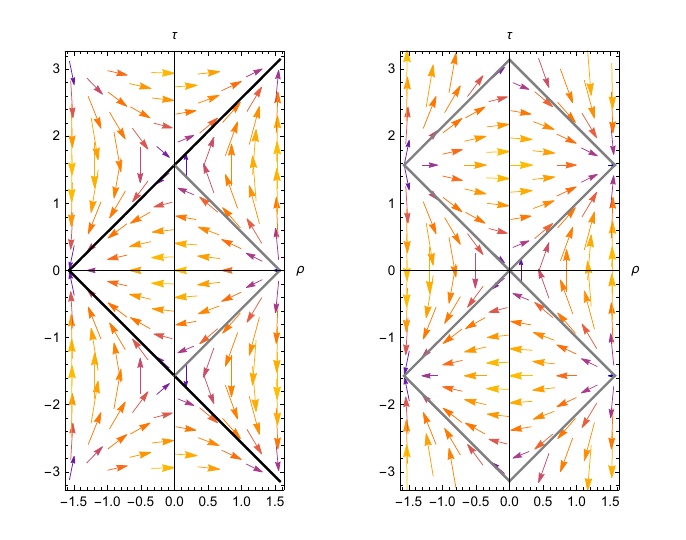}
    \caption{\small Vector plot for the generators of the AdS$_2$ isometries $J_{-1,1}$ (left) and $J_{1,0}$ (right).}
    \label{fig:ads2_poincare_modflow}
\end{figure}
We are interested in the Killing vector $J_{-1,1}$ that preserves the Poincar\'e patch, see Figure \ref{fig:ads2_poincare_modflow}. On the boundary $\rho=\frac{\pi}{2}$, the field becomes $J_{-1,1}|_{\rho=\frac{\pi}{2}} = i\sin\tau \,\p_\tau$. The integral curve is
\begin{eqnarray}\label{ads2_future_modflow}
    \tau(s) = 2 \tan^{-1}\lb e^s \tan\lb\frac{\tau(0)}{2}\rb  \rb\ .
\end{eqnarray}
Using \eqref{bdry_poincare_time} for the boundary Poincar\'e time, we see that the flow is a dilatation on the boundary
\begin{eqnarray}
    t(s) = e^s\, t(0)\ .
\end{eqnarray}
The full set of integral curves can be obtained by solving the set of coupled PDEs
\begin{eqnarray}
    &&\frac{d\tau(s)}{ds} = \sin(\rho(s)) \sin(\tau(s)) \ ,\qquad \frac{d\rho(s)}{ds} = -\cos(\rho(s)) \cos(\tau(s)) \nn \\
    &&\implies  \frac{d\tau}{d\rho} = -\tan\rho \tan\tau\ .
\end{eqnarray}
The solutions are of the form
\begin{eqnarray}
    \gamma \sin\tau = \cos\rho \ , \qquad \gamma \in\mbR
\end{eqnarray}
where $ \gamma $ labels the integral curves in the bulk, see Figure \ref{fig:ads2_integral_curves}. In the region corresponding to the future lightcone for $0<\gamma<1$, we can now decouple the PDEs to get
\begin{figure}[t]
    \centering
    \includegraphics[width=.75\textwidth]{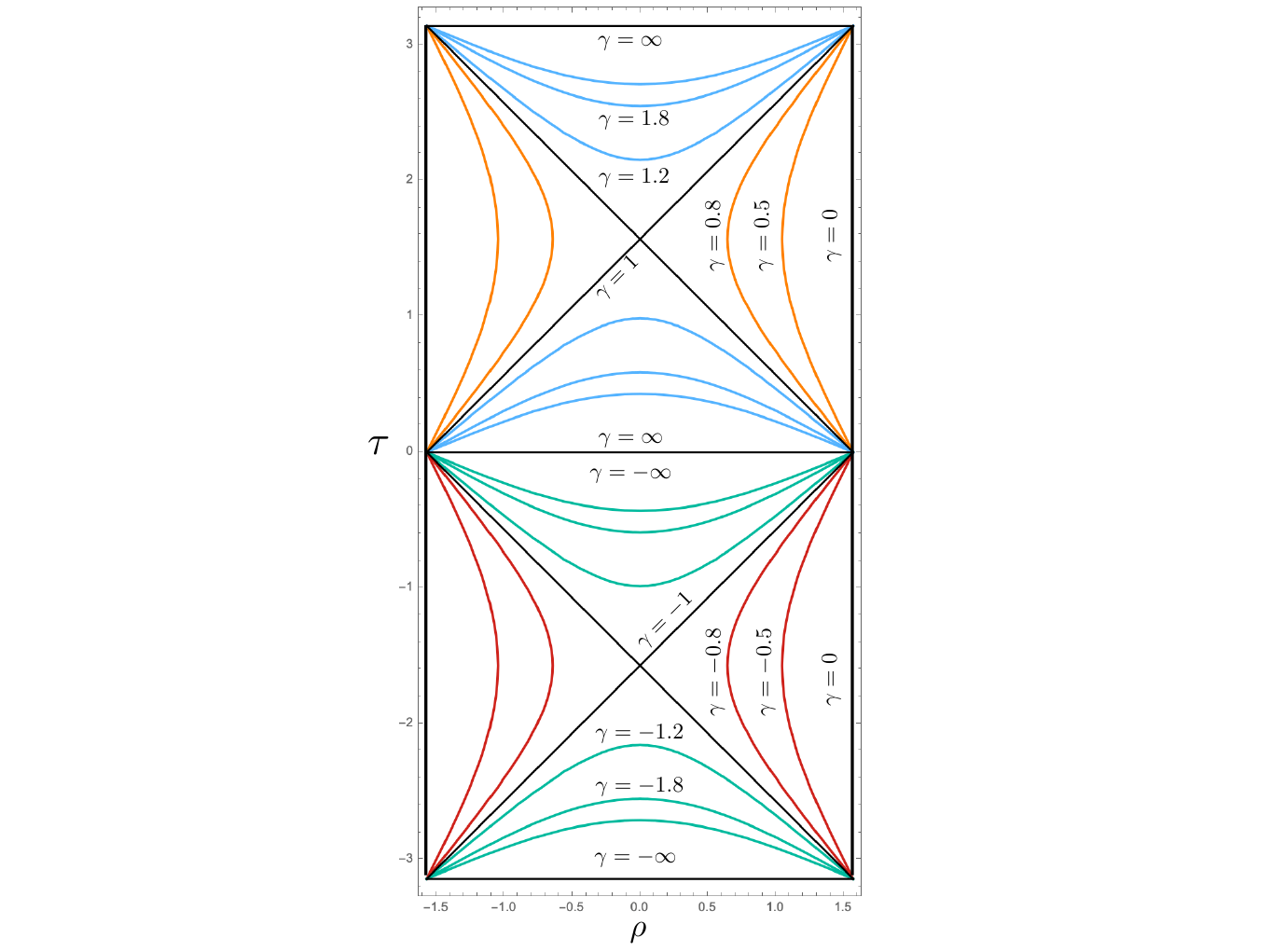}
    \caption{\small Integral curves for the modular flow for future lightcone ($0<\gamma<1$) in Poincar\'e AdS$_2$.}
    \label{fig:ads2_integral_curves}
\end{figure}
\begin{eqnarray}
    \frac{d\tau(s)}{ds} = \sin\tau\,\sqrt{1-\gamma^2 \sin^2\tau} \ .
\end{eqnarray}
The integral curves are then given by
\begin{eqnarray}
    e^{-isJ_{-1,1}} \tau = \tau(s) &=& -\tan^{-1}\lb \frac{1}{\sqrt{1-\gamma^2} \sinh(s+c_1)}\rb\ ,
     \nn \\
    c_1 &=& \sinh^{-1}\lb\frac{-1}{\sqrt{1-\gamma^2}\tan\tau(0) } \rb 
\end{eqnarray}
with the range for $-\cot^{-1}$ being restricted to $[0,\pi]$. In the limit $\gamma \to 0$, we recover the integral curve on the boundary in \eqref{ads2_future_modflow}.
In terms of the Poincar\'e coordinates, evaluating the integral curves is more straightforward. Using \eqref{isometries_poincare}
\begin{eqnarray}
    \frac{dt(s)}{ds} = t(s), \qquad \frac{dz(s)}{ds} = z(s)\ .
\end{eqnarray}
The solutions are the familiar dilatations
\begin{eqnarray}\label{poincare_dilatation}
t(s) = e^s t(0),\qquad z(s) = e^s z(0)\ .    
\end{eqnarray}

\section{Massive free fields in Poincar\'e AdS$_2$}\label{AdS2massive}

Consider the vacuum state of a massive free scalar field theory in AdS$_2$ spacetime. We expect it to be diffeomorphism invariant under the bulk AdS$_2$ isometries in Appendix \ref{app:ads_geometry}. Let us focus on the theory on the Poincar\'e patch. The equations of motion in Poincar\'e coordinates are
\begin{eqnarray}
    &&(\Box -m^{2}) \varphi(t,z) = 0 \label{eq-scalar-equation} \nn \\
    &&\left(-\partial_{t}^{2} + \partial_{z}^{2} - \frac{m^2}{z^{2}} \right) \varphi(t,z) = 0 \ .\label{eq-scalar-equation-2}
\end{eqnarray}
Consider a propagating wave solution with positive energy mode given by
\begin{eqnarray}\label{eq-pos-energy-mode}
    &&\varphi_{\omega}(t,z) = \frac{1}{\sqrt{2\omega}}  \chi_{\omega}(z)  e^{-i\omega t} \nn \\ 
    &&\chi_{\omega}(z) = \sqrt{\omega z}  J_{\nu}(\omega z) \ , \qquad \omega > 0   
\end{eqnarray}
with $\nu^{2} = m^2 + 1/4  $.
The solutions are normalized according to the Klein-Gordon inner product
\begin{eqnarray}
    \braket{ \varphi_1, \varphi_2 } = -i \int_{\Sigma}  d\Sigma \sqrt{|h|} n^{\mu}  \left[\varphi_1 \partial_{\mu} \varphi^\dagger_2 - (\partial_{\mu}\varphi_1) \varphi^\dagger_2 \right]  
\end{eqnarray}
where $\Sigma$ is a Cauchy slice, $h_{ij}$ is the induced metric on $\Sigma$ and $n^\mu$ is the (timelike) unit vector normal to $\Sigma $. In the case of Poincar\'e AdS$_2$ 
\begin{eqnarray}
\sqrt{|h|} =1/z \ , \qquad n^{\mu} = z \delta^{\mu t}\partial_{t}  \ .
\end{eqnarray} 
Therefore the inner product simplifies to 
\begin{eqnarray}
    \braket{ \varphi_1, \varphi_2 } = -i \int_{0}^{\infty}  dz \left[\varphi_1 \partial_{t} \varphi^\dagger_2 - (\partial_{t}\varphi_1) \varphi^\dagger_2 \right]\Big\vert_{t=0}  
\end{eqnarray}
where we have chosen $\Sigma$ as the $t=0$ slice.
The normalization of $\chi_{\omega}(z)$ \footnote{The Bessel function satisfies the closure equation for $\nu > -1/2$
\begin{eqnarray}
    \int_0^\infty dz\,z J_\nu(\omega z) J_\nu(\omega' z ) = \frac{1}{\omega} \delta(\omega - \omega')\ .
\end{eqnarray}}
\begin{eqnarray}
    \int_{0}^{\infty} dz  \chi_{\omega}(z)  \chi_{\omega'}(z) = \delta(\omega-\omega')  \label{eq-completeness-cond}
\end{eqnarray}
ensures that the positive energy modes in \eqref{eq-pos-energy-mode} satisfy
\begin{eqnarray}
    \braket{\varphi_{\omega},\varphi_{\omega'}} = \delta(\omega-\omega')  ,\quad\quad  \braket{\varphi^\dagger_{\omega},\varphi^\dagger_{\omega'}} = -\delta(\omega-\omega')  ,\quad\quad  \braket{\varphi_{\omega},\varphi^\dagger_{\omega'}} = 0 \ .\label{eq-ortho-modes}
\end{eqnarray}

Suppose we have a Killing vector field $V$. Since $\mL_V g_{\mu\nu} = 0$, we note that that the covariant scalar field equations of motion in \eqref{eq-scalar-equation} do not change under the transformation generated by $V$ and hence $\mL_V\varphi$ is also a solution to the equations of motion. However, this does not necessarily mean that the vacuum state is invariant. The vacuum state is defined as the state annihilated by all annihilation operators. Thus it is invariant if and only if all the positive energy modes transform among themselves. Then the negative energy modes also transform among themselves and annihilate the vacuum. More precisely, suppose that  we have a Bogoliubov transformation 
\begin{eqnarray}
    \mL_V\varphi_{\omega}(z,t) = \int_{0}^{\infty} d\omega'  \left[ \Gamma_V(\omega,\omega')  \varphi_{\omega'}(z,t)  +  \tilde{\Gamma}_V(\omega,\omega')  \varphi^\dagger_{\omega'}(z,t)\right]  
\end{eqnarray}
where $\Gamma_V(\omega,\omega')$ and $\tilde{\Gamma}_V(\omega,\omega')$ are the Bogoliubov coefficients given by
\begin{eqnarray}
    \Gamma_V(\omega,\omega') = \braket{ \mL_V\varphi_{\omega} , \varphi_{\omega'} } \ , \qquad     \tilde{\Gamma}_V(\omega,\omega') = -\braket{ \mL_V\varphi_{\omega} , \varphi^\dagger_{\omega'} }  \label{eq-bogo-KG}\ .
\end{eqnarray}
Then the vacuum state is invariant under the diffeomorphism if and only if $\tilde{\Gamma}_V(\omega,\omega') = 0$ for all  $\omega,\omega' > 0$.

We will now show this explicitly for the isometries of Poincar\'e AdS$_2$. Consider time translations generated by $H = J_{-1,0}+ J_{0,1} = -i \partial_{t}$ in \eqref{isometries_poincare}. Using \eqref{eq-pos-energy-mode}, we get
\begin{eqnarray}
    \mL_{H} \varphi_{\omega}(z,t) = -\omega  \varphi_{\omega}(z,t) \ .
\end{eqnarray}
This means that
\begin{eqnarray}
    \Gamma_{H}(\omega,\omega') = -\omega \delta(\omega-\omega') \ , \qquad \tilde{\Gamma}_{H}(\omega,\omega') = 0 \ .
\end{eqnarray}
Thus the vacuum is time translation invariant.
Now for dilatations generated by $J_{-1,1}$,
\begin{eqnarray}\label{eq-lie-D-int}
    \mL_{J_{-1,1}} \varphi_{\omega}(z,t) &=& i \lb t  \partial_{t} \varphi_{\omega}(z,t) + z \partial_{z} \varphi_{\omega}(z,t) \rb =\frac{i}{\sqrt{2\omega}}  e^{-i\omega t}  \left[ -i \omega t \chi_{\omega}(z) + z  \partial_{z} \chi_{\omega}(z)\right]  \nn \\
    &=&  \frac{i}{\sqrt{2\omega}}  e^{-i\omega t}  \left[ -i \omega t \chi_{\omega}(z) + \omega  \partial_{\omega} \chi_{\omega}(z)\right]  =  i\sqrt{\frac{\omega}{2} } \partial_{\omega} \left[ e^{-i\omega t} \chi_{\omega}(z) \right]  \nn\\
    &=& \sqrt{\omega} \partial_{\omega} \left[ \sqrt{\omega} \varphi_{\omega}(z,t) \right]  =  \frac{1}{2} \varphi_{\omega}(z,t)  +  \omega \partial_{\omega} \varphi_{\omega}(z,t) \ .
\end{eqnarray}
where we have used that $ z  \partial_{z} \chi_{\omega}(z) = \omega  \partial_{\omega} \chi_{\omega}(z)$. This implies that
\begin{eqnarray}
    \Gamma_{J_{-1,1}}(\omega,\omega') = \frac{1}{2} \delta(\omega-\omega') + \omega\partial_{\omega}\delta(\omega-\omega')  \ ,\qquad \tilde{\Gamma}_{J_{-1,1}}(\omega,\omega') = 0 \ .
\end{eqnarray}
Therefore, the vacuum is invariant under dilatations.
Finally, we consider the special conformal transformations $S = J_{-1,0}- J_{0,1}$. In this case, we will use the formula for $\tilde{\Gamma}_{S}(\omega,\omega')$ in terms of the Klein-Gordon inner product
\begin{eqnarray}
    \tilde{\Gamma}_{S}(\omega,\omega')  =  -\braket{ \mL_{S}\varphi_{\omega} , \varphi^\dagger_{\omega'}} = i  \int_{0}^{\infty}  dz \bigg[ \mL_{S}\varphi_{\omega}  \partial_{t} \varphi_{\omega'} - (\partial_{t}\mL_{S}\varphi_{\omega}) \varphi_{\omega'} \bigg]\Bigg\vert_{t=0} \ .\label{eq-Gamma-s-int}
\end{eqnarray}
We get
\begin{eqnarray}
    \mL_{S}\varphi_{\omega}(z,t) = \frac{-i}{\sqrt{2\omega}}  e^{-i\omega t}  \bigg[-i\omega(t^2 + z^2) \chi_{\omega}(z) + 2t z  \partial_{z} \chi_{\omega}(z) \bigg]  .
\end{eqnarray}
Therefore, 
\begin{eqnarray}
    \mL_{S}\varphi_{\omega}(z,t)\Big\vert_{t=0}= - \sqrt{\frac{\omega}{2}}  z^{2}  \chi_{\omega}(z)  ,
\end{eqnarray}
and
\begin{eqnarray}
    \partial_{t}\mL_{S}\varphi_{\omega}(z,t)\Big\vert_{t=0}= i \sqrt{\frac{\omega^{3}}{2}}  z^{2}  \chi_{\omega}(z)  -i \sqrt{\frac{2}{\omega}}  z \partial_{z} \chi_{\omega}(z)  .
\end{eqnarray}
With these identities, we simplify \eqref{eq-Gamma-s-int} as
\begin{eqnarray}
    \tilde{\Gamma}_{S}(\omega,\omega')  =   \int_{0}^{\infty}  dz \bigg[ \sqrt{\frac{\omega}{\omega'}}  \frac{\omega-\omega'}{2} z^{2}\chi_{\omega'}(Z)\chi_{\omega}(z)  -  \frac{1}{\sqrt{\omega\omega'}} \chi_{\omega'}(z)  z  \partial_{z} \chi_{\omega}(z) \bigg] .\label{eq-Gamma-s-2}
\end{eqnarray}
To further simplify this equation, we observe that the equation of motion in \eqref{eq-scalar-equation-2} implies
\begin{eqnarray}
    z^{2}\partial_{z}^{2}\chi_{\omega}(z)  - m^{2} \chi_{\omega}(z)  + \omega^{2}z^{2} \chi_{\omega}(z) = 0  .
\end{eqnarray}
With this equation, we write \eqref{eq-Gamma-s-2} as
\begin{eqnarray}
    \tilde{\Gamma}_{S}(\omega,\omega')  =   \int_{0}^{\infty}  dz \bigg[ \sqrt{\frac{\omega}{\omega'}}  \frac{\omega-\omega'}{2\omega^{2}} \chi_{\omega'}(z) (m^{2}-z^{2}\partial^{2}_{z})\chi_{\omega}(z)  -  \frac{\chi_{\omega'}(z)}{\sqrt{\omega\omega'}}   z  \partial_{z} \chi_{\omega}(z) \bigg] .
\end{eqnarray}
Now using 
\begin{eqnarray}
    z^{k}  \partial_{z}^{k}  \chi_{\omega}(z) = \omega^{k}  \partial_{\omega}^{k}  \chi_{\omega}(z)  
\end{eqnarray}
for $k \in \mathbb{N}$, we get
\begin{eqnarray}
    \tilde{\Gamma}_{S}(\omega,\omega')  &=&   \sqrt{\frac{\omega}{\omega'}}   \bigg[\frac{\omega-\omega'}{2\omega^{2}} m^{2}  -  \frac{\omega-\omega'}{2} \partial_{\omega}^{2} - \partial_{\omega} \bigg] \int_{0}^{\infty} dz \chi_{\omega'}(z)\chi_{\omega}(z)  \nn\\
    &=&  \sqrt{\frac{\omega}{\omega'}}   \bigg[\frac{\omega-\omega'}{2\omega^{2}} m^{2}  -  \frac{\omega-\omega'}{2} \partial_{\omega}^{2} - \partial_{\omega} \bigg]  \delta(\omega-\omega') 
\end{eqnarray}
where we have used the normalization condition in \eqref{eq-completeness-cond}. Now using the properties of the delta function, we get
\begin{eqnarray}
    &&\tilde{\Gamma}_{S}(\omega,\omega') = -  \sqrt{\frac{\omega}{\omega'}}   \bigg[ \frac{\omega-\omega'}{2}\partial_{\omega}^{2} + \partial_{\omega} \bigg]  \delta(\omega-\omega') \nn\\
    &=& - \frac{1}{2} \sqrt{\omega^{3}}  \partial_{\omega}^{2} \left[\frac{1}{\sqrt{\omega'}}\delta(\omega-\omega') \right]  + \frac{1}{2} \sqrt{\omega}  \partial_{\omega}^{2} \left[ \sqrt{\omega'} \delta(\omega-\omega')\right]  -  \sqrt{\omega}  \partial_{\omega} \left[\frac{1}{\sqrt{\omega'}} \delta(\omega-\omega')\right] \nn \\
    &=& - \frac{1}{2} \sqrt{\omega^{3}}  \partial_{\omega}^{2} \left[\frac{1}{\sqrt{\omega}}\delta(\omega-\omega') \right]  + \frac{1}{2} \sqrt{\omega}  \partial_{\omega}^{2} \left[ \sqrt{\omega} \delta(\omega-\omega')\right]  -  \sqrt{\omega}  \partial_{\omega} \left[\frac{1}{\sqrt{\omega}} \delta(\omega-\omega')\right] \nn\\
    &=&  -\frac{1}{2} \left[\frac{3}{4\omega} \delta(\omega-\omega') - \partial_{\omega}\delta(\omega-\omega') + \omega \partial_{\omega}^{2}\delta(\omega-\omega') \right]  -  \left[ -\frac{1}{2\omega} \delta(\omega-\omega') + \partial_{\omega}\delta(\omega-\omega') \right]  \nn\\ 
    && +  \frac{1}{2}  \left[-\frac{1}{4\omega} \delta(\omega-\omega') + \partial_{\omega}\delta(\omega-\omega') + \omega  \partial_{\omega}^{2}\delta(\omega-\omega') \right] = 0 \ .
\end{eqnarray}
Therefore, the vacuum is invariant under special conformal transformations. Thus we have shown that the vacuum of a massive free scalar field theory in Poincar\'e AdS$_2$ spacetime is invariant under its $PSL(2,\mbR)$ isometries.

\section{Half-sided modular inclusion and translation}\label{App:Hsmi}

\begin{definition}\label{HSMIdef}
   We say we have a future (past) half-sided modular inclusion, i.e. HSMI$+$ (HSMI$-$), if we have a proper inclusion of von Neumann algebras $\mA\subset \mR$ with a common cyclic and separating vector $\ket{\Omega}$ such that 
\begin{eqnarray}
\forall t>0 (t<0):\qquad \Delta_\mR^{-it}\mA\Delta_\mR^{it}\subset \mA
\end{eqnarray}
where $\Delta_\mR$ is the modular operator of $\mR$ in the state $\ket{\Omega}$. We call $\mA$ a modular future (past) subalgebra.
\end{definition}
Note that the $\Delta^{-it}$ convention is opposite to mathematics literature. Under this choice, the modular time is in the same direction as the physical time.  

\begin{theorem}[HSMI Theorem]\label{thm:hsmi}
    Suppose $\mA \subset \mR$ is HSMI$\pm$. Then,  
    \begin{eqnarray}
    G = \frac{1}{2\pi} \left( \log\Delta_{\mA}-\log\Delta_{\mR}\right)
    \end{eqnarray}
    is a positive operator and satisfies the following commutation relation
    \begin{eqnarray}
    \frac{1}{2\pi}[\log\Delta_{\mR}, G] = \pm i G\ .
    \end{eqnarray}
    Moreover, the unitary $U(a) = e^{iaG}$ satisfies following properties:
    \begin{eqnarray}
        && \Delta_{\mR}^{-it} U(a) \Delta_{\mR}^{it} = \Delta_{\mA}^{-it} U(a) \Delta_{\mA}^{it} = U(e^{\pm 2\pi t}a) \, , \label{eq-hsmi-1}\\
        && J_{\mR} U(a) J_{\mR} = J_{\mA} U(a) J_{\mA} = U(-a) \, , \label{eq-hsmi-2}\\
        && \mA = U(\pm 1) \mR U(\mp 1) \label{eq-hsmi-3}\, .
    \end{eqnarray}
\end{theorem}
\begin{proof}
The proofs are standard and can be found in \cite{borchers2000revolutionizing,araki2005extension,wiesbrock1993symmetries}.
\end{proof}
\begin{corollary}
    If $\mA\subset\mR$ is HSMI$\pm$, then
    \begin{align}
        J_{\mR}J_{\mA} \, =& \, U(\mp 2) \, ,\label{eq-hsmi-4}\\
        \Delta_{\mR}^{it}\Delta_{\mA}^{-it} \, =& \, U(\mp 1 \pm e^{\mp 2\pi t}) \, .\label{eq-hsmi-5}
    \end{align}
\end{corollary}

\begin{corollary} \label{kwing-corollary}
    If $\mA \subset \mR$ is HSMI$\pm$, then $\mR' \subset \mA'$ is HSMI$\mp$ and $\Delta_{\mA}^{-it}\mR\Delta_{\mA}^{it} \subset \mR$ for all $\mp t > 0$.
\end{corollary}

The HSMI theorem above can be written in the language of ergodic hierarchy as 
\begin{corollary}
  Consider a von Neumann algebra in the standard GNS representation $\{\mH,\ket{\Omega},\mR\}$. If we have a modular future subalgebra $\mA\subset \mR$, then the positive operator \begin{eqnarray}
      G=\frac{1}{2\pi}(\log\Delta_\mA-\log\Delta_\mR)\geq 0
  \end{eqnarray}
  generates a unitary flow $U(s)=e^{is G}$ such that: 
  \begin{enumerate}
      \item {\bf Maximal Chaos:} The flow by $U(s)$ corresponds to a growing Anosov mode with a maximal Lyapunov exponent $\lambda=2\pi$, i.e. 
  \begin{eqnarray}
      \forall s, t \in \mathbb{R}:\qquad \Delta_\mR^{-it}e^{i s G}\Delta_\mR^{it}=\Delta_\mA^{-it}e^{i s G}\Delta_\mA^{it}=e^{i  e^{2\pi t}sG}\ .
  \end{eqnarray}
  \item {\bf Future algebra with $G>0$:} The algebra $\mR$ is a future algebra with respect to the dynamics $U(s)$, i.e.
  \begin{eqnarray}
      \forall s>0:\qquad U(s) \mR U(s)^\dagger\subset \mR\ .
  \end{eqnarray}
  \item  {\bf Quantum detailed balance:} We also have 
  \begin{eqnarray}\label{detailedassu}
       \forall s\in \mathbb{R}:\qquad J_\mR e^{i s G}J_\mR=J_\mA e^{i s G}J_\mA=e^{-i sG}\ .
  \end{eqnarray}
  \end{enumerate}
\end{corollary}

\begin{theorem}[Half-sided translations]\label{thm:HST}
    Consider a dynamical flow $U(t)=e^{i Gt}$ that leaves the vacuum invariant $U(t)\ket{\Omega}=\ket{\Omega}$. If $\mA$ is a future algebra ($\mA_t\subset \mA$ for all $t>0$) then the following are equivalent:
    \begin{enumerate}
        \item The generator is positive: $G\geq 0$.
        \item The quantum detailed balance condition $\Delta^{1/2}U(t)=U(-t)\Delta^{1/2}$
        \item The flow satisfies Borchers' relations $\Delta^{-is}U(t)\Delta^{is}=U(e^{2\pi s}t)$.
        \item The algebra $\mA_t$ is a modular future algebra of $\mA$.
    \end{enumerate}
\end{theorem}
\begin{proof}
See \cite{ouseph2024local} for the proof.
\end{proof}

\section{Normal and standard inclusions of von Neumann algebras}\label{app:StandardNormal}
Consider an inclusion of von Neumann algebras $\mA_1\subset \mA_2'$. 
Assume that the relative commutant $\mL\equiv(\mA_1'\cap \mA_2')''$ is non-trivial (inclusion is not singular). We will be interested in the GNS Hilbert space of the state on $\mL$ that we denote by $\mH=\overline{\mL\ket{\Omega}}$. Such a triple $\{\mA_1,\mA_2',\ket{\Omega}\}$ is called a {\it pseudo-standard inclusion}. In a local theory, if we associate von Neumann algebras $\mathcal{X}$ to subsets $X$. Every pair of non-overlapping causally-disjoint subsets $A_1$ and $A_2$ give an inclusion $\mA_1\subset \mA_2'$. To simplify the picture, we often assume Haag's duality: the commutant $\mathcal{X}'$ is the same as the algebra of the causal complementary subset $X'$. Therefore, $\mL'$ is the algebra of $A_1\cup A_2$. 

We are interested in studying the automorphisms of a pseudo-standard inclusion of von Neumann algebras $\{\mA_1,\mA_2',\ket{\Omega}\}$. These are the maps that preserve $\mA_1$, $\mA_2'$ and the state $\ket{\Omega}$. Assuming Haag's duality, such a map would preserve $\mL$, as well. Therefore, it will have a representation in the Hilbert space $\mH$ such that 
\begin{eqnarray}
    U^\dagger\mA_1 U=\mA_1\ , \qquad U^\dagger\mA_2' U=\mA_2'\ , \qquad J_\mL U=U J_\mL\ .
\end{eqnarray}
In particular, such a unitary preserves the natural cone $P^\sharp(\mL)\ket{\Omega}$, which is the subspace of $\mH$ of the states symmetric under $J_\mL$. 

For any pseudo-standard inclusion, we can define intermediate algebras
\begin{eqnarray}
&&\mA_1\subset \underline{\mN}\subset \overline{\mN}\subset \mA_2'\nn\\
    &&\underline{\mN}\equiv (\mA_1\cup \mA_1^c)'' \nn\\
    &&\overline{\mN}\equiv (\mA_2'\cap (\mA_2')^c)''\nn\\
    &&\mathcal{X}^c\equiv J_\mL \mathcal{X} J_\mL\ .
\end{eqnarray}
It is clear that the automorphisms $U$ will also preserve $\overline{\mN}$ and $\underline{\mN}$. We warn the reader not to confuse the notation $\mathcal{X}^c$ defined above with the notation defined in (\ref{Haagdual}).

We say we have a {\it standard inclusion} $\{\mA_1,\mA_2',\ket{\Omega}\}$ if the vector $\ket{\Omega}$ is cyclic and separating with respect to all three algebras $\mA_1,\mA_2',\mL$ \cite{doplicher1984standard}. 
In the context of the local observable algebras of QFT, consider two non-overlapping causally-disjoint regions $A_1$ and $A_2$ and their corresponding algebras $\mA_1$ and $\mA_2$. As long as $A$ has finite volume (not a null segment), by Reeh-Schlieder theorem, the inclusion $\mA_1\subset \mA_2'$ is standard. Corollary 1.3 of \cite{doplicher1984standard} proves that if $\mA_1$ and $\mA_2'$ are both factors then $\mA_1\subset \mA_2'$ is a standard inclusion if and only if $\mA_1$ and $\mL$ are both properly infinite (Type I$_\infty$, II$_\infty$ or Type III). 

Two special classes of standard inclusions relevant to local physics are pseudo-normal inclusions and split inclusions.
{\it Pseudonormal} inclusions are standard inclusions that satisfy $\underline{\mN}=\overline{\mN}=\mN$.
Split inclusions are standard inclusions such that there is a Type I algebra $\mN$ such that $\mA_1\subset \mN\subset \mA_2'$. An inclusion $\mA_1\subset \mA'_2$ is {\it normal} if the relative commutant of the inclusion $\mL\subset \mA'_2$ is equivalent to $\mA_1$.

The local algebras of QFT, often give pseudonormal and standard inclusions.
It follows from the definition above that $\mA_1^c\subset \mA_2'$. If $\mA_1\subset \mA_2'$ is standard the inclusion $\mA_1^c\subset \mA_2'$ is also standard. Furthermore, if the inclusion $\mA_1\subset \mA_2'$ is standard and splits, and both $\mA_1$ and $\mA_2'$ are factors, then the inclusion $\mA_1^c\subset \mA_2'$ is normal (Proposition 4.1 of \cite{doplicher1984standard}). Since $\mA_1^{cc}=\mA_1$ it is clear from the definition above that $\overline{\mN^c}=\overline{\mN}$ and $\underline{\mN^c}=\underline{\mN}$, as a result $\mA_1\subset \mA_2'$ is also pseudo-normal. Corollary 7.5 of \cite{doplicher1984standard} proves that in a standard inclusion, $\mN$ is always properly infinite (Type I$_\infty$, II$_\infty$ or Type III). Proposition 7.6 of \cite{doplicher1984standard} clarifies that when $\mA_1$ and $\mA_2'$ are both hyperfinite and $\mA_1\subset \mA_2'$ is a pseudonormal and standard inclusion the algebra $\mN$ is also hyperfinite.

\bibliographystyle{ieeetr}
\bibliography{main}
\end{document}